\documentclass[12pt]{article}
\usepackage{amssymb}
\usepackage{amsmath}
\usepackage{amsthm}
\usepackage{epsfig}
\usepackage{mathrsfs}
\usepackage{graphicx}
\usepackage{ textcomp }
\usepackage{xcolor}
\usepackage{hyperref}
\usepackage{bbold}
\usepackage{graphicx}
\usepackage{caption}
\usepackage{subcaption}
\usepackage{authblk}
\usepackage{fullpage}
\usepackage{enumerate}

\def\R{\mathbb{R}}

\def\la{\langle}
\def\ra{\rangle}

\def\to{\rightarrow}

\def\pa{\partial}

\def\eps{\varepsilon}

\newtheorem{theorem}{Theorem}[section]
\newtheorem{lemma}[theorem]{Lemma}
\newtheorem{proposition}[theorem]{Proposition}
\newtheorem{definition}[theorem]{Definition}

\newtheorem{remark}[theorem]{Remark}
\newtheorem{rk&ex}[theorem]{Remarks \& Examples}
\newtheorem{corollary}[theorem]{Corollary}

\newcommand{\beqar}{\begin{eqnarray*}}
\newcommand{\eeqar}{\end{eqnarray*}}
\newcommand{\beqarl}{\begin{eqnarray}}
\newcommand{\eeqarl}{\end{eqnarray}}
\newcommand{\lp}{\left(}
\newcommand{\rp}{\right)}

\newcommand{\be}{\begin{equation}}
\newcommand{\ee}{\end{equation}}

\newcommand{\nn}{\nonumber}

\newcommand{\nvec}{\mathbf{n}}
\newcommand{\vezero}{\mathbf{e_{1}}}
\newcommand{\uu}{\mathbf{u}}
\newcommand{\vv}{\mathbf{v}}

\newcommand{\Id}{\mathrm{Id}}
\newcommand{\PD}{\mbox{PD}}
\newcommand{\tr}{\textnormal{Tr}}
\newcommand{\ud}{\mathrm{d}}

\newcommand{\unitq}{{\mathbb{H}_1}}
\newcommand{\Ima}{\mbox{Im}}
\newcommand{\Real}{\mbox{Re}}
\newcommand{\q}{\mathbf{q}}
\newcommand{\qk}{\mathbf{q}_k}
\newcommand{\bqk}{\bar{\mathbf{q}}_k}

\newcommand{\p}{\mathbf{p}}

\newcommand{\rvec}{\mathbf{r}}

\allowdisplaybreaks[3]

\title{Quaternions in collective dynamics}

\author[(1)]{Pierre Degond}
\author[(2)]{Amic Frouvelle}
\author[(3)]{Sara Merino-Aceituno}
\author[(4)]{Ariane Trescases}

\affil[(1)(3)]{Department of Mathematics, Imperial College London, South Kensington Campus\\
London, SW7 2AZ\\
United Kingdom,}
\affil[(1)]{pdegond@imperial.ac.uk}
\affil[(3)]{s.merino-aceituno@imperial.ac.uk}

\affil[(2)]{CEREMADE,  UMR CNRS 7534, Universit\'e de Paris-Dauphine, PSL Research University\\
Place du Mar\'echal De Lattre De Tassigny\\
PARIS, 75775 CEDEX 16, France\\
frouvelle@ceremade.dauphine.fr
}

\affil[(4)]{Department of Pure Mathematics and Mathematical Statistics, University of Cambridge\\
Wilberforce Road,
Cambridge, CB3 0WA\\
atrescases@maths.cam.ac.uk}

\begin{document}

\maketitle

\begin{abstract}
We introduce a model of multi-agent dynamics for self-organised motion; individuals travel at a constant speed while trying to adopt the averaged body attitude of their neighbours. The body attitudes are represented through unitary quaternions. We prove the correspondance with the model presented in Ref. \cite{bodyattitude} where the body attitudes are represented by rotation matrices. Differently from this previous work, the individual based model (IBM) introduced here is based on nematic (rather than polar) alignment. From the IBM, the kinetic and macroscopic equations are derived. The benefit of this approach, in contrast to Ref. \cite{bodyattitude}, is twofold: firstly, it allows for a better understanding of the macroscopic equations obtained and, secondly, these equations are prone to numerical studies, which is key for applications. 
\end{abstract}

\paragraph{Keywords:} Body attitude coordination; quaternions; collective motion; nematic alignment; Q-tensor; Vicsek model; Generalized Collision Invariant; dry active matter; self-organised hydrodynamics.

\paragraph{AMS Subject Classification:} 35Q92, 82C22, 82C70, 92D50

\newpage

\section{Introduction}

In this paper we consider collective dynamics where individuals are described by their location in space and position of their body (body attitude). The body attitude  is determined by a frame, i.e, three orthonormal vectors such that one vector indicates the direction of motion of the agent and the other two represent the relative position of the body around this direction. For this reason, the body frame of a given individual can be characterised by the rotation of a fixed reference frame.  This rotation (and hence, the body attitude) will be represented here by elements of the group of  unitary quaternions, denoted as $\unitq$ (see Fig. \ref{Fig:quaternion_explained}).
There exist multiple ways of describing rotations in $\R^3$, as we will see in Sec. \ref{sec:compare_micro}. Here, we choose the quaternionic representation as it is the one mostly employed in numerical simulations given their efficiency in terms of memory usage and complexity of operations \cite{salamin1979application}. This is key to apply the results of the present work.

 \begin{figure}
 \centering
\includegraphics[scale=0.5]{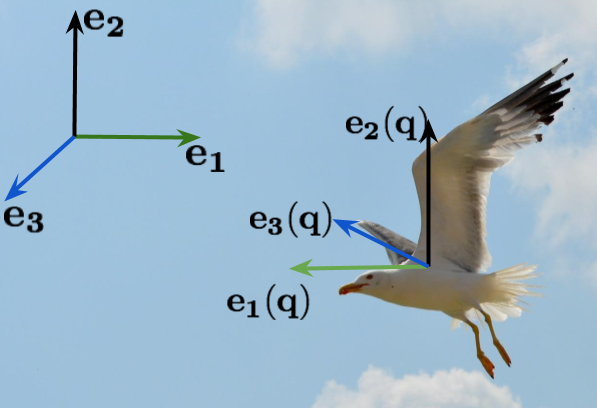}  
 \caption{The reference frame given by the orthonormal basis  $\left\{\vezero, \mathbf{e_2}, \mathbf{e_3}\right\}$ is rotated according to the unitary quaternion $\q$ and gives a new frame $\{\vezero(\q), \mathbf{e_2}(\q), \mathbf{e_3}(\q)\}$ describing the body attitude of the bird (where $\mathbf{e_i}(\q)$ denotes the rotation of the vector $\mathbf{e_i}$ by $\q\in \unitq$). The direction of motion is given by the vector $\vezero(\q)$ while the pair $\lp\mathbf{e_2}(\q), \mathbf{e_3}(\q)\rp$ gives the relative position of the body around this direction. }
  \label{Fig:quaternion_explained}
 \end{figure}

Agents move at a constant speed while attempting to coordinate their body attitude with those of near neighbours. Here we present an Individual Based Model (particle model) for body attitude coordination. We derive the corresponding macroscopic equations from the associated mean-field equation, which we refer to as the Self-Organized Hydrodynamics  based on Quaternions (SOHQ), by reference to the Self-Organized Hydrodynamics (SOH) derived from the Vicsek dynamics (see Ref.~\cite{degond2008continuum} and discussion below). Our model is inspired from the one in Ref. \cite{bodyattitude} where the body attitude is represented by elements of the rotation group $SO(3)$. The macroscopic equations obtained in Ref.~\cite{bodyattitude} present various drawbacks. Firstly, some of the terms in the equations do not have a clear interpretation. Secondly, and  most importantly, the macroscopic equations are impractical for numerical simulations due to their complexity, specially since some terms are defined implicitly (see Eq.~\eqref{def:D_x}). Moreover, the matrix representation is computationally inefficient: storing a rotation matrix requires nine entries of memory (while quaternions only require 4) and orthonormalising matrices is computationally expensive (while for quaternions this corresponds to normalising a 4-dimensional vector). Our objective is to sort out these problems that impeed the applications of the results obtained in Ref. \cite{bodyattitude} by considering a quaternionic representation. 

In contrast with the use of rotation matrices in  Ref. \cite{bodyattitude}, the use of the quaternion representation makes the modeling more difficult at the individual based level; firstly, it is not clear how to define a mean body attitude based on quaternions and, secondly, we need to consider nematic alignment rather than polar alignment. However, 
 the macroscopic equations obtained are easier to interpret than in Ref. \cite{bodyattitude} and  provide the right framework to carry out numerical simulations. The \textbf{main contributions} of the present paper are the derivation of the macroscopic equations; finding the right modelling at the individual based level; and proving the equivalence of the models and results obtained here with the ones in Ref. \cite{bodyattitude} for the rotation-matrix representation.

\medskip

There exist already a variety of models for collective behaviour depending on the type of interaction between agents. In the case of body attitude coordination, apart from Ref. \cite{bodyattitude},  other models has been proposed, see Ref. \cite{sarlette2009autonomous} and references therein.  This has applications in the study of collective motion of biological systems such as sperm dynamics; animals such as birds and fish; and it is a stepping stone to model more complex agents formed by articulated bodies (corp\-ora) \cite{constantin2010onsager,constantin2010high}. In the rest of the section we present related results in the literature and the structure of the document. 

\bigskip

The literature on collective behaviour is extensive. Such systems are ubiquitous in nature: fish schools, flocks of birds, herds \cite{buhl2006disorder,cavagna2010scale,parrish1997animal}; bacteria \cite{ben2000cooperative,zhang2010collective}; human walking behaviour \cite{helbing2007dynamics} are some examples.
The main benefit to study collective motion and self-organisation is to gain understanding in their emergent properties: local interactions between a large number of agents give rise to large scale structures (see the review in Ref.~\cite{vicsek2012collective}). Given the large number of agents, a statistical description of the system is more pertinent than an agent-based one. With this in mind mean-field limits are devised when the number of agents tend to infinity. From  them macroscopic equations can be obtained using hydrodynamic limit techniques (as we explain below).

\medskip

The body attitude coordination model presented here and the one in Ref. \cite{bodyattitude} are inspired from the Vicsek model. The Vicsek model is a particular type of model for self-propelled particles \cite{aldana2003phase,couzin2002collective,gregoire2004onset,vicsek1995novel} where agents travel at a constant speed while attempting to align their direction of motion with their neighbours.
  Other refinements and adaptations of the Vicsek model (at the particle level) or the SOH model (at the continuum level) have been proposed in the literature, we just mention a couple as examples: in Ref.~\cite{cavagna2014flocking} an individual-based model is proposed to better describe collective motion of turning birds; in Ref.~\cite{degond2015multi} agents are considered to have the shape of discs and volume exclusion is included in the dynamics. 
 
 One key difference in the modelling with respect to Ref.~\cite{degond2008continuum} is that we consider nematic alignment rather than polar alignment: given $\q\in \unitq$, $\q$ and $-\q$ represent the same rotation. Collective dynamics based on nematic alignment is not, though, new, see for example Refs. \cite{degond2015continuum,degond2015multi} and references therein. Nematic alignment also appears extensively in the literature of liquid crystals and colloids, like suspensions of polymers, see Ref.~\cite{degond2015continuum} and  the reference book \cite{doi1988theory}.

\medskip

Our results are inspired by the Self-Organized Hydrodynamics (SOH)  model (the continuum version of the Vicsek model) presented in Ref. \cite{degond2008continuum}, where we have substituted velocity alignment by body attitude coordination. The macroscopic equations are obtained from the mean-field limit equation, which takes the form of a Fokker-Planck equation.

To obtain the macroscopic equations, the authors in Ref.~\cite{degond2008continuum} use the well-known tools of hydrodynamic limits, first developed in the framework of the Boltzmann equation for rarefied gases \cite{cercignani2013mathematical,degond2004macroscopic,sone2012kinetic}. Since its first appearance, hydrodynamics limits have been used in other different contexts, including traffic flow modeling \cite{aw2002derivation,helbing2001traffic} and supply chain research \cite{armbruster2006model,degond2007stochastic}. However, in Ref.~\cite{degond2008continuum} a methodological breakthrough is introduced: the Generalized Collision Invariant (GCI), which will be key for the present study (Sec. \ref{sec:GCI}). Typically to obtain the macroscopic equations we require as many conserved quantities in the kinetic equation as the dimension of the equilibria (see again Ref.~\cite{vicsek2012collective}). In the mean-field limit of the Vicsek model this requirement is not fulfilled and the GCI is used to sort out this problem. For other cases where the GCI concept has been used see Refs.~\cite{degond2014hydrodynamics,degond2014macroscopic,bodyattitude,degond2012hydrodynamics,degond2014evolution,degond2015self,frouvelle2012continuum}.

\bigskip

After this introduction, we discuss the main results in Sec. \ref{sec:discussion_results}. In Section~\ref{sec:modeling} we explain the derivation of the Individual Based Model for body coordination dynamics and show its equivalence with the model in Ref.~\cite{bodyattitude} in Sec. \ref{sec:compare_micro}. Then we give the corresponding (formal) mean-field limit for the evolution of the empirical measure when the number of agents goes to infinity in Sec. \ref{sec:mean_field_limit}.

The following part concerns the derivation of the macroscopic equations (Theorem~\ref{th:macro_limit}) for the macroscopic density of the particles~$\rho=\rho(t,x)$ and the quaternion of the mean body attitude~$\bar \q = \bar \q(t,x)$. To obtain these equations we first study the rescaled mean-field equation (Eq. \eqref{eq:f_eps} in Section~\ref{sec:scaling}), which is, at leading order, a Fokker-Planck equation. We determine its equilibria (Eq. \eqref{eq:equilibria}). In Section~\ref{sec:GCI} we obtain the Generalized Collision Invariants (Prop.~\ref{prop:non constant GCI}), which are the main tool  to derive the macroscopic equations in Section~\ref{sec:macro_limit}. Finally, in Sec. \ref{sec:comparison} we prove the equivalence of our equations and results with the ones obtained in Ref.~\cite{bodyattitude}.

\section{Discussion of the main results}
\label{sec:discussion_results}

\subsection{Preliminary on quaternions}
Some basic notions on quaternions are necessary to understand the main results of this paper. We introduce them here. The set of quaternions $\mathbb{H}$ is a field which forms a four-dimensional algebra on $\R$ and whose elements are of the form:
$$\p = p_0+ p_1\vec{\imath}+p_2\vec{\jmath}+p_3\vec{k},$$
with $p_0,p_1,p_2,p_3\in \R$ and $\vec{\imath}, \vec{\jmath}, \vec{k}$, the fundamental quaternion units, satisfy: $\vec{\imath}^2=\vec{\jmath}^2=\vec{k}^2=\vec{\imath}\vec{\jmath}\vec{k}=-1$. From this, one can check that non-zero quaternions form a non-commutative group, particularly, it holds
$$\vec{\imath}\vec{\jmath}=-\vec{\imath}\vec{\jmath}=\vec{k}, \quad \vec{\jmath}\vec{k}=-\vec{k}\vec{\jmath}=\vec{\imath},\quad
\vec{k}\vec{\imath}=-\vec{\imath}\vec{k}=\vec{\jmath}.$$
The zeroth element $p_0=\Real(\p)$ is called the real part of $\p$ and the first to third elements form the imaginary part $p_1\vec{\imath}+p_2\vec{\jmath}+p_3\vec{k}=\Ima(\p)$. 

The conjugate of $\p\in\mathbb{H}$ is defined as $\p^*=p_0- p_1\vec{\imath}-p_2\vec{\jmath}-p_3\vec{k}$. The inner product corresponds to:
\be \label{eq:inner product quaternions}
\p\cdot \p'=p_0p_0'+p_1p_1'+p_2p_2'+p_3p_3'= \Real(\p'\p^*),
\ee
which generates the norm $|\p|^2= \Real(\p\p^*)$. Unitary quaternions are a subgroup of $\mathbb{H}$ defined as
 $$\unitq=\{\q\in\mathbb{H}\, \mbox{ such that } |\q|=1\} \subset\mathbb{H}.$$
Notice that $\unitq$ can be parametrized as the 3-dimensional sphere $\mathbb{S}^3$ (see proof of Prop. \ref{prop:volume_element}). Unitary quaternions can be represented as follows 
\be \label{eq:exponential form unitary quaternion}
\q = e^{\frac{\theta}{2}(n_1 \vec{\imath}+n_2\vec{\jmath}+n_3\vec{k})}=\cos \frac{\theta}{2} + \sin\frac{\theta}{2}\lp n_1 \vec{\imath}+n_2\vec{\jmath}+n_3\vec{k} \rp,
\ee
where $\nvec:=(n_1,n_2,n_3)$ is a unitary vector in $\R^3$ and $\theta\in[0,2\pi]$. With these notations $\q \in \unitq$ represents a rotation in $\R^3$ around the axis given by $\nvec$ and of angle $\theta$, anti-clockwise. Specifically, for any vector $\mathbf{v}\in \R^3$,  the corresponding rotated vector $\mathbf{\bar v}\in \R^3$ is obtained as follows (see remark below):
\be  \label{eq:rotation by quaternion}
\mathbf{\bar v} := \Ima(\q \mathbf{v}\q^*).
\ee
The quaternions $\q$ and $-\q$ represent the same rotation. The product of unitary quaternions corresponds to the composition of rotations.
\begin{remark}[Identification between purely imaginary quaternions and vectors in $\R^3$]
\label{rem:abuse of notation}
Notice that in Eq. \eqref{eq:rotation by quaternion} we abuse notation: the product $\q \mathbf{v}\q^*$ must be understood in quaternion sense (therefore we consider $\mathbf{v}=(v_1,v_2,v_3)$ as a quaternion which is purely imaginary, i.e., $\mathbf{v}= v_1\vec{\imath}+v_2\vec{\jmath}+v_3\vec{k}$). Conversely $\mathbf{\bar v}$ is understood as a vector in $\R^3$ rather than a purely imaginary quaternion. This abuse of notation where we identify vectors in $\R^3$ with purely imaginary quaternions (and the converse) will be used thorough the text. By the context, it will be clear the sense of the interpretation.\\
We will also use in general $\q$ to denote a unitary quaternion and $\p$ to denote an arbitrary quaternion. 
\end{remark}

\subsection{Self-Organized Hydrodynamics based on Quaternions (SOHQ)}\label{sec:SOHQ}

In Sec. \ref{sec:modeling} we introduce an individual based model for collective dynamics where individuals are described by their location in space and position of their body (body attitude). Individuals move at a constant speed while trying to adopt the same body attitude, up to some noise, see Eqs. \eqref{eq:particleX}--\eqref{eq:particleQ}. The body attitude  is given  by three orthonormal vectors where one of the vector indicates the direction of motion and the other two represent the relative position of the body around this direction. In this manner, the body frame of a given individual is characterised by the rotation of a fixed reference frame.  This rotation will be represented here by elements of the group of  unitary quaternions, denoted as $\unitq$ (see Fig. \ref{Fig:quaternion_explained}). The main result of this paper is Th. \ref{th:macro_limit}, which gives the macroscopic equations for these dynamics, i.e., the time-evolution equations for the macroscopic mass of agents $\rho=\rho(t,x)$ and the mean quaternion $\bar \q=\bar\q(t,x)$, which corresponds, as explained, to the mean body attitude. Here $t\ge0$ is the time and $x\in\R^3$ denotes a point of the physical space. We will refer to this system as the Self-Organised Hydrodynamics based on Quaternions (SOHQ). 

To discuss this result we first introduce the (right) relative differential operator on $\unitq$: for a function $\q=\q(t,x)$ where $\q(t,x)\in\unitq$ and for $\pa\in\{\pa_t,\pa_{x_1},\pa_{x_2},\pa_{x_3}\}$, let
\begin{eqnarray}
\pa_{\text{rel}} \q := (\pa \q) \q^\ast,\, \Big( = \Ima( (\pa \q) \q^\ast)\Big),
\end{eqnarray}
where $\pa \q$ belongs to the orthogonal space of $\q$, and the product has to be understood in the sense of quaternions. Notice that, effectively,  $\pa_{\text{rel}} \q$ is a purely imaginary quaternion, since $\Real((\pa\q)\q^*)=\q \cdot \pa\q=0$ by Eq. \eqref{eq:inner product quaternions}, and it can be identified with a vector in $\R^3$ (recall Rem. \ref{rem:abuse of notation}). 

With this notation the SOHQ corresponds to:
\begin{eqnarray}
\label{sys:macro1}&&\pa_t \rho +\nabla_x \cdot (c_1  \vezero ( \bar \q) \rho) = 0,\\
\nn &&\rho \lp \partial_t \bar \q +  c_2 (\vezero(\bar\q) \cdot \nabla_x) \bar \q\rp + c_3\left[\vezero(\bar\q) \times\nabla_x \rho \right] \bar \q  \\
\label{sys:macro3}&& \qquad \quad +c_4\rho  \left[  \nabla_{x,\text{rel}}\bar \q \,\vezero(\bar \q) +  (\nabla_{x,\text{rel}}\cdot \bar \q) \vezero(\bar\q) \right]\bar \q =0,
\end{eqnarray}
where $\vezero$ is a vector in $\R^3$ and $\vezero(\bar\q)$ denotes the rotation of $\vezero$ by the quaternion $\bar \q$, that is,
\begin{eqnarray}
\label{sys:macro4}\vezero(\bar \q) &=& \Ima(\bar \q\vezero \bar \q^*) ;
\end{eqnarray}
and where we used the (right) relative space differential operators
\begin{eqnarray} \label{eq:def_gradient_divergence_rel}
\nabla_{x,\textnormal{rel}} \bar\q =  (\pa_{x_i,\text{rel}}\bar\q)_{i=1,2,3} =((\partial_{x_i} \bar \q) \bar \q^\ast)_{i=1,2,3} \in (\R^3)^3\subset \mathbb{H}^3,\\
\label{eq:def_gradient_divergence_rel2} \nabla_{x,\textnormal{rel}} \cdot \bar \q =\sum_{i=1,2,3} (\partial_{x_i,\text{rel}} \bar \q)_i =\sum_{i=1,2,3} ((\partial_{x_i}\bar \q) \bar\q^\ast)_i\in\R,
\end{eqnarray}
where $((\partial_{x_i}\bar \q) \bar\q^\ast)_i$ indicates the $i$-th component of $(\partial_{x_i}\bar \q) \bar\q^\ast$. In Eq. \eqref{sys:macro4} and in the last three terms of Eq. \eqref{sys:macro3} we use the abuse of notation explained in Rem. \ref{rem:abuse of notation}.
The matrix product in the fourth term of Eq. \eqref{sys:macro3} has to be understood as a matrix product, giving rise to a scalar product in $\mathbb{H}$:
\begin{eqnarray}\label{def:Drel_times_e1}
\nabla_{x,\textnormal{rel}}\bar \q \, \vezero(\bar \q) = ((\partial_{x_i,\textnormal{rel}}\bar \q)\cdot \vezero(\bar \q))_{i=1,2,3} \,.
\end{eqnarray}

In Eqs. \eqref{sys:macro1}--\eqref{sys:macro3}, $c_1,c_2,c_3$ and $c_4$ are explicit constants (given in Th. \ref{th:macro_limit}) that depend on the parameters of the model, namely, the rate of coordination and the level of the noise. The constants $c_2$, $c_3$ and $c_4$ depend on the Generalised Collision Invariant (see Introduction and Sec. \ref{sec:GCI}). 
Interestingly, $c_1$ had a special meaning as a `(polar) order parameter' in Refs. \cite{bodyattitude,degond2008continuum} (see Rem. \ref{rem:constant_c1}). Here it has the same meaning, but as a `nematic' order parameter.

Eq. \eqref{sys:macro1} gives the continuity equation for the mass $\rho$ and ensures mass conservation. The convection velocity is given by $c_1\vezero(\bar\q)$ where the direction is given by $\vezero(\bar \q)$, a unitary vector (since $\vezero$ is unitary), and the speed is $c_1$. Notice that the convection term is quadratic in $\bar\q$. This is a new structure with respect to Refs \cite{bodyattitude,degond2008continuum}. We consider next Eq. \eqref{sys:macro3} for $\bar \q$. Observe first that all the terms in the equation belong to the tangent space at $\bar\q$ in $\unitq$, i.e, to $\bar\q^\perp$. This is true for the first term since $\lp \partial_t + c_2(\vezero(\bar\q)\cdot \nabla_x\rp$ is a differential operator (giving the transport of $\bar\q$) and it also holds for the rest of the terms since they are of the form $\mathbf{u}\bar\q$ with $\mathbf{u}$ purely imaginary (see Prop. \ref{prop:tangent space} in the appendix). 

The term corresponding to $c_3$ gives the influence of $\nabla_x\rho$ (pressure gradient) on the body attitude $\bar\q$. It has the effect of rotating the body around the vector directed by $\vezero(\bar\q)\times\nabla_x \rho$ at an angular speed given by $\frac{c_3}{\rho}\|  \vezero(\bar\q)\times\nabla_x \rho\|$, so as to align $\vezero(\bar\q)$ with $-\nabla_x\rho$. Indeed the solution to the differential equation 
$$\frac{\ud\q}{\ud t} + \gamma \mathbf{u}\q=0,$$
when $\mathbf{u}$ is a constant purely imaginary unitary quaternion and $\gamma$ a constant scalar, is given by $\q(t)=\exp(-\gamma \mathbf{u}t)\q(0)$, and $\exp(-\gamma\mathbf{u}t)$ is the rotation of axis $\mathbf{u}$ and angle $-\gamma t$ (see Eq. \eqref{eq:exponential form unitary quaternion}). Since $c_3$ is positive, the influence of this term consists of relaxing the direction of movement $\vezero(\bar\q)$ towards $-\nabla_x\rho$, i.e., making the agents move from places of high concentration to low concentration. In this manner, the $\nabla_x\rho$ term has the same effect as a pressure gradient in classical hydrodynamics. In the present case the pressure gradient provokes a change in the full body attitude $\bar \q$. 
 Finally, notice that in regions where $\rho>0$ we can divide Eq. \eqref{sys:macro3} by $\rho$ and this gives us the influence of each term depending on the local density. After division by $\rho$, we observe that the only term depending on the density $\rho$ in Eq. \eqref{sys:macro3} is on the third term in the form
$$c_3\, \left[ \vezero(\bar \q)\times \frac{\nabla_x \rho}{\rho}\right]\bar\q.$$
Therefore, for small densities, this term may take large values and become dominant, while for large densities it becomes small and the other terms in the equation prevail for reasonably large $\nabla_x\rho$. The fact that agents tend to relax their direction of motion towards regions of lower concentration creates dispersion: this term is a consequence of the noise at the microscopic level. However, the relaxation becomes weaker once agents are in regions of high density or areas with small variations of density.  The last two terms in Eq.  \eqref{sys:macro3} are unique to the body attitude coordination model and are the main difference with respect to the SOH equations for the Vicsek model.

\medskip

Analogously to the discussions in Refs. \cite{bodyattitude,degond2008continuum} for the body attitude model based on rotation matrices and for the Vicsek model, the SOHQ model bears analogies with the compressible Euler equations, where Eq. \eqref{sys:macro1} is the mass conservation equation and Eq. \eqref{sys:macro3} is akin to the momentum conservation equation, where momentum transport  is balanced by a pressure force. There are however major differences. Firstly, the pressure term belongs to $\bar\q^\perp$ in order to ensure that $\bar\q\in\unitq$ for all times; in the Euler equations the velocity is an arbitrary vector, not necessarily normalized. Secondly, the convection speed $c_2$ is a priori different from the mass conservation speed $c_1$. This difference signals the lack of Galilean invariance of the system, which is a common feature of all dry active matter models (models for collective motion not taking place in a fluid), see Ref. \cite{tu1998sound}. Finally, the last two terms of Eq. \eqref{sys:macro4} do not have a clear analog to the compressible Euler equations: they seem quite specific to our model.

\subsubsection{The equation as a relative variation}

The (right) relative differential operator $\pa_{\text{rel}}$ can be interpreted as the (right) relative variation of $\q$, i.e.,
$$\pa_{\text{rel}} \q=\pa \q \,\q^{-1},$$
where $\q^{-1}=\q^\ast$ is the inverse of $\q$ since $\q$ is unitary.

This expression would have a clear meaning in a commutative setting. For example in the case of (unit) complex numbers (that is, if we consider rotations in 2 dimensions), we consider the analogous definition (for $z=z(t,x)$ a function with values in the group $\mathbb{U}$ of unitary complex numbers)
$$\pa_{\text{rel},\mathbb{C}} z=\pa z \,z^{-1},$$
where $z^{-1}$ is the inverse of $z$ (which is also its complexe conjugate). In this case because $\mathbb{C}$ is commutative we can write simply
$$\pa_{\text{rel},\mathbb{C}} z=\frac{\pa z}{z}.$$
Equivalently, we can also recognize
$$\pa_{\text{rel},\mathbb{C}} z=\pa \lp\log z\rp,$$
and the interpretation in terms of a relative variation is standard.

Let us go back to quaternions. The logarithm is well defined on $\unitq\setminus\{-1\}$, and for $\q\in\unitq\setminus\{-1\}$, we have that $\log\q=\log(\exp(\theta\nvec/2))=\theta\nvec/2$ with the notations of Eq. \eqref{eq:exponential form unitary quaternion}. But, because of the lack of commutativity of $\mathbb{H}$, it is not clear that the logarithm and the relative operator satisfy any relevant relation globally. Since such an interpretation cannot be, \emph{a priori}, directly translated to quaternions, we propose the following \emph{local} interpretation. Locally around a fixed point $(t_0,x_0)\in\R^+\times\R^3$, we can write
\be \label{eq:r}
\q(t,x) = \mathbf{r}(t,x) \q(t_0,x_0),
\ee
where $\mathbf{r}(t,x)=\q(t,x)\q(t_0,x_0)^\ast\in\unitq$ represents the variation of $\q$ around $\q(t_0,x_0)$ with $\mathbf{r}(t_0,x_0)= 1$. Then, with these notations, it holds that
 $$\pa_{\text{rel}} \q=\left.(\pa \mathbf{r})\right|_{(t,x)=(t_0,x_0)}.$$
 
 \begin{remark} When $\pa=\pa_t$ is the time derivative, for a function $\q=\q(t,x)$ with values in $\unitq$, the vector $\partial_{t,\textnormal{rel}} \q =\pa_t \q \,\q^{-1}$ is half of the angular velocity of a solid of orientation represented by $\q$. By analogy, the vector $\partial_{x_i, \textnormal{rel}} \q=\pa_{x_i}\q \,\q^{-1}$ for $i=1,\,2,\,3$ is half of the angular variation in space of a solid of orientation represented by $\q$.
 \end{remark}
 
Multiplying from the right the evolution equation \eqref{sys:macro3} for $\bar \q$  by $\bar\q^\ast$, we obtain the following equivalent equation
\begin{eqnarray}
&&\rho \partial_{t,\textnormal{rel}} \bar \q + \rho  c_2 (\vezero(\bar\q) \cdot \nabla_{x,\textnormal{rel}}) \bar \q + c_3\left[\vezero(\bar\q)\times\nabla_x \rho\right]  \nn \\
&& \quad \qquad +c_4\rho \left[  \nabla_{x,\textnormal{rel}}\bar \q \, \vezero(\bar \q) +  (\nabla_{x,\textnormal{rel}}\cdot \bar \q) \vezero(\bar\q) \right] =0. \label{eq:q_relative_form}
\end{eqnarray}
In this equation we notice that all the differential operators naturally appear under their (right) relative form. Notice also that all other nonlinearities in $\bar \q$ are expressed in terms of $\vezero(\bar\q)$. Therefore, the previous system can be interpreted as the evolution of the relative changes of $\bar\q$.

In terms of $\mathbf{r}$, the previous system can be recast into:
 \begin{eqnarray*}
\Big[\rho \partial_t \mathbf{r} + \rho  c_2 (\vezero(\q) \cdot \nabla_x) \mathbf{r} + c_3 \vezero(\q)\times\nabla_x \rho
+\left.c_4\rho  \left[ \nabla_{x}\mathbf{r} \, \vezero(\q) + (\nabla_{x}\cdot \mathbf{r})\vezero(\q) \right]\Big]\right|_{(t_0,x_0)}\\ =0.
\end{eqnarray*}
For an interpretation (again, \emph{local}) in terms of $\mathbf{b}:=\log \mathbf{r}$ we refer the reader to Sec. \ref{sec:comparison_macro}.

\subsection{Equivalence with the previous body attitude model}

In Ref. \cite{bodyattitude} a model for body attitude coordination is presented were the body attitude is represented by a rotation matrix (element in $SO(3)$, orthonormal group) rather than by a quaternion. In Sec. \ref{sec:compare_micro} we will prove the equivalence between the individual based model presented in Ref. \cite{bodyattitude} and the one here, in the sense that the two stochastic processes are the same in law (Cor. \ref{cor:equivalence_process}). 

In Ref. \cite{bodyattitude} also the macroscopic equations  are obtained for the mean body attitude $\Lambda=\Lambda(t,x)\in SO(3)$ and spatial density of agents $\rho=\rho(t,x)\ge0$, called Self-Organised Hydrodynamics for Body Attitude Coordination (SOHB):
\begin{align}
&\partial_t \rho + \nabla_x \cdot \big( \tilde c_1 \rho \Lambda \mathbf{e_1}\big)=0, \label{eq:macro_rho}\\
&\rho \Big( \partial_t\Lambda+ \tilde c_2 \big((\Lambda \vezero) \cdot \nabla_x\big)\Lambda \Big)\Lambda^t +\left[ (\Lambda \vezero) \times \big( \tilde c_3 \nabla_x \rho+\tilde c_4\rho\,\rvec_x(\Lambda)\big) + \tilde c_4\rho\,\delta_x(\Lambda)\Lambda \vezero\right]_\times =0,\label{eq:macro_lambda}
\end{align}
with explicit constants $\tilde c_i$, $i=1,\hdots,4$, where $\Lambda^t$ indicates the transpose matrix of $\Lambda$, and where for a vector $\mathbf{u}=(u_1,u_2, u_3)$, the antisymmetric matrix~$\left[ \mathbf{u}\right]_\times$, is defined by
\be \label{eq:def_operator_asym}
\left[ \mathbf{u}\right]_\times := \left[
\begin{array}{ccc}
0 & -u_3 & u_2\\
u_3 & 0 & -u_1\\
-u_2 & u_1 & 0
\end{array}
\right].
\ee  
The scalar~$\delta_x(\Lambda)$ and the vector~$\rvec_x(\Lambda)$ are first order differential operators intrinsic to the dynamics. We define them next. For a smooth function $\Lambda=\Lambda(x)$ from $\R^3$ to $SO(3)$, we define the matrix $\mathscr{D}_x(\Lambda)$ by the equality
\be\label{def:D_x}
\text{for all } \mathbf{w}\in\R^3, \quad (\mathbf{w} \cdot \nabla_x) \Lambda = [\mathscr{D}_x(\Lambda)\mathbf{w}]_\times \Lambda,
\ee
(the matrix $\mathscr{D}_x(\Lambda)$ is well defined, see Ref. \cite{bodyattitude}).
The first order operators $\delta_x(\Lambda)$ and $\rvec_x(\Lambda)$ are then defined by
\be\label{def:delta_and_r}
\delta_x(\Lambda) = \textnormal{Tr}\lp \mathscr{D}_x(\Lambda) \rp, \qquad [\rvec_x(\Lambda)]_\times = \mathscr{D}_x(\Lambda)-\mathscr{D}_x(\Lambda)^t. 
\ee

Since the individual based model formulated in terms of quaternions is equivalent (in law) to the one formulated with rotation matrices, we expect their respective macroscopic limits to be also equivalent. This is the case, as expressed in Th. \ref{th:equivalence_macro_equations}, i.e., if at time $t=0$, $\Lambda(0)$ and $\bar\q(0)$ represent the same rotation, then $\Lambda(t)$ and $\bar \q(t)$ represent the same rotation for all $t$ where the solutions are well defined. 

There are, however, important differences between the SOHB and SOHQ macroscopic equations. On one hand, notice that the operators $\delta_x$ and $\rvec_x$ cannot be expressed under a simple explicit form, which makes the meaning of these operators less clear. In the quaternion case, all the elements in Eqs. \eqref{sys:macro1}--\eqref{sys:macro3} are explicit. Moreover, quaternions give the right framework for numerical simulations (in terms of memory and operation efficiency), as explained in the introduction. On the other hand, when using rotation matrices, we obtain clear equations for the evolution of each one of the orthonormal vectors that define the body frame (see Ref. \cite{bodyattitude}). However, the expressions for these vectors  in the quaternion formulation  is complicated and little revealing, due to the quadratic structure of the rotation (see Eq. \eqref{sys:macro4}).

\section{Modeling: the individual based model and its mean-field limit}
\label{sec:modeling}

\subsection{The individual based model}

Consider a reference frame in $\R^3$ given by the orthonormal basis $\{\vezero, \mathbf{e_2}, \mathbf{e_3}\}$.  Consider, also, $N$ agents labelled by $k=1,\hdots, N$ with position $X_k(t)\in \R^3$ and body attitude given by the unitary quaternion $\q_k(t)\in \unitq$. As explained in the introduction, the body frame of agent $k$ corresponds to $\{\vezero(\qk), \mathbf{e_2}(\qk), \mathbf{e_3}(\qk) \}$, where $\mathbf{e_i}(\qk)$ denotes the rotation of $\mathbf{e_i}$ by $\qk$ for $i=1,2,3$.  The vector $\vezero(\qk)$ gives the direction of motion of agent $k$ and the other two vectors give the position of the body relative to this direction (see Fig. \ref{Fig:quaternion_explained}).
Our goal is to model collective dynamics where agents move at a constant speed while adopting the body attitude of their neighbours, up to some noise.

\medskip

\textit{Evolution of the positions, $(X_k)_{k=1,\hdots,N}$}. The fact that agent $k$ moves in direction $\vezero(\qk)$ at constant speed $v_0>0$, simply corresponds to the equation
$$\frac{\ud X_k}{\ud t} = v_0 \vezero(\qk), \quad \vezero(\qk):= \Ima(\qk \vezero \qk^*),$$
(recall Eq. \eqref{eq:rotation by quaternion} and Rem. \ref{rem:abuse of notation}).
Notice that the speed for all agents is constant and equal to $v_0>0$.

\medskip
\textit{Evolution for the body attitudes, $(\qk)_{k=1,\hdots,N}$.}
Agents try to coordinate their body attitudes with those of their neighbours. To model this phenomenon, we need to firstly define an "average" body attitude around a given agent $k$ and to secondly express the relaxation of the body attitude of agent $k$ towards this average. 

\begin{remark}[Nematic alignment, sign invariance]
\label{re:invariance_change_of_sign}
The body attitude is uniquely defined by quaternions up to a sign since $\q_k$ and $-\q_k$ represent the same rotation. This implies, firstly,  that the time evolution equation for $\q_k=\q_k(t)$ must be sign invariant, and secondly, that the average must take this sign invariance into account. This is called `nematic alignment' (in opposition to `polar' alignment) and it appears in other collective models \cite{degond2015continuum,degond2015multi}  and in liquid crystals \cite{doi1988theory}. Therefore, we cannot define the average analogously as in Ref. \cite{degond2008continuum}  since the alignment is polar in this case. For example, if one considers the normalised averaged quaternion defined in the same way as in the Vicsek model:
\be \label{eq:wrong average}
\frac{\sum_{i=1}^N \q_i}{\left|\sum_{i=1}^N \q_i\right|} \in \unitq,
\ee
we obtain a unitary quaternion that can be interpreted as a rotation. However, the meaning of this rotation is unclear and it is not invariant under changes of sign of any of the vectors $\q_i$. We cannot use, either the nematic average used in Ref. \cite{degond2015continuum} since it is only valid in $\R^2$. 
\end{remark}

 We define (up to a sign) the average around $\qk$  by
\beqarl \nn
\bar\q_k &:=& \mbox{unitary eigenvector of the maximal eigenvalue of $Q_k$}\\
&=& \arg \max \{ \q\cdot Q_k \q, \, \q \in \unitq \}, \label{eq:definition bar q}
\eeqarl
with
\beqarl \label{eq:definition of Q_k}
Q_k = \frac{1}{N} \sum_{i=1}^N K(|X_i-X_k|) \left(\q_i \otimes \q_i - \frac{1}{4} \Id \right),
\eeqarl
where the nonnegative-valued function $K$ is a kernel of influence. It is in the definition that $\bqk\in \unitq$; one can check that if $\bqk$ is an average, so is $-\bqk$ (so it is sign invariant); $\bqk$ remains invariant under the change of sign of any of the arguments $\q_1, \hdots, \q_N$; and, $\bqk$ maximises over $\unitq$:
\beqarl
\q \mapsto \q\cdot Q_k \q &=& \frac{1}{N}\sum_{i=1}^N K(|X_i-X_k|)\, \lp (\q_i\cdot \q)^2-\frac{1}{4}\rp \nonumber\\
&=&\frac{1}{N}\sum_{i=1}^N K(|X_i-X_k|)\, \lp\cos^2(\widehat{\la \q_i,\q\ra})-\frac{1}{4}\rp, \label{eq:definition average}
\eeqarl
where $\widehat{\la \q_i,\q\ra}$ denotes the angle between $\q_i$ and $\q$ (seen as elements of the hypersphere $\mathbb{S}^3$).

\medskip
Now to express the relaxation of $\q_k$ towards this average we define first
\beqarl \label{eq:definition F_k}
F_k =  \left(\bqk \otimes \bqk - \frac{1}{4} \Id\right) \qk,
\eeqarl
and write
\be \label{eq:relaxation_term}
\frac{\ud\qk}{\ud t} = P_{\mathbf{q}_\mathbf{k}^\perp}F_k\, \lp =\frac{1}{2}\nabla_{\qk}\lp\qk \cdot F_k \rp \rp,
\ee
where $P_{\mathbf{q}_\mathbf{k}^\perp}$ indicates the projection on the orthogonal space to $\qk$   (which corresponds to the tangent space of $\qk$ in $\unitq$); $\nabla_{\qk}$ indicates the gradient in $\unitq$ (the second equality in Eq. \eqref{eq:relaxation_term} is proven in Prop. \ref{prop:}). Eq. \eqref{eq:relaxation_term} relaxes $\qk$ towards the maximizer of $\q\mapsto \q \cdot\left(\bqk \otimes \bqk - \frac{1}{4} \Id\right)\q$ which corresponds precisely to $\bqk$ or $-\bqk$.

\medskip

Finally, putting everything together, we obtain the evolution equations:
\beqarl \label{eq:particleX}
\frac{\ud X_k}{\ud t} &=& v_0 \vezero(\qk), \quad \vezero(\qk):= \Ima(\qk \vezero \qk^*),\\
\ud\qk &=& P_{\mathbf{q}_\mathbf{k}^\perp} \circ \lp\nu F_k \ud t+ \sqrt{D/2}\,\ud B_t^k \rp. \label{eq:particleQ}
\eeqarl
where $\nu>0$ indicates the intensity of the relaxation.
The evolution for the body attitudes results from two competing phenomena: body attitude coordination (the $F_k$-term) and noise due to errors that the agents make when trying to coordinate. The noise term is given by $(B_t^k)_{k=1,\hdots,N}$ independent 4-dimensional Brownian motions. The constant $D>0$ gives the intensity of the noise term (by modifying the variance of the Brownian motion). The projection $P_{\qk^\perp}$ in Eq. \eqref{eq:particleQ} ensures that $\qk(t)\in \unitq$ for all times $t$ where the solution is defined. The stochastic differential equation \eqref{eq:particleX}-\eqref{eq:particleQ} must be understood in the Stratonovich sense, see Ref. \cite{hsu2002stochastic}.
 
\begin{remark} Some comments:
\begin{enumerate}[(i)]
\setlength\itemsep{0em}
\item Typically the noise term would be scaled by $\sqrt{2D}$, because then the generator of the process is the laplacian with coefficient $D$. However, we chose the scaling $\sqrt{D/2}$ to make the model equivalent with the one based on rotation matrices in Ref. \cite{bodyattitude}. This will be discussed in Sec. \ref{sec:compare_micro}.
\item The operator $Q_k$ in Eq. \eqref{eq:definition of Q_k} corresponds to the de Gennes Q-tensor that appears in the theory of liquid crystal \cite{doi1988theory} and which is also related to the so-called `nematic order coefficient'. Notice that in the definition of $Q_k$ in Eq. \eqref{eq:definition of Q_k} the $1/4$-factor could be ignored and the definition of the average $\bqk$ would remain unchanged. Also in $F_k$ in Eq. \eqref{eq:definition F_k} the $1/4$-factor can be ignored since that term dissapears in the projection  in Eq. \eqref{eq:relaxation_term}. We keep it here for the parallelisms that it bears with the theory of liquid crystal and because it will appear when we define the equilibrium distribution in Eq. \eqref{eq:equilibria}.

\item Notice that to define the average $\bqk$ in Eq. \eqref{eq:definition average}, we assume that the maximal eigenvalue is \emph{simple}. At the formal level, this assumption is reasonnable since for general symmetric matrices the event of a multiple maximal eigenvalue is negligible. Of course, for a rigourous analysis, we would need to ensure carefully that this event can actually be neglected, or we would need to add an extra rule to determine uniquely the average.
\item Notice that we could have defined the relaxation $F_k$ by considering directly $F_k=Q_k \qk$ instead of Eq. \eqref{eq:definition F_k}, since, in this case, Eq. \eqref{eq:relaxation_term} relaxes also to $\bqk$. However, for this case the relaxation is weaker. This is a modeling choice. We will prove in Sec. \ref{sec:compare_micro} that our choice here is the one that corresponds to the model presented in Ref. \cite{bodyattitude} where the body attitude is described with rotation matrices.
\item One can check that the particle system \eqref{eq:particleX}-\eqref{eq:particleQ} is frame invariant, in the sense that if $\tilde{X}_k=R_{\textnormal{frame}}(X_k)$ for $k=1\dots N$, with $R_{\textnormal{frame}}$ the frame change associated to the quaternion $q_{\textnormal{frame}}$, and $\tilde{q}_k=q_{\textnormal{frame}}q_k$, then the pair $(\tilde{X}_k, \tilde{q}_k)$ satisfies the same system (with the appropriate initial conditions).

\end{enumerate}
\end{remark}

\subsection{Mean-field limit}
\label{sec:mean_field_limit}

We now obtain formally the mean-field limit for Eqs. \eqref{eq:particleX}--\eqref{eq:particleQ} as the number of particles $N\to \infty$. The rigorous mean-field limit has been proven for the Vicsek model in Ref. \cite{bolley2012mean}. A key difference between the Vicsek model and the system \eqref{eq:particleX}--\eqref{eq:particleQ} is the way we compute the average in Eq. \eqref{eq:definition average}.
Consider the empirical distribution in $(x,\q)\in\R^3\times \unitq$ over time 
$$f^N(t, x, \q): = \frac{1}{N} \sum_{i=1}^N \delta_{(X_i(t), \q_i(t))}(x,\q),$$
where $(X_i(t),\q_i(t))_{i=1,\hdots, N}$ satisfy Eqs. \eqref{eq:particleX}--\eqref{eq:particleQ}.

Assume that $f^N$ converges weakly to $f=f(t,x,\q)$ as $N\to\infty$. 
It is standard to show (formally) that $f$ satisfies
\beqarl \label{eq:kinetic_equation}
&&\pa_t f + \nabla_x \cdot (v_0 \vezero(\q) f) + \nabla_\q \cdot (F_f f) = \frac{D}{4} \Delta_\q f,\\
&&F_f = \nu {P_{\q^\perp}}\bar \q^K_f \left( \bar\q^K_f \cdot \q \right) =  \nu {P_{\q^\perp}} \left(\bar \q_f^K \otimes \bar \q_f^K \right) \q,\nn\\
&&\bar \q^K_f =\text{unitary eigenvector of the maximal eigenvalue of $Q^K_f$}\nn\\
&&\qquad = \arg \max \{ \q \mapsto \q\cdot Q^K_f \q \}, \nn\\
&& Q_f^K = \int_{\R^3} \int_{\unitq} K(|x-y|) \lp \q \otimes \q -\frac{1}{4}\Id \rp f(t,y,\q) \ud \q \ud y, \nn
\eeqarl
where $\nabla_\q$ and $\Delta_\q$ denote the gradient and the laplacian in $\unitq$, respectively, and $\ud \q$ is the Lebesgue measure on $\unitq$.

\section{Hydrodynamic limit}

The goal of this section is the derivation of the macroscopic equations for Eq. \eqref{eq:kinetic_equation}. After a dimensional analysis and a time and space scaling described next in Sec. \ref{sec:scaling}, we recast Eq. \eqref{eq:kinetic_equation} into
\beqarl
\label{eq:f_eps}
\pa_t f^\eps + \nabla_x \cdot (  \vezero(\q) f^\eps)   &=& \frac{1}{\eps}\Gamma(f^\eps) + \mathcal{O}(\eps),\\
\Gamma(f) &:=& -\nu \nabla_\q \cdot ({P_{\q^\perp}}\lp\lp \bar\q_{f} \otimes \bar\q_{f}\rp \q\rp\, f) +\frac{D}{4} \Delta_\q f ,\label{eq:def_Gamma}
\eeqarl
with $\bar \q_{f}$ defined by
\beqarl
\bar \q_f &=& \text{unitary eigenvector of the maximal eigenvalue of $Q_f$} \nonumber \\
\label{eq:definition_q_f} &=& \arg \max \{ \q \mapsto \q\cdot Q_f \q,\, \q\in\unitq \},
\eeqarl
where
\be\label{def:Q_f} Q_f= \int_\unitq f(t,x,\q) \, \lp \q \otimes \q - \frac{1}{4} \Id \rp \, \ud\q,\ee
$\ud \q$ being the Lebesgue measure on $\unitq$. (Note that after the dimensional analysis and the rescaling the values of the parameters $D$ and $\nu$ as well as the variables $t$ and $x$ have changed: see details in Sec. \ref{sec:scaling}.)

We then analyse in Sec. \ref{sec:equilibrium} the collisional operator $\Gamma$ in Eq. \eqref{eq:def_Gamma}, particularly, we determine its (von-Mises-like) equilibria, given by (for $\bar \q\in\unitq$)
\be \label{eq:equilibria}
M_{\bar \q}(\q) = \frac{1}{Z}\exp\lp\frac{2}{d}\lp\lp \bar \q \cdot \q\rp^2-\frac{1}{4} \rp\rp,
\ee
where $d$ is a parameter given by
\be d=\frac{D}{\nu}, \label{eq:def_d}
\ee
and where $Z$ is a normalizing constant (such that $\int_\unitq M_{\bar \q} \ud\q =1$). We then describe the structure of the \emph{Generalized Collision Invariants} for $\Gamma$ in Sec. \ref{sec:GCI}.

With this information we are ready to prove our main result:

\begin{theorem}[(Formal) macroscopic limit]
 \label{th:macro_limit}
When $\eps\to 0$ in the kinetic equation \eqref{eq:f_eps} it holds (formally) that
$$f^\eps \to f= f(t,x, \q)=\rho M_{\bar\q}(\q), \quad \bar\q=\bar\q(t,x)\in \unitq, \, \rho=\rho(t,x) \geq 0.$$
Moreover, if the convergence is strong enough and the functions $\bar \q$ and $\rho$ are regular enough, then they satisfy the system \eqref{sys:macro1}--\eqref{sys:macro3} that we recall here:
\begin{eqnarray}
\label{eq:continuity_equation}&&\pa_t \rho +\nabla_x \cdot (c_1  \vezero ( \bar \q) \rho) = 0,\\
\nn &&\rho \lp \partial_t \bar \q +  c_2 (\vezero(\bar\q) \cdot \nabla_x) \bar \q\rp + c_3\left[\vezero(\bar\q) \times\nabla_x \rho \right] \bar \q  \\
\label{eq:constants} && \qquad \quad +c_4\rho  \left[  \nabla_{x,\textnormal{rel}}\bar \q \,\vezero(\bar \q) +  (\nabla_{x,\textnormal{rel}}\cdot \bar \q) \vezero(\bar\q) \right]\bar \q =0,
\end{eqnarray}
where the (right) relative differential operator $\nabla_{x,\textnormal{rel}}$ is defined in Sec. \ref{sec:SOHQ}, where
\begin{eqnarray}
\vezero(\bar \q) &=& \Ima(\bar \q\vezero \bar \q^*),
\end{eqnarray}
and where $c_i$, $i=1,\hdots, 4$ are explicit constants. To define them we use the following notation: for two real functions $g$, $w$ consider
\be\la g \ra_{w} := \int^\pi_0 g(\theta) \frac{w(\theta)}{\int^\pi_0 w(\theta') \ud\theta'}\, \ud\theta \label{eq:definition_angles}.\ee
Then the constants are given by
\begin{eqnarray}
\label{eq:c1} 
c_1&=& \frac{2}{3}\la 1/2+\cos\theta\ra_{ m\, \sin^2(\theta/2)},\\
c_2&=& \frac{1}{5} \la 1 + 4 \cos\theta \ra_{m \sin^4(\theta/2) h(\cos(\theta/2)) \cos(\theta/2)},\\
c_3&=&\frac{d}{2},\\
c_4&=& \frac{1}{5}\la 1-\cos\theta \ra_{m\sin^4(\theta/2) h(\cos(\theta/2)) \cos(\theta/2)}, 
\end{eqnarray}
where $d$ is given in Eq. \eqref{eq:def_d}, where
\be
m(\theta) := \exp\lp d^{-1}\lp\frac{1}{2}+\cos\theta\rp \rp, \label{eq:tilde_m}
\ee
and where $h$ is the solution of the differential equation \eqref{eq:ode_h}.
\end{theorem}

We recall that in Sec. \ref{sec:discussion_results} we provided a discussion of this main result. The proof is given in Sec. \ref{sec:macro_limit}. We conclude this study with Sec. \ref{sec:comparison_macro} where we compare the macroscopic limit obtained here with the corresponding one for the body attitude model with rotation matrices  from Ref. \cite{bodyattitude}.

\subsection{Scaling and expansion}
\label{sec:scaling}
We assume that the kernel of influence $K$ is Lipschitz, bounded and such that
\be \label{eq:K_properties}
K=K(|x|)\geq 0, \quad \int_{\R^3}K(|x|)\, d x=1, \quad \int_{\R^3}|x|^2 K(|x|)\, dx <\infty.
\ee

We express the kinetic equation \eqref{eq:kinetic_equation} in dimensionless variables. Let $\nu_0$ be the typical interaction frequency, i.e., $\nu=\nu_0 \nu'$ with $\nu'=\mathcal{O}(1)$. We consider also the typical time and space scales $t_0, x_0$ with $t_0=\nu_0^{-1}$ and $x_0=v_0t_0$. With this we define the non-dimensional variables $t'=t/t_0$, $x'=x/t_0$. Consider also the dimensionless diffusion coefficient $D'=D/\nu_0$ and the rescaled influence kernel $K'(|x'|)=K(x_0|x'|)$. Skipping the primes we get the same equation as Eq. \eqref{eq:kinetic_equation} except that $v_0=1$, all the quantities are dimensionless and $D, \nu$ and $K$ are assumed to be of order 1. Notice, in particular, that
$$d^{-1}=\frac{\nu}{D}=\frac{\nu'}{D'},$$
is the same before and after the dimensional analysis.

To perform the macroscopic limit we rescale space and time by $x'=\eps x$ and $t'=\eps t$. After skipping the primes we obtain:
\beqarl
&&\eps\left[\pa_t f^\eps + \nabla_x \cdot (  \vezero(\q) f^\eps)\right] + \nabla_\q \cdot (F^\eps_{f^\eps} f^\eps) = \frac{D}{4} \Delta_\q f^\eps, \label{eq:rescaled_kinetic equation}\\
&&F^\eps_f = \nu {P_{\q^\perp}}\bar \q^\eps_f \left(\bar \q^\eps_f \cdot \q \right) =  \nu {P_{\q^\perp}} \left(\bar \q^\eps_f \otimes \bar \q^\eps_f \right) \q,\nn\\
&&\bar \q^\eps_f = \text{unitary eigenvector of the maximal eigenvalue of $Q^\eps_f$}\nn\\
&& \qquad= \arg \max \{ \q \mapsto \q\cdot Q^\eps_f \q, \, \q\in\unitq\},\label{eq:def_qf}\\
&&Q^\eps_f = \frac{1}{\eps^3} \int_{\R^3} \int_{\unitq} K\lp \frac{|x-y|}{\eps}\rp  \lp \q \otimes \q - \frac{1}{4} \Id \rp f(t,y,\q)\, \ud \q \ud y. \nn
\eeqarl

\begin{lemma}\label{lem:Q_expansion}
For any sufficiently smooth function $f$, we have the expansion
$$Q^\eps_f = \int f(t,x,\q)\,  \lp \q \otimes \q - \frac{1}{4} \Id \rp \, \ud\q + \mathcal{O}(\eps^2).$$

\end{lemma}
\begin{proof}
The result is obtained by a Taylor expansion in $\eps$ and using that (recall Eq. \eqref{eq:K_properties})
$$\frac{1}{\eps^3}\int K\lp \frac{|x|}{\eps}\rp\, dx =1, \qquad
\frac{1}{\eps^3}\int|x|^2 K(|x|) \, dx =\mathcal{O}(\eps^2).$$
\end{proof}

\begin{proposition}
\label{prop:expansion_q}
For any sufficiently smooth function $f$, it holds that
\be\label{exp:q_f}\bar \q^\eps_f= \left( \bar\q^\eps_f \cdot \bar \q_f \right) \bar \q_f + \mathcal{O}(\eps^2) \quad \mbox{as } \eps\to 0.\ee
In particular, we have
$$\bar \q^\eps_f \otimes \bar \q^\eps_f= \bar \q_f \otimes \bar \q_f + \mathcal{O}(\eps^2) \quad \mbox{as } \eps\to 0.$$
\end{proposition}
\begin{proof}
Let $\lambda^\eps_{max}$, respectively $\lambda_{max}$, be the maximal eigenvalue of $Q^\eps_{f}$, respectively $Q_f$ (we assume them to be uniquely defined). From Lem. \ref{lem:Q_expansion}, we have
$$Q^\eps_f = Q_f+\mathcal{O}(\eps^2),$$
and, multiplying by $\bar\q^\eps_f$ on both sides,
\beqar
\lambda^\eps_{max}&=& \bar\q^\eps_f \cdot Q_f \bar\q^\eps_f + \mathcal{O}(\eps^2).
\eeqar
By maximality of $\lambda_{max}$, we have that $\bar\q^\eps_f \cdot Q_f \bar\q^\eps_f \le \lambda_{max}$ so that
\beqar
\lambda^\eps_{max} \le \lambda_{max}+ \mathcal{O}(\eps^2).
\eeqar
By symmetry, we also have $\lambda_{max} \le \lambda^\eps_{max}+ \mathcal{O}(\eps^2)$, therefore
$$\lambda^\eps_{max}-\lambda_{max}=\mathcal{O}(\eps^2).$$

\medskip
On the other hand, we have that
\beqar
(Q_f-\lambda_{max}\Id) P_{\bar \q_f^\perp}\bar \q^\eps_f
&=&(Q_f-\lambda_{max}\Id) (\bar \q^\eps_f-(\bar \q^\eps_f\cdot \bar \q_f)\bar \q_f)\\
&= & (Q_f-\lambda_{max}\Id) \bar \q_f^\eps -0\\
&=& (Q^\eps_f-\lambda^\eps_{max}\Id + \mathcal{O}(\eps^2))\bar \q_f^\eps\\
&=& \mathcal{O}(\eps^2).
\eeqar
By our crucial assumption that $\lambda_{max}$ is a single eigenvalue with eigenvector $\bar \q_f$, we can invert the matrix $(Q_f-\lambda_{max}\Id)$ on the 3-dimensional space $\bar \q_f^\perp$. By a small abuse of notation we write $(Q_f-\lambda_{max}\Id)^{-1}$ its inverse on $\bar \q_f^\perp$. Finally we have
\beqarl
P_{\bar \q_f^\perp}\bar \q^\eps_f = (Q_f-\lambda_{max}\Id)^{-1}  \mathcal{O}(\eps^2) =  \mathcal{O}(\eps^2),
\eeqarl
which proves Eq. \eqref{exp:q_f}. Taking the scalar product with $\bar \q^\eps_f$ and using the fact that $\bar \q^\eps_f$ is unitary, we have that
$$ 1 - (\bar \q_f\cdot \bar \q^\eps_f)^2 =  \mathcal{O}(\eps^2),$$
so using  Eq. \eqref{exp:q_f} we can finally show that
$$\bar \q^\eps_f \otimes \bar \q^\eps_f= \bar \q_f \otimes \bar \q_f + \mathcal{O}(\eps^2).$$
\end{proof}

Using Prop. \ref{prop:expansion_q} we recast the rescaled kinetic equation \eqref{eq:rescaled_kinetic equation} as
Eq. \eqref{eq:f_eps}

\subsection{Equilibrium solutions and Fokker-Planck formulation}\label{sec:equilibrium}
Define $d=D/\nu$ and consider the generalisation of the von-Mises distribution in $\unitq$:
\be
M_{\bar \q}(\q) = \frac{1}{Z}\exp\lp\frac{2}{d}\lp\lp \bar \q \cdot \q\rp^2-\frac{1}{4} \rp\rp, \quad \int_{\unitq}M_{\bar\q}(\q)\, \ud\q =1, \quad \bar \q\in\unitq,
\ee
where $Z$ is a normalizing constant. Observe that $Z<\infty$ is independent of $\bar\q$ since the volume element in $\unitq$ is left-invariant, i.e.,
\beqar
Z&=& \int_{\unitq}\exp\lp\frac{2}{d}\lp\lp \bar \q \cdot \q\rp^2-\frac{1}{4} \rp\rp \, \ud\q = \int_{\unitq} \exp\lp\frac{2}{d}\lp\lp 1 \cdot \q\bar\q^*\rp^2-\frac{1}{4} \rp\rp \, \ud\q\\
&=& \int_{\unitq} \exp\lp\frac{2}{d}\lp\lp 1 \cdot \q\rp^2-\frac{1}{4} \rp\rp\, \ud\q.
\eeqar
Note that we can recast
\be \label{eq:M_and_m}
M_{\bar \q}(\q) = \frac{m(\theta)}{4\pi \int_0^\pi m(\theta') \sin^2(\theta'/2) \ud\theta'}, \quad \text{with } \bar\q\cdot\q = \cos(\theta/2),
\ee
 with $m(\theta)$ given by Eq. \eqref{eq:tilde_m}. Indeed,
\beqar
M_{\bar \q}(\q) &=& \frac{1}{Z} \exp\lp \frac{2}{d} (\cos^2(\theta/2)-\frac{1}{4})\rp = \frac{1}{Z} \exp\lp \frac{1}{d} (\cos(\theta)+\frac{1}{2})\rp= \frac{1}{Z} m(\theta),
\eeqar
and by Prop. \ref{prop:volume_element},
\be
Z = 4\pi \int_0^\pi \exp\lp \frac{2}{d} (\cos^2(\theta'/2)-\frac{1}{4})\rp \sin^2(\theta'/2) \ud\theta' = 4\pi \int_0^\pi m(\theta') \sin^2(\theta'/2) \ud\theta'.
\ee
\begin{proposition}[Properties of $\Gamma$]
\label{prop:properties_Gamma}
 The following holds:
\begin{itemize}
\item[i)] The operator $\Gamma$ in Eq. \eqref{eq:def_Gamma} can be written as
\be \label{eq:Gamma_Fokker_Planck_form}
\Gamma(f)= \frac{D}{4}\nabla_\q \cdot \lp M_{\bar \q_f}\nabla_\q \lp \frac{f}{M_{\bar \q_f}} \rp \rp,
\ee
and we have
\be \label{eq:H_def}
\textnormal{H}(f):= \int_\unitq \Gamma(f)\, \frac{f}{M_{\bar \q_f}}\, \ud\q = -\frac{D}{4}\int_{\unitq} M_{\bar \q_f} \left| \nabla_\q \lp\frac{f}{M_{\bar \q_f}} \rp \right|^2\, \ud\q\, .
\ee
\item[ii)] The equilibria, i.e., the functions $f=f(\q)\ge0$ such that $\Gamma(f)=0$ form a 4-dimensional manifold $\mathcal{E}$ given by
$$\mathcal{E}= \left\{\rho M_{\bar\q}(\q)\, |\, \rho\geq 0, \, \bar\q\in\unitq \right\},$$
where $\rho$ is the macroscopic mass, i.e.,
$$\rho = \int_{\unitq} \rho M_{\bar \q}(\q)\, \ud\q,$$
and $\bar \q$ is the eigenvector corresponding to the maximum eigenvalue of 
$$\int_\unitq \q \otimes \q \, \rho M_{\bar \q}(\q)\, \ud\q  .$$
Furthermore, $\textnormal{H}(f)=0$ iff $f=\rho M_{\bar \q}$ for some $\rho\geq 0$ and $\bar \q \in \unitq$.

\end{itemize}
\end{proposition}

\begin{remark}[Comparison with the equilibria considered in Ref. \cite{bodyattitude}]
\label{rem:equilibria_comparison}
Thanks to Eq. \eqref{eq:equivalence_scalar_product}, one can check that the equilibria $M_{\bar\q}$ represents the same equilibria as for the kinetic model corresponding to the body attitutude model with rotation matrices in Ref. \cite{bodyattitude}, which is given by:
$$M_{\Lambda}(A)=\frac{1}{Z'}\exp\lp d^{-1}(A\cdot \Lambda) \rp, \qquad\mbox{for } \Lambda, A\in SO(3),$$
(where $Z'$ is a normalizing constant), i.e., as long as $\Lambda=\Phi(\bar\q)$ and $A=\Phi(\q)$ ($\Phi$ is defined in Eq. \eqref{eq:def_Phi}), we have $Z'\, M_{\Lambda}(A) = Z M_{\bar \q}(\q)$. Note that the normalizing constants $Z$ and $Z'$ are not equal since the measures chosen on $SO(3)$ and $\unitq$ are identical only up to a multiplicative constant.
\end{remark}

\begin{proof}[Proof of Prop. \ref{prop:properties_Gamma}] 
\mbox{}
\paragraph{Proof of point i)}
Eq. \eqref{eq:Gamma_Fokker_Planck_form} is consequence of the fact that
$$(\bar \q \cdot \q)^2 = \q \cdot \lp \bar \q \otimes \bar \q\rp\q,$$
and that (see Prop. \ref{prop:})
$$\frac{1}{2}\nabla_\q \lp\q \cdot (\bar \q \otimes \bar \q)\q-\frac{1}{4} \rp= P_{\q^\perp}((\bar \q \otimes \bar \q) \q).$$
A computation similar to  \cite[Lemma 4.3] {bodyattitude} allows us to conclude Eq. \eqref{eq:Gamma_Fokker_Planck_form}. Inequality \eqref{eq:H_def} follows from Eq. \eqref{eq:Gamma_Fokker_Planck_form} and the Stokes theorem in $\unitq$. 

\paragraph{Proof of point ii)}
From inequality \eqref{eq:H_def}, we have that if $\Gamma(f)=0$, then $\frac{f}{M_{\bar \q_f}}$ is constant in $\q$. We denote this constant by $\rho$ (which is positive since $f$ and $M_{\bar \q_f}$ are positive).
\medskip
We are left with proving  that $\bar \q$ is the eigenvector corresponding to the maximum eigenvalue of 
$$\int_{\unitq} \q\otimes \q\, M_{\bar \q}(\q)\, \ud\q,$$
(this will not change if multiplied by $\rho$ since it is positive).

For any quaternion $\p_0\in\mathbb{H}$, the left multiplication by $\p_0$, that is $\p\in\mathbb{H}\mapsto \p_0\p \in\mathbb{H}$, is an endomorphism on $\mathbb{H}$. We write $E^l(\p_0)$ the associated matrix, so that for all $\p~\in~\mathbb{H}$, the (quaternion) product $\p_0\p$ is equal to the (matrix) product $E^l(\p_0)\p$.
Using the change of variable $\q'=\bar \q^\ast \q$, we compute
\begin{eqnarray}
\int_\unitq \q \otimes \q M_{\bar \q}(\q)\, \ud\q &=& \int_\unitq (\bar \q\q) \otimes (\bar \q\q) M_{1}(\q)\, \ud\q \nn \\
&=& \int_\unitq E^l(\bar \q) (\q \otimes \q)E^l(\bar \q)^t M_{1}(\q)\, \ud\q \nn\\
&=& E^l(\bar \q) \lp \int_\unitq (\q \otimes \q) M_{1}(\q)\, \ud\q \rp E^l(\bar \q)^t \nn.
\end{eqnarray}
To compute the value of the integral in the term above, first note that $M_1(\q)$ depends only on $\Real \q$. We use a change of variable that switch $q_i$ and $q_j$ (for $i\neq j$) to check that the off-diagonal terms $(i,j)$ and $(j,i)$ are zero. Then we compute the diagonal terms: the zeroth diagonal term is clearly given by $\int_\unitq(\Real(\q))^2 M_{1}(\q)\, \ud\q$, while with the same changes of variable that switch $q_i$ and $q_j$ (for $i\neq j$) we check that the first to third diagonal terms are identical and equal to $\frac{1}{3}\int_\unitq \Ima^2(\q) M_{1}(\q)\, \ud\q$. Using the fact that $\Real^2(\q) + \Ima^2(\q)=1 $, we obtain
\begin{eqnarray}
\int_\unitq \q \otimes \q M_{\bar \q}(\q)\, \ud\q &=& E^l(\bar \q) \lp \int_\unitq \text{diag}[(\Real(\q))^2,\frac{1-(\Real(\q))^2}{3},\dots] M_{1}(\q)\, \ud\q \rp E^l(\bar \q)^t\nn\\
\label{eq:diagonalization1}&=&  E^l(\bar \q) \lp \text{diag}[I^2,\frac{1-I^2}{3},\dots] \rp E^l(\bar \q)^t,
\end{eqnarray}
where we defined
\begin{eqnarray}
I^2:=\int_\unitq (\Real(\q))^2 M_{1}(\q)\, \ud\q >0.
\end{eqnarray}

Note that for any $\p \in\mathbb{H}$, we have $\p^t E^l(\bar \q)^t E^l(\bar \q) \p = \lp E^l(\bar \q) \p\rp^t E^l(\bar \q ) \p = |\bar \q  \p|^2=|\p|^2$. Therefore, $E^l(\bar \q )^t E^l(\bar \q )=\Id$, which implies, since $E^l(\bar \q )$ is invertible (with inverse $E^l(\bar \q ^\ast)$), that $E^l(\bar \q )^t = E^l(\bar \q )^{-1}=E^l(\bar \q ^\ast)$.

Therefore, equality \eqref{eq:diagonalization1} is a diagonalization of the matrix $\int_\unitq \q \otimes \q M_{\bar \q}(\q)\, \ud\q$ in a orthonormal basis. It is direct to check that $\bar \q$ is an eigenvector corresponding to the first eigenvalue $I^2$. It is the maximum eigenvalue, if and only if,
$$ I^2>\frac{1-I^2}{3},$$
that is, if and only if,
\be \label{cdt:I2} I^2 > \frac{1}{4}.
\ee
We compute (using Prop. \ref{prop:volume_element})
\begin{eqnarray}
I^2&=&\frac{\int_\unitq (\Real(\q))^2 \exp{\lp(\Real(\q))^2/d\rp}\, \ud\q}{\int_\unitq \exp{\lp(\Real(\q))^2/d\rp}\, \ud\q}=\frac{\int_0^\pi \cos^2\theta \exp(\cos^2 \theta/d) \sin^2\theta \ud\theta}{\int_0^\pi \exp(\cos^2\theta/d)\sin^2\theta\, \ud\theta}\\
&=&\frac{\int_{[-1,1]} r^2 \exp(r^2/d) (1-r^2)^{1/2} dr}{\int_{[-1,1]} \exp(r^2/d)(1-r^2)^{1/2}\, dr}
 =:\frac{I_1}{I_0},
\end{eqnarray}
and, writing $w(r)=(1-r^2)^{1/2} \exp{\lp r^2/d\rp}$, we have that
\begin{eqnarray}
\frac{\ud}{\ud d}I^2&=& \frac{I_0 \frac{\ud}{\ud d}I_1- I_1 \frac{\ud}{\ud d} I_0}{I_0^2}\\
&=&\frac{-1}{d^2}\frac{1}{I_0^2} \lp  \int_{-1}^1 w(r) \ud r  \int_{-1}^1 r^4 w(r) \ud r - \lp \int_{-1}^1 r^2 w(r) \ud r \rp^2   \rp < 0,
\end{eqnarray}
by Jensen's inequality.

Therefore, we conclude
\begin{eqnarray}
I^2 > \lim_{d\longrightarrow \infty} I^2.
\end{eqnarray}
By the dominated convergence theorem, and using an integration by parts, we have
\begin{eqnarray}
\lim_{d\longrightarrow \infty} I^2 = \frac{\int_0^\pi \cos^2\theta \sin^2\theta \ud \theta}{\int_0^\pi \sin^2\theta \ud \theta} =\frac{1}{4},
\end{eqnarray}
so that \eqref{cdt:I2} holds true, which completes the proof.

\end{proof}

\subsection{Generalised Collision Invariants}
\label{sec:GCI}

\subsubsection{Definition and characterisation}

Consider the rescaled kinetic equation \eqref{eq:f_eps}--\eqref{eq:def_Gamma}. Formally, the limit $f$ of $f^\eps$ as $\eps \to 0$ belongs to the kernel of $\Gamma$ which, by Prop. \ref{prop:properties_Gamma}, means that $f(t,x, \q) = \rho(t,x) M_{\bar\q(t,x)}(\q)$ for some functions $\rho(t,x)\ge0$ and $\bar\q(t,x)\in\unitq$. To obtain the macroscopic equations for $\rho$ and $\bar \q$ we start by looking for conserved quantities of the kinetic equation, i.e., we want to identify functions $\psi=\psi(\q)$ such that
$$\int_{\unitq} \Gamma(f)\psi\, \ud\q =0, \quad \mbox{for all } f.$$
By Prop. \ref{prop:properties_Gamma}, this can be rewritten as
$$0=-\int_\unitq M_{\bar\q_f}\nabla_\q \lp\frac{f}{M_{\bar\q_f}}\rp \cdot \nabla_\q \psi\, \ud\q,$$ 
which particularly holds for $\nabla_\q\psi=0$, i.e., when $\psi$ is a constant. Consequently, we only know one conserved quantity for our model corresponding to the macroscopic mass $\rho$. To obtain the macroscopic equation for $\bar \q$, \emph{a priori} we would need 3 more conserved quantities. To sort out this problem, we use the method of the Generalized Collision Invariants (GCI) introduced in Ref. \cite{degond2008continuum}.

\paragraph{Definition of the GCI.}

Define the operator 
$$C(f, \bar \q)=\nabla_\q \cdot \lp M_{\bar\q}\nabla_\q \lp\frac{f}{M_{\bar \q}} \rp \rp,$$
for a function $f$ and $\bar \q\in\unitq$. Notice that
$$\Gamma (f)=C(f, \bar \q_f).$$
\begin{definition}\label{definition:GCI}[Generalised Collision Invariant] A function $\psi\in H^1(\unitq)$ is a generalised collision invariant (GCI) associated with $\bar \q \in \unitq$ if and only if
$$\int_\unitq C(f, \bar \q) \psi\, \ud\q=0, \, \mbox{ for all } f\mbox{ such that } P_{\bar \q^\perp}\left[\int (\q\otimes\q) f(\q) \, \ud\q\, \bar\q\right]=0.$$
We write $GCI(\bar\q)$ the set of GCI associated with $\bar \q$.
\end{definition}
If $\psi$ exists for any given $\bar \q \in \unitq$, consider particularly $\psi_{\bar \q_{f^\eps}}$, the GCI associated with $\bar \q_{f^\eps}$ given by Eq. \eqref{eq:def_qf}. It holds that 
\be \label{eq:GCIworks}
\frac{1}{\eps}\int_\unitq \Gamma(f^\eps)\psi_{\bar \q_{f^\eps}}\, \ud\q = \frac{1}{\eps}\frac{D}{4}\int_\unitq C(f^\eps, \bar \q_{f^\eps}) \psi_{\bar \q_{f^\eps}}\, \ud\q=0,
\ee
 since
\beqar
P_{\bar \q_{f^\eps}^\perp}\left[\int (\q\otimes\q) f^\eps \, \ud\q\, \bar\q_{f^\eps}\right]=P_{\bar \q_{f^\eps}^\perp}(\lambda^\eps_{max} \bar \q_{f^\eps})=0.
\eeqar
Therefore, after multiplying the kinetic equation \eqref{eq:f_eps} by $\psi_{\bar \q_{f^\eps}}$ and integrating on $\unitq$, the right hand side is of order $\eps$.

\paragraph{Characterisation of the GCI.}
The main result of this section is the following description of the set of GCI.
\begin{proposition}[Description of the set of GCI] \label{prop:set_GCI}
Let $\bar \q\in\unitq$. Then
\be
GCI(\bar \q) = \text{span} \left\{ 1,\, \cup_{\beta\in{\bar\q}^\perp} \psi^\beta\right\},
\ee
where, for $\beta\in{\bar\q}^\perp$, the function $\psi^\beta$ is defined by
\be
\psi^\beta(\q):=(\beta\cdot \q) \,h(\q\cdot\bar\q), \label{eq:psi_beta}
\ee
with $h=h(r)$ the unique solution of the following differential equation on $(-1,1)$:
\be\begin{split} \label{eq:ode_h} (1-r^2)^{3/2} \exp\lp\frac{2 r^2}{d}\rp \lp\frac{-4}{d}r^2-3\rp h(r)+ \frac{\ud}{\ud r} \left[  (1-r^2)^{5/2} \exp\lp\frac{2 r^2}{d}\rp h'(r) \right]\\
=r\, (1-r^2)^{3/2}  \exp\lp\frac{2 r^2}{d}\rp. \end{split}\ee
Furthermore, the function $h$ is \emph{odd}: $h(-r)=-h(r)$, and it satisfies for all $r\ge0$, $h(r)\le0$.
\end{proposition}
This proposition will be crucial to compute the hydrodynamical limit in Sec. \ref{sec:macro_limit}. The proof is done in the two subsections below.

\subsubsection{Existence and uniqueness of GCI}
We prove here the
\begin{proposition}[First characterisation of the GCI] \label{prop:equivalent_definition_GCI}
Let $\bar \q\in\unitq$. We have that $\psi \in GCI(\bar \q)$ if and only if 
\begin{equation} \label{eq:for_psi}
\text{there exists }\beta\in \bar \q^\perp\text{ such that }\nabla_\q \cdot (M_{\bar \q}\nabla_\q \psi) = (\beta \cdot \q)(\bar \q\cdot \q)\, M_{\bar \q}.
\end{equation}

\end{proposition}

\begin{proof}[Proof of Prop.~\ref{prop:equivalent_definition_GCI}]
We denote by $\mathcal{L}$ the linear operator $C(\cdot,\bar \q)$ on $L^2(\unitq)$, and~$\mathcal{L}^*$ its adjoint. We have the following sequence of equivalences, starting from Definition \ref{definition:GCI}:
\begin{align*}
&\psi\in GCI(\bar \q)\Leftrightarrow \int_{\unitq}\psi \mathcal{L}(f) \, \ud\q =0, \quad \mbox{for all } f\mbox{ such that } P_{\bar \q^\perp}\lp \left[\int_\unitq \q\otimes \q f(\q)\, \ud\q \right]\bar\q\rp=0\\
&\Leftrightarrow  \int_{SO(3)}\mathcal{L}^*(\psi) f \, \ud\q=0, \mbox{ for all } f\mbox{ s.t. } \forall \beta\in \bar \q^\perp, \beta \cdot \left[\int_\unitq \q\otimes \q f(\q)\, \ud\q \right]\bar\q=0\\
&\Leftrightarrow  \int_{SO(3)}\mathcal{L}^*(\psi) f \, \ud\q=0, \quad \mbox{for all } f\in \mathscr{F}_{\bar \q}^\perp\\
&\Leftrightarrow  \mathcal{L}^*(\psi) \in \lp \mathscr{F}^\perp_{\bar \q}\rp^\perp,
\end{align*}
where
$$\mathscr{F}_{\bar \q}:= \left\{ f:\unitq\to \R, \mbox{ with } f(\q)=(\beta\cdot \q)(\bar \q \cdot \q), \, \mbox{ for some } \beta\in \bar \q^\perp\right\},$$
and~$\mathscr{F}_{\bar \q}^\perp$ is the space orthogonal to~$\mathscr{F}_{\bar \q}$ in~$L^2(\unitq)$. Note that $\mathscr{F}_{\bar\q}$ is a vector subspace of $L^2$ isomorphic to $\bar \q^\perp$: indeed, if for some $\beta \in \bar \q^\perp$ we have that $f(\q)=(\beta\cdot \q)(\bar \q \cdot \q)=0$ for all $\q\in\mathbb{H}$, then $\beta\cdot \q=0$ for all $\q\in\mathbb{H}\setminus{\bar \q}^\perp$, so that by continuity and density it is also true for all $\q\in\mathbb{H}$, which finally implies $\beta=0$. Therefore, $\mathscr{F}_{\bar \q}$ is closed (finite dimensional of dimension 3), and we have $(\mathscr{F}_{\bar \q}^\perp)^\perp=\mathscr{F}_{\bar \q}$. Therefore we get
\[\psi\in GCI(\bar \q)\Leftrightarrow  \mathcal{L}^*(\psi) \in  \mathscr{F}_{\bar \q}\Leftrightarrow  \mbox{ there exists } \beta \in \bar \q^\perp\mbox{ such that } \mathcal{L}^*(\psi)(\q)= (\beta\cdot \q)(\bar \q \cdot \q) ,
\]
which ends the proof since the expression of the adjoint is~$\mathcal{L}^*(\psi)=\frac1{M_{\bar \q}}\nabla_\q\cdot (M_{\bar \q}\nabla_{\q}\psi)$.
\end{proof}

We now verify that Eq. \eqref{eq:for_psi} in Prop. \ref{prop:equivalent_definition_GCI} has a unique solution in the space
\begin{equation*}
H^1(\unitq) := \left\{ \psi: \unitq \to \mathbb{R} \, \left|\, \int_\unitq |\psi|^2 \,\ud \q+ \int_\unitq |\nabla_{\q}\psi|^2 \,\ud \q < \infty \right.\right\}.
\end{equation*}

\begin{proposition}[Existence and uniqueness of the GCI] \label{prop:exist_unique_GCI}
Let $\bar \q\in \unitq$ and let $\beta\in\bar\q^\perp$. Then, Eq. \eqref{eq:for_psi} has a unique solution $\psi$ (up to an additive constant) in $H^1(\unitq)$.
\end{proposition}

\begin{proof}
To prove this proposition, we rewrite Eq. \eqref{eq:for_psi} in its weak formulation as
\begin{equation}\label{eq:lax-milgram}
\int_\unitq M_{\bar \q} \nabla_\q \psi \cdot \nabla_\q \varphi \,\ud \q = - \beta \cdot \int_{\unitq} \q (\q \cdot \bar \q) \varphi M_{\bar\q}\,\ud \q,
\end{equation}
for all test functions $\varphi$ in $H^1(\unitq)$. Denote by $H^1_0(\unitq)$ the set of zero-mean functions in $H^1(\unitq)$, i.e.,
\begin{equation*}
H^1_0(\unitq) = \left\{ \psi \in H^1(\unitq)\left| \int_\unitq \psi \,\ud \q=0 \right.\right\}.
\end{equation*}
Note that thanks to the Poincar\'e inequality on the sphere $\unitq$, the usual semi-norm on $H^1(\unitq)$ given by $\psi\mapsto \int_\unitq |\nabla_\q \psi |^2 \,\ud\q$ is a norm on $H^1_0(\unitq)$. The weak formulation \eqref{eq:lax-milgram} is equivalent on $H^1(\unitq)$ and on $H^1_0(\unitq)$: indeed, if $\psi$ is a solution of \eqref{eq:lax-milgram} on $H^1_0(\unitq)$, then by a change of variable $\q':=\bar \q^\ast \q$ (and using $\beta \in \bar\q^\perp$),
$$ \beta \cdot \int_{\unitq} \q (\q \cdot \bar \q) M_{\bar\q}\,\ud \q = \beta \cdot \bar \q \int_{\unitq} \q (\Real \q) M_{1}(\Real \q)\,\ud \q =0,$$
so that Eq. \eqref{eq:lax-milgram} is also satisfied on the set of constant functions $\varphi$, and, by linearity, $\psi$ solves Eq. \eqref{eq:lax-milgram} on $H^1(\unitq)$.

We want to apply Lax-Milgram's theorem to Eq. \eqref{eq:lax-milgram} in the Hilbert space $H^1_0(\unitq)$. The left-hand side of Eq. \eqref{eq:lax-milgram} is a bilinear operator in $(\psi,\varphi)\in(H^1_0(\unitq))^2$, which is continuous (thanks to a Cauchy-Schwarz inequality, using furthermore the fact that $M_{\bar\q}$ is upper bounded pointwise on $\unitq$), and coercive (by definition of the norm on $H^1_0(\unitq)$, and thanks to the fact that $M_{\bar\q}$ is lower bounded pointwise on $\unitq$). The right-hand side of Eq. \eqref{eq:lax-milgram} is a linear form in $\varphi\in H^1_0(\unitq)$, which is continuous (thanks to Cauchy-Schwarz inequality, using again the pointwise upper bound for $M_{\bar\q}$).

We can therefore apply Lax-Milgram's theorem, which guarantees the existence of a unique solution $\psi_0\in H^1_0(\unitq)$ of Eq. \eqref{eq:for_psi}. We conclude noticing that any function $\psi$ in $H^1(\unitq)$ is a solution of Eq. \eqref{eq:for_psi} if and only if its zero-mean projection $\psi_0:=\psi - \frac{1}{2\pi^2} \int_\unitq \psi \,\ud \q$ is also a solution of Eq. \eqref{eq:for_psi}.

\end{proof}

From all this we conclude the following

\begin{corollary}\label{cor:set_GCI}
Let $\bar \q\in\unitq$. Then
\be
GCI(\bar \q) = \text{span} \left\{ 1,\, \cup_{\beta\in{\bar\q}^\perp} \tilde\psi^\beta\right\},
\ee
where, for $\beta\in{\bar\q}^\perp$, the function $\tilde\psi^\beta$ is the unique solution in $H^1_0(\unitq)$ of Eq. \eqref{eq:for_psi}.
\end{corollary}

\begin{remark}
The linear mapping $\beta\in{\bar\q}^\perp \mapsto \tilde\psi^\beta \in H^1_0(\unitq)$ is injective (by Prop. \ref{prop:exist_unique_GCI}). Therefore, $GCI(\bar \q)$ is a 4-dimensional vector space.
\end{remark}

\subsubsection{The non-constant GCIs}\label{sec:non_constant_GCI}

\begin{proposition}
\label{prop:non constant GCI}
Let $\bar \q\in\unitq$. Let $\psi$ be a function of the form
$$\psi(\q) = (\beta \cdot \q) h(\q \cdot \bar \q),$$
for some $\beta \in \bar \q^\perp$ and some smooth ($C^2$) scalar function $h$.
Then $\psi$ is solution of Eq. \eqref{eq:for_psi} in $H^1_0(\unitq)$ if and only if $h$ is a solution of Eq. \eqref{eq:ode_h}. Furthermore, the solution $h$ of Eq. \eqref{eq:ode_h} exists and is unique, is an odd function, and satisfies for all $r\ge0$, $h(r)\le0$.

\end{proposition}
\begin{proof}
Eq. \eqref{eq:for_psi} is equivalent to
\be \label{eq:aux_prop_GCI}
\nabla_\q (\log M_{\bar \q}) \cdot \nabla_\q \psi + \Delta_\q \psi = (\beta \cdot \q) (\q \cdot \bar \q),
\ee
where we compute
$$\nabla_\q (\log M_{\bar \q}) = \frac{2}{d}\nabla_\q (\q \cdot \bar \q)^2.$$
Next, we substitute $\psi = (\beta \cdot \q) h(\q\cdot \bar \q)$ into Eq. \eqref{eq:aux_prop_GCI}. To carry out the computations we will use the following expressions:
\beqar
\nabla_\q \psi &=& \lp \nabla_\q (\beta \cdot \q)\rp h(\q \cdot \bar \q) + (\beta \cdot \q) h'(\q \cdot \bar \q) \nabla_\q (\q \cdot \bar \q),\\
\Delta_\q [h(\q \cdot \bar \q)] &=&  \nabla_\q \cdot (\nabla_\q [h(\q \cdot \bar \q)]) = \nabla_\q \cdot (h'(\q \cdot \bar \q) \nabla_\q (\q \cdot \bar \q) )\\
&=& h''(\q \cdot \bar \q) |\nabla_\q (\q \cdot \bar \q)|^2 + h'(\q \cdot \bar \q) \Delta_\q (\q \cdot \bar \q),\\
\Delta_\q \psi &=& \Delta_\q (\beta \cdot \q ) h(\q \cdot \bar \q) + 2 \nabla_\q (\beta \cdot \q) \cdot \nabla_\q [h(\q \cdot \bar \q)] + (\beta \cdot \q) \Delta_\q [h(\q \cdot \bar \q)]\\
&=& \Delta_\q (\beta \cdot \q) h(\q \cdot \bar \q) + 2 \nabla_\q (\beta \cdot \q) \cdot \nabla_\q (\q \cdot \bar \q) h'(\q \cdot \bar \q) \\
&&+ (\beta \cdot \q) \left[ h''(\q \cdot \bar \q) |\nabla_\q (\q \cdot \bar \q)|^2 + h'(\q \cdot \bar \q) \Delta_\q (\q \cdot \bar \q) \right].
\eeqar
Substituting the previous expressions in Eq. \eqref{eq:aux_prop_GCI}, and grouping terms, we obtain that $\psi$ satisfies Eq. \eqref{eq:for_psi} if and only if
\beqarl
&&\left\{\frac{2}{d} \nabla_\q (\q\cdot \bar \q)^2 \cdot \nabla_\q (\beta \cdot \q) + \Delta_\q (\beta \cdot \q) \right\} h(\q \cdot \bar \q)\nn\\
&&+ \left\{\frac{2}{d}\nabla_\q(\q\cdot\bar \q)^2 \cdot \nabla_\q(\q\cdot\bar \q) (\beta \cdot \q) + 2 \nabla_\q (\beta \cdot \q) \cdot \nabla_\q (\q \cdot \bar \q) + \Delta_\q (\q\cdot \bar \q) (\beta \cdot \q) \right\} h'(\q \cdot \bar \q)\nn\\
&&+ \left\{|\nabla_\q (\q \cdot \bar \q)|^2 (\beta \cdot \q) \right\} h''(\q \cdot \bar \q)\nn\\
&&= (\beta \cdot \q) (\q \cdot \bar \q). \label{eq:aux_prop_GCI2}
\eeqarl
To compute this expression we will use the following identities:
\beqar
\nabla_\q (\q \cdot \bar \q)\cdot \nabla_\q(\beta \cdot  \q) &=& P_{\q^\perp}(\bar \q) \cdot P_{\q^\perp}(\beta) = - (\q \cdot \beta) (\q \cdot \bar \q),\\
\nabla_\q (\q \cdot \bar \q)^2 &=& 2 (\q \cdot \bar \q) P_{\q^\perp}\bar \q,\\
\nabla_\q (\q \cdot \bar \q)^2\cdot \nabla_\q(\beta \cdot \q) &=& -2(\q \cdot \bar \q)^2 (\beta \cdot \q),\\
\nabla_\q (\q \cdot \bar \q)^2 \cdot \nabla_\q (\q \cdot \bar \q) &=& 2 (\q \cdot \bar \q) \left[ 1- (\q \cdot \bar \q)^2 \right],\\
\Delta_\q(\beta \cdot \q) &=& -3 (\beta \cdot \q),\\
|\nabla_\q(\q\cdot \bar \q)|^2 &=& |P_{\q^\perp}(\bar \q)|^2 = 1- (\q \cdot \bar \q)^2,
\eeqar
where we used that $\beta \cdot \bar \q =0$ and that in the sphere $\mathbb{S}^3$ it holds $\Delta_\q(\mathbf{p} \cdot \q)= -3(\mathbf{p} \cdot \q)$. Substituting the previous expressions in Eq. \eqref{eq:aux_prop_GCI2} we obtain that $\psi$ solves Eq. \eqref{eq:for_psi} if and only if
\beqar
&&\left\{ -\frac{4}{d}(\beta \cdot \q) (\q \cdot \bar \q)^2 -3(\beta \cdot \q)\right\} h(\q \cdot \bar \q)\\
&& + \left\{ \frac{4}{d}(\beta \cdot \q)  (\q \cdot \bar \q) \left[ 1- (\q \cdot \bar \q)^2 \right] -3 (\beta \cdot \q) (\q \cdot \bar \q) -2 (\beta \cdot \q) (\q \cdot \bar \q)\right\} h'(\q \cdot \bar \q)\\
&& + \left\{(\beta \cdot \q) (1-(\q \cdot \bar \q)^2) \right\} h''(\q \cdot \bar \q)\\
&& = (\beta \cdot \q) (\q\cdot \bar \q).
\eeqar
When $\q$ ranges in $\unitq$, $r:= (\q \cdot \bar \q)$ ranges in $[-1,1]$. Therefore the previous equality can be rewritten (after factorizing out and simplifying the terms $(\beta \cdot \q$), using a continuity argument) as an equation in $r\in[-1,1]$:
\be\label{eq:ode_h_alt}\lp \frac{-4}{d}r^2-3\rp h+ \lp\frac{4}{d}(1-r^2)-5 \rp\,r\, h' + (1-r^2) h''=r. \ee

Finally, we recast this equation as shown in \eqref{eq:ode_h}.

\medskip

We define the Hilbert space
\beqar
H_{(-1,1)}:= \Bigg\{h:(-1,1)\longrightarrow\R, &\mbox{such that} &\int_{-1}^1 (1-r^2)^{3/2} h^2(r)\, \ud r <\infty \mbox{ and }\\
&&\int_{-1}^1 (1-r^2)^{5/2} \lp h'(r)\rp^2\, \ud r<\infty \Bigg\}.
\eeqar

By a Lax-Milgram argument, we obtain the existence and uniqueness of the solution $h$ in $H_{(-1,1)}$. By uniqueness of $h$, we see that $h$ is an odd function of its argument. By a maximum principle, we furthermore obtain that $h(r)\le0$ for $r\ge0$.

To conclude it only remains to show that the solution $h$ corresponds to a function $\psi\in H_0^1(\unitq)$. Since we know by Prop. \ref{prop:exist_unique_GCI} that the GCI exists and it is unique in $H_0^1(\unitq)$, this proves that $\psi$ given by Eq. \eqref{eq:psi_beta} is the GCI. For that, we first compute the $L^2(\unitq)$ norm of gradient of $\psi$ with
\beqar
\int_\unitq |\nabla_\q \psi |^2 \ud \q &\le & 2 \int_\unitq |P_{\q^\perp}\beta|^2 h^2(\q \cdot \bar \q) \ud \q + 2 \int_\unitq (\beta \cdot \q)^2 |P_{\q^\perp}\bar \q|^2 (h'(\q \cdot \bar \q))^2 \ud \q\\
&= & 2 \left[\int_\unitq |P_{(\q \bar \q)^\perp}\beta|^2 h^2(\Real \q) \ud \q + \int_\unitq (\beta \cdot (\q \bar \q))^2 |P_{(\q \bar \q)^\perp}\bar \q|^2 (h'(\Real \q))^2 \ud \q\right]\\
&= & 2 \left[\int_\unitq |P_{\q^\perp}(\beta\bar \q^\ast)|^2 h^2(\Real \q) \ud \q + \int_\unitq (\beta \bar \q^\ast \cdot \q)^2 (1-\Real^2\q) (h'(\Real \q))^2 \ud \q\right]\\
&= & 2 (\beta\bar \q^\ast)\cdot \int_\unitq (\Id - \q\otimes \q) h^2(\Real \q) \ud \q \,(\beta\bar \q^\ast)\\
&&+ 2(\beta\bar \q^\ast)\cdot \int_\unitq (\q\otimes \q) (1-\Real^2\q) (h'(\Real \q))^2 \ud \q \, (\beta\bar \q^\ast).
\eeqar
We see directly that a sufficient condition for the first term of the last expression above to be finite is that
\be \label{cdt:h2finite}
 \int_\unitq h^2(\Real \q) \ud \q <\infty.
\ee
Since $\Real(\beta \bar\q^\ast) =0$ (remember that by definition $\beta\in\bar\q^\perp$), the second term is finite as soon as the $3\times 3$ submatrix corresponding to the imaginary coordinate of the integral $\int_\unitq (\q\otimes \q) (1-\Real^2\q) (h'(\Real \q))^2 \ud \q$ is finite, that is, as soon as
\beqar
\int_\unitq (\Ima\q\otimes\Ima \q) (1-\Real^2\q) (h'(\Real \q))^2 \ud \q <\infty,
\eeqar
that is, when the diagonal terms are finite (the off-diagonal elements being null by the changes of variables which change the sign of the coordinate $q_i$ for $i=1,\,2,\,3$), i.-e.
\beqar
 \int_\unitq q_i^2 (1-\Real^2\q) (h'(\Real \q))^2 \ud \q <\infty.
\eeqar
Summing for $i=1,\,2,\,3$ (all terms being nonnegative), this is true when
\be\label{cdt:dh2finite}
\int_\unitq (1-\Real^2\q)^2 (h'(\Real \q))^2 \ud \q <\infty.
\ee
After a change of variable $r=\cos(\theta/2)=\Real \q$, using Prop. \ref{prop:volume_element}, conditions \eqref{cdt:h2finite} and \eqref{cdt:dh2finite} are rewritten as
\be\label{cdt:H10}
 \int_{-1}^1 (1-r^2)^{1/2} h^2(r) \ud r <\infty \quad \text{and} \quad \int_{-1}^1 (1-r^2)^{5/2} (h'(r))^2 \ud r <\infty.
\ee
Since $h$ is in $H_{(-1,1)}$, the second condition is true. Let us check the first condition. Using that $h$ is in $H_{(-1,1)}$, we have that $h\in H^1(a,b)$ for all $-1<a<b<1$.
By a Sobolev injection, this implies that $h$ is continous on $(-1,1)$. Now, since $h$ is an odd and continuous function on $(-1,1)$, to obtain the first condition of \eqref{cdt:H10} it is enough to show that
\beqar
I_h&:=& \int_{1/2}^1 r(1-r^2)^{1/2} h^2(r) \ud r <\infty.
\eeqar
We compute, for some $\delta\in(0,1/2),$ using an integration by parts and the inequality $2 ab\le a^2 + b^2$ for real numbers $a$ and $b$,
\beqar
I_h(\delta)&:=& \int_{1/2}^{1-\delta} r(1-r^2)^{1/2} h^2(r) \ud r\\
  &=& \left[ -\frac{1}{3} (1-r^2)^{3/2} h^2 \right]_{1/2}^{1-\delta} + \int_{1/2}^{1-\delta} \frac{2}{3} (1-r^2)^{3/2} h h'\\
 &\le& \frac{1}{3} \lp\frac{3}{4}\rp^{3/2} h^2(1/2) + \int_{1/2}^{1-\delta} \frac{2}{3} (1-r^2)^{1/4} h (1-r^2)^{5/4}h'\\
 &\le& \frac{\sqrt{3}}{8} h^2(1/2) + \int_{1/2}^{1-\delta} \frac{1}{3} (1-r^2)^{1/2} h^2+ \int_{1/2}^{1-\delta} \frac{1}{3} (1-r^2)^{5/2}(h')^2\\
  &\le& \frac{\sqrt{3}}{8} h^2(1/2) +\frac{2}{3}I_h(\delta)+ \int_{1/2}^{1} \frac{1}{3} (1-r^2)^{5/2}(h')^2.
\eeqar
Therefore, taking the limit $\delta\longrightarrow0$, the integral on $(0,1)$ is finite: $I_h=I_h(0)<\infty$.
This proves conditions \eqref{cdt:H10}, so that $\nabla_\q\psi\in L^2(\unitq)$. Note that we have proved in particular that
\beqar
 (1-r^2)^{5/4} h \in H^1(-1,1).
 \eeqar
By a similar computation as for $\nabla_\q\psi$, we see that
\beqar
\int_\unitq  \psi  \ud \q &= & \int_\unitq (\beta\cdot\q) \,h(\q \cdot \bar \q) \ud \q \\
&=& (\beta\bar \q^\ast)\cdot \int_\unitq \q\,h(\Real \q) \ud \q = \Real\lp\beta\bar\q^\ast\rp \int_\unitq \Real \q\,h(\Real \q) \ud \q
\eeqar
which is null since $\beta\in\bar\q^\perp$, so that $\psi$ has mean zero on $\unitq$.

\end{proof}

We are now ready to prove Prop. \ref{prop:set_GCI}.

\begin{proof}[Proof of Prop. \ref{prop:set_GCI}]
The statement is a direct consequence of Prop. \ref{prop:exist_unique_GCI}, Cor.  \ref{cor:set_GCI} and Prop. \ref{prop:non constant GCI}.

\end{proof}

\subsection{The macroscopic limit}
\label{sec:macro_limit}

This section is devoted to the proof of Th. \ref{th:macro_limit}. We will use the following:
\begin{lemma} 
\label{lem:computation integral continuity equation}
It holds that
\be \label{eq:continuity_integral}
\int_\unitq  \vezero(\q)M_{\bar \q}(\q) \ud \q = c_1 \vezero(\bar\q),
\ee
where the positive constant $c_1$ is given in Eq. \eqref{eq:c1}.
\end{lemma}
\begin{remark}[Comments on the constant $c_1$] \mbox{}
\label{rem:constant_c1}
\begin{itemize}
\item[i)] The value for the constant $c_1$ obtained here is the same one as in the body attitude coordination model based on rotation matrices in Ref. \cite{bodyattitude}. This will allow us to prove the equivalence between the respective continuity equations (see Sec. \ref{sec:comparison_macro}).
\item[ii)] In the case of the Vicsek model in Ref. \cite{degond2008continuum} and in the body attitude coordination model based on rotation matrices in Ref. \cite{bodyattitude}), the constant $c_1$ played a role of `order parameter'. Particularly, it holds that $c_1\in [0,1]$ and the larger its value, the more organised (coordinated/aligned) the dynamics are (and the other way around, the smaller $c_1$, the more disordered the dynamics are). The extreme cases take place, for example, when $D\to \infty$, and then $c_1=0$, and when $D\to\infty$, giving $c_1=1$. Here we have the same properties and interpretations for $c_1$. 
\end{itemize}
\end{remark}

\begin{proof}[Proof of Lem. \ref{lem:computation integral continuity equation}]
We first make the change of variable $\q'=\q \bar \q^\ast$:
\beqar
\int_\unitq  \vezero(\q)M_{\bar \q}(\q) \ud \q &=& \int_\unitq  \vezero(\q\bar\q)M_{1}(\Real(\q)) \ud \q,
\eeqar
where, for $\q=\cos(\theta/2)+\sin(\theta/2)\nvec\in \unitq$,
\beqarl 
M_1(\Real(\q)) &=& \frac{1}{Z}\exp\lp\frac{2}{d}\lp (1\cdot \q)^2-1/4\rp\rp= \frac{1}{Z}\exp\lp \frac{2}{d}\lp \Real(\q)^2-1/4\rp\rp \nn\\
&=& \frac{1}{Z}\exp\lp\frac{2}{d}\lp\cos^2(\theta/2)-\frac{1}{4}\rp\rp= \frac{1}{Z}\exp\lp d^{-1}\lp\frac{1}{2}+\cos\theta\rp \rp.\label{eq:equilibrium_change_variables}
\eeqarl
 
Then, defining $\mathbf{\bar{e}_1}=\bar \q \vezero \bar \q^\ast$, we decompose
\be \begin{split}
\vezero(\q\bar\q) &= \Ima (\q \mathbf{\bar{e}_1} \q^\ast)\\
&= \lp2\Real^2 (\q)-1\rp \Ima (\mathbf{\bar{e}_1}) + 2 \Real(\q) \lp \Ima(\q) \times \Ima (\mathbf{\bar{e}_1}) \rp + 2 \lp \Ima(\q) \otimes \Ima(\q)\rp \Ima (\mathbf{\bar{e}_1}) \\
&= \lp2\Real^2 (\q)-1\rp \vezero(\bar\q) + 2 \Real(\q) \lp\Ima(\q) \times \vezero(\bar\q) \rp + 2 \lp \Ima(\q) \otimes \Ima(\q)\rp \vezero(\bar\q),
 \label{eq:e_1_bar} \end{split} \ee
 where we used that for $\q, \rvec\in\unitq$ it holds that $\Real(\q\rvec) = \Real(\rvec)\Real(\q) -\Ima(\q)\cdot\Ima(\rvec)$ and $\Ima(\q\rvec)= \Real(\q)\Ima(\rvec)+\Real(\rvec)\Ima(\q)+\Ima(\q)\times\Ima(\rvec)$.
 
We integrate against $M_1(\Real(\q))$: by arguments of parity the contribution of the second term vanishes (with a change of variable $\q'=\q^\ast$), and the contribution of the last term is diagonal (with the changes of variable which change the sign of $q_i$ for $i=1,\,2,\,3$), so that
\beqar
\int_\unitq  \vezero(\q)M_{\bar \q}(\q) \ud \q &=& \int_\unitq  \left[ \lp2\Real^2 (\q)-1\rp \Id + 2 \text{diag}(q_1^2,\,q_2^2,\,q_3^2) \right]  \, \vezero(\bar\q) \, M_{1}(\Real(\q)) \ud \q.
\eeqar
Using the changes of variable that switch the coordinates $q_i$ and $q_j$ for $i\neq j$, we see that the diagonal elements corresponding to $q_i^2$ for $i=1,2,3$ give rise to the same value for the integral, therefore
\beqar
\int_\unitq  \vezero(\q)M_{\bar \q}(\q) \ud \q &=& \int_\unitq  \left[ \lp2\Real^2 (\q)-1\rp \Id + \frac{2}{3} |\Ima (\q)|^2 \Id \right]  \, \vezero(\bar\q) \, M_{1}(\Real(\q)) \ud \q\\
&=& \int_\unitq  \left[ \lp2\Real^2 (\q)-1\rp + \frac{2}{3} (1-\Real^2(\q)) \right]  \, \vezero(\bar\q) \, M_{1}(\Real(\q)) \ud \q,
\eeqar
so that the equality \eqref{eq:continuity_integral} holds for
\beqar
c_1&=&\frac{1}{3} \lp \int_\unitq  \lp 4 \lp \Real(\q)\rp^2 -1 \rp M_1(\Real (\q)) \ud \q \rp\\
&=& \frac{2}{3} 4\pi \int_{0}^\pi  \lp \frac{1}{2}+\cos\theta\rp \frac{m(\theta)}{Z} \sin^2(\theta/2) \ud \theta \\
&=& \frac{2}{3} \la \frac{1}{2}+\cos\theta\ra_{m \sin^2(\theta/2)},
\eeqar
where we used Prop. \eqref{prop:volume_element} on the volume element.
\end{proof}

We are now ready to prove Th. \ref{th:macro_limit}.

\begin{proof}[Proof of Th. \ref{th:macro_limit}]
By Eq. \eqref{eq:f_eps}, we have that $\Gamma(f^\eps)=\mathcal{O}(\eps)$. Formally, the limit of $f^\eps$ as $\epsilon \to 0$ (if the limit exists) is in the kernel of $\Gamma$. Therefore, by Prop. \ref{prop:properties_Gamma}, the limit has the form 
\beqarl
\label{ansatz}
f(t,x,\q) = \rho(t,x) M_{\bar \q(t,x)}(\q),
\eeqarl
for some $\rho=\rho(t,x)\geq 0$ and $\bar \q=\bar\q(t,x)\in\unitq$.
We integrate the kinetic Eq. \eqref{eq:f_eps} on $\unitq$ to obtain
\beqarl
\pa_t \rho^\eps + \nabla_x \cdot \left(\int_\unitq   \vezero(\q) f^\eps \ud \q\right)  &=& \mathcal{O}(\eps). \label{eq:eps_continuity}
\eeqarl
Taking the limit $\eps\to 0$ and substituting the value of $f$ with expression \eqref{ansatz}, we obtain the equation
\beqarl \label{comp:eps0}
\pa_t \rho + \nabla_x \cdot \left(   \rho(t,x) \int_\unitq  \vezero(\q)M_{\bar \q(t,x)}(\q) \ud \q\right) = 0.
\eeqarl
Lem. \ref{lem:computation integral continuity equation} gives us the value of the integral in the previous expression, from which we conclude the continuity equation \eqref{eq:continuity_equation}.

\bigskip

We compute next the evolution equation for $\bar \q = \bar\q(t,x)$. We multiply the rescaled kinetic equation \eqref{eq:f_eps} by the GCI $\psi$ associated with $\bar \q_{f^\eps}$, that is, by Prop. \ref{prop:set_GCI}, 
$$\psi(\q) = (\beta \cdot \q)\, h(\q \cdot \bar \q_{f^\eps}), \quad \mbox{for } \beta \in \bar\q^\perp,$$
and integrate over $\unitq$. We obtain (using Eq. \eqref{eq:GCIworks}) 
$$\int_\unitq \left[ \partial_t f^\eps + \nabla_x \cdot ( \vezero(\q) f^\eps) \right] (\beta \cdot \q)\, h(\q \cdot \bar \q_{f^\eps})\, \ud\q = \mathcal{O}(\eps).$$
Making $\eps \to 0$ and using that (formally) $\bar\q_{f^\eps} \to \bar \q$, the previous expression gives:
$$\int_\unitq \left[ \partial_t (\rho M_{\bar \q}) + \nabla_x \cdot (  \vezero(\q) \rho M_{\bar \q}) \right] (\beta \cdot \q)\, h(\q \cdot \bar \q)\, \ud\q =0 \quad \mbox{for all } \beta \in \bar \q^\perp.$$
Particularly, this implies that
$$\beta\cdot Y = 0, \quad \mbox{for all }\beta \in \bar\q^\perp,$$
for
$$Y= \int_\unitq \left[ \partial_t (\rho M_{\bar \q}) + \nabla_x \cdot (  \vezero(\q) \rho M_{\bar \q}) \right] \, h(\q \cdot \bar \q)\, \q\, \ud\q.$$
This is equivalent to
$$X:= P_{\bar\q^\perp}Y=P_{\bar \q^\perp} \int_\unitq \left[ \partial_t (\rho M_{\bar \q}) + \nabla_x \cdot (\vezero(\q) \rho M_{\bar \q}) \right]  h(\q \cdot \bar \q)\, \q\, \ud\q =0. $$

Next we compute each term in the previous expression. We have that
$$\partial_t (\rho M_{\bar \q}) = (\partial_t \rho) M_{\bar \q}+ \rho M_{\bar \q} \frac{4}{d}(\q \cdot \bar \q) (\q \cdot \partial_t \bar \q),$$
and
$$\nabla_x \cdot (  \vezero(\q) \rho M_{\bar \q})=   \vezero(\q) \cdot \lp(\nabla_x \rho)\, M_{\bar\q} + \rho M_{\bar \q} \frac{4}{d}(\q \cdot \bar \q)\nabla_x (\q \cdot \bar \q) \rp.$$
We define
\be \label{def:bigH} H(\q \cdot \bar \q) := M_{\bar \q}(\q) h(\q \cdot \bar \q) (\q \cdot \bar \q),
\ee
which is even in its argument (recall that $h$ is odd).
We have that
$$X=: X_1+X_2+X_3+X_4=0,$$
where
\beqar
X_1 &=& P_{\bar \q^\perp} \int_{\unitq} (\partial_t \rho) M_{\bar \q} h(\q \cdot \bar \q)\, \q\, \ud\q ,\\
X_2 &=& P_{\bar \q^\perp} \int_{\unitq}\rho M_{\bar \q} \frac{4}{d}(\q \cdot \bar \q) (\q \cdot \partial_t \bar \q)  h(\q \cdot \bar \q)\, \q\, \ud\q \\
&=& P_{\bar \q^\perp} \int_{\unitq}\rho  \frac{4}{d} (\q \cdot \partial_t \bar \q)  H(\q \cdot \bar \q)\, \q\, \ud\q ,\\
X_3 &=& P_{\bar \q^\perp} \int_{\unitq}   \vezero(\q) \cdot (\nabla_x \rho)\, M_{\bar\q}  h(\q \cdot \bar \q)\, \q\, \ud\q ,\\
X_4 &=& P_{\bar \q^\perp} \int_{\unitq}   \vezero(\q) \cdot \lp \rho M_{\bar \q} \frac{4}{d}(\q \cdot \bar \q)\nabla_x (\q \cdot \bar \q) \rp  h(\q \cdot \bar \q)\, \q\, \ud\q \\
&=& P_{\bar \q^\perp} \int_{\unitq}   \rho \frac{4}{d}\Big( \q \cdot \lp \vezero(\q) \cdot \nabla_x \rp \bar \q\Big) H(\q \cdot \bar \q)\, \q\, \ud\q .
\eeqar
Using the change of variables $\q'=\q \bar \q^*$ (and skipping the primes) we obtain:
\beqar
X_1 &=& P_{\bar \q^\perp} \int_{\unitq} (\partial_t \rho) M_{1}(\Real(\q)) h(\Real(\q))\, \q\, \ud\q\, \bar \q, \\
X_2 &=& P_{\bar \q^\perp} \int_{\unitq}\rho  \frac{4}{d} (\q \cdot (\partial_t \bar \q)\bar \q^*)  H(\Real(\q)))\, \q\, \ud\q\, \bar \q,\\
X_3 &=& P_{\bar \q^\perp} \int_{\unitq}    \vezero(\q\bar\q) \cdot (\nabla_x \rho)\, M_{1}(\Real(\q))  h(\Real(\q))\, \q\, \ud\q \, \bar \q,\\
X_4 
&=&\rho \frac{4}{d}   P_{\bar \q^\perp} \int_{\unitq} \Big( \q \bar \q \cdot \lp \vezero(\q \bar \q) \cdot \nabla_x \rp \bar \q\Big) H(\Real(\q))\, \q\, \ud\q \bar \q,
\eeqar
with $M_1(\q)$ given by Eq. \eqref{eq:equilibrium_change_variables}.

Firstly notice the following: for any $\q \in \mathbb{H}$ and $\bar \q \in \unitq$ it holds
\be \label{eq:projection_ofproduct}
P_{\bar\q^\perp}(\q\bar\q)= \q\bar\q - (\q\bar\q \cdot \bar \q) \bar \q = \Ima (\q) \bar \q.
\ee

\paragraph{The term $X_1$.}
We apply Eq. \eqref{eq:projection_ofproduct} on $X_1$:
\beqar
X_1 &=& \int_{\unitq} (\partial_t \rho) M_{1}(\Real(\q)) h(\Real(\q))\, \Ima(\q)\, \ud\q\, \bar \q ,
\eeqar
which gives an odd integral in $\Ima (\q)$, hence $X_1=0$. 

\paragraph{The term $X_2$.}
To compute $X_2$ we firstly apply Eq. \eqref{eq:projection_ofproduct} again:
\beqar
X_2 &=&  \int_{\unitq}\rho  \frac{4}{d} (\q \cdot (\partial_t \bar \q)\bar \q^*)  H(\Real(\q)))\, \Ima(\q)\, \ud\q\, \bar \q.
\eeqar
Now, it holds that
\be \label{eq:image_product}
(\partial_t \bar \q) \bar \q^* = \Ima \lp (\partial_t \bar \q) \bar \q^* \rp,
\ee
since $\Real\lp (\partial_t \bar \q) \bar \q^* \rp= \partial_t \bar \q \cdot \bar \q =0$ given that $\bar \q \perp \partial_t \bar \q$. Eq. \eqref{eq:image_product} implies, in particular, that $\q \cdot (\partial\bar \q) \bar \q^* = \Ima (\q) \cdot \Ima((\partial_t \bar \q) \bar \q^*)$. With all these considerations we conclude that
\be
X_2 =\frac{4}{d}\rho \left[ \lp \int_{\unitq} \Ima(\q) \otimes \Ima(\q) \,   H(\Real(\q))\, \ud\q \rp  \Ima\lp (\partial_t \bar \q)\bar\q^* \rp \right] \bar \q. \label{eq:X2_2}
\ee 
Now observe that the off-diagonal elements in $\Ima(\q)\otimes \Ima (\q)$ are odd in the components $q_1, q_2, q_3$ (where $\q=(q_0, q_1,q_2,q_3)$), therefore the off-diagonal elements give integral zero. The diagonal elements corresponding to $q_i^2$ for $i=1,2,3$ can be permuted giving the same value for the integral. With these considerations in mind, we have that
\beqarl
\int_{\unitq} \Ima(\q) \otimes \Ima(\q)\, H(\Real(\q))\, \ud\q
&= & \int_{\unitq} \frac{q_1^2+q_2^2+q_3^2}{3} \Id \, H(\Real(\q))\, \ud\q \nonumber\\
&=& \int_{\unitq} \frac{1-\Real^2(\q)}{3} \, H(\Real(\q))\, \ud\q \, \Id \nonumber\\
&=:& C_2 \Id. \label{eq:integral_tensor_product}
\eeqarl 
We substitute Eq. \eqref{eq:integral_tensor_product} into Eq. \eqref{eq:X2_2} and conclude that
\be \label{eq:X2_final}
X_2 =\frac{4}{d}C_2\rho   \partial_t \bar \q\, ,
\ee
with $C_2$ given in Eq. \eqref{eq:integral_tensor_product}.

\paragraph{The term $X_3$.} 

We apply Eq. \eqref{eq:projection_ofproduct} to obtain
$$X_3 =   \int_{\unitq} \vezero(\q\bar\q) \cdot (\nabla_x \rho)\, M_{1}(\Real(\q))  h(\Real(\q))\, \Ima(\q)\, \ud\q \, \bar \q.$$
Then, we use the decompositon of $\vezero(\q\bar\q)$ in \eqref{eq:e_1_bar} 
to compute 
$$X_3:= X_{3,1}+X_{3,2}+X_{3,3},$$ 
where
\beqar
X_{3,1} &=&    \int_{\unitq} (2\Real^2(\q)-1) \left[\vezero(\bar\q) \cdot \nabla_x \rho\right]\, M_{1}(\Real(\q))  h(\Real(\q))\, \Ima(\q)\, \ud\q \, \bar \q \, ,\\
X_{3,2} &=&     \int_{\unitq}  2 \Real(\q)\lp \Ima (\q) \times \vezero(\bar\q) \rp \cdot (\nabla_x \rho)\, M_{1}(\Real(\q))  h(\Real(\q))\, \Ima(\q)\, \ud\q \, \bar \q \, ,\\
X_{3,3} &=&     \int_{\unitq}  2\lp (\Ima(\q) \otimes \Ima(\q))\vezero(\bar\q) \rp \cdot (\nabla_x \rho)\, M_{1}(\Real(\q))  h(\Real(\q))\, \Ima(\q)\, \ud\q \, \bar \q.\\
\eeqar
The integrands in the terms $X_{3,1}$ and $X_{3,3}$ are odd in $\Ima(\q)$ so $X_{3,1}=X_{3,3}=0$. Next, using that $(\Ima (\q) \times \vezero(\bar\q)) \cdot \nabla_x\rho= ( \vezero(\bar\q) \times \nabla_x\rho )\cdot \Ima(\q)$  we have
\beqar
X_3 =X_{3,2} &=& 2   \left[ \lp \int_{\unitq}   \Ima (\q) \otimes \Ima(\q) \,H(\Real(\q))\, \ud\q \rp \, ( \vezero(\bar\q) \times \nabla_x\rho ) \right] \bar \q \\
&=&2  C_2 (\vezero(\bar\q) \times \nabla_x \rho ) \bar \q,
\eeqar
by using Eq. \eqref{eq:integral_tensor_product}.

\paragraph{The term $X_4$.}
We first apply the projection
\beqar
X_4 
&=&\rho \frac{4}{d}   P_{\bar \q^\perp} \int_{\unitq}   \Big( \q \bar \q \cdot \lp \vezero(\q \bar \q) \cdot \nabla_x \rp \bar \q\Big) \,   H(\Real(\q))\, \q\, \ud\q \bar \q \\
&=&\rho \frac{4}{d}    \int_{\unitq}   \Big( \q \bar \q \cdot \lp \vezero(\q \bar \q) \cdot \nabla_x \rp \bar \q\Big) \,   H(\Real(\q))\, P_{\bar \q^\perp}\lp\q\,  \bar \q\rp \ud\q \\
&=&\rho \frac{4}{d}    \int_{\unitq}   \Big( \q \bar \q \cdot \lp \vezero(\q \bar \q) \cdot \nabla_x \rp \bar \q\Big) \,   H(\Real(\q))\, \Ima \lp\q\, \rp \ud\q  \bar \q.
\eeqar
Now, note that, since $\partial_{x_j} \bar \q \in \bar \q^\perp$, we have, for all $\q\in\unitq$, that
$$\lp\lp\Real(\q)\rp \bar \q\rp \cdot \lp \lp \vezero(\q \bar \q) \cdot \nabla_x \rp \bar \q \rp = 0.$$
Therefore
\beqar
X_4 
&=&\rho \frac{4}{d}    \int_{\unitq}   \Big(\lp\lp\Ima(\q)\rp \bar \q \rp \cdot \lp \lp \vezero(\q \bar \q) \cdot \nabla_x \rp \bar \q \rp\Big) \,   H(\Real(\q))\, \Ima \lp\q\, \rp \ud\q  \bar \q.\\
\eeqar
Then, using again the decomposition \eqref{eq:e_1_bar},
we can write
\beqar
X_4 
&=&X_{4,a}+X_{4,b}+X_{4,c}.
\eeqar
We  compute next
\beqar
X_{4,a} 
&=&\rho \frac{4}{d}    \int_{\unitq}   \lp 2\Real^2 \q-1\rp  \Big(\Ima\lp \q\rp \bar \q \cdot \lp \vezero(\bar \q) \cdot \nabla_x \rp \bar \q\Big)  \,   H(\Real(\q))\, \Ima \lp\q\, \rp \ud\q \, \bar \q\\
&=&\rho \frac{4}{d}    \int_{\unitq}   \lp 2\Real^2 \q-1\rp  \Big(\Ima\lp \q\rp \cdot [\lp \vezero(\bar \q) \cdot \nabla_x \rp \bar \q]  \bar \q^\ast\Big)  \,   H(\Real(\q))\, \Ima \lp\q\, \rp \ud\q \, \bar \q\\
&=&\rho \frac{4}{d}    \int_{\unitq}    \lp2\Real^2 \q-1\rp \, H(\Real(\q))   (\Ima(\q)) \otimes (\Ima(\q))  [\lp \vezero(\bar \q) \cdot \nabla_x \rp \bar \q]  \bar \q^\ast \ud\q\,\bar \q\\
&=&\rho \frac{4}{3d}    \int_{\unitq}    \lp2\Real^2 \q-1\rp \lp1-\Real^2 \q\rp \, H(\Real(\q)) \ud\q  \, \lp \vezero(\bar \q) \cdot \nabla_x \rp \bar \q,
\eeqar
where we use that the expression is odd in $\q_i$, $i=1,\hdots, 3$ on the off-diagonal terms and the symmetry in $\q_i$ to group the diagonal terms, analogously to the computation of the term \eqref{eq:integral_tensor_product}. Next, we have that
\beqar
X_{4,b} 
&=&\rho \frac{4}{d}    \int_{\unitq}    2 \Real(\q)  \Big(\lp \lp\Ima(\q)\rp \bar \q\rp \cdot \lp \lp \Ima(\q) \times  \vezero(\bar\q) \rp \cdot \nabla_x \rp \bar \q\Big)  \,   H(\Real(\q))\, \Ima \lp\q\, \rp \ud\q  \bar \q\\
&=&0,
\eeqar
since the integrand is odd in $\Ima(\q)$). Finally, we compute
\beqar
X_{4,c} 
&=&\rho \frac{4}{d}    \int_{\unitq}    2 \Big(\lp\lp\Ima(\q)\rp \bar \q \rp \cdot \lp \lp\lp \Ima(\q) \otimes \Ima(\q)\rp \vezero(\bar\q) \cdot \nabla_x \rp \bar \q \rp\Big) \,   H(\Real(\q))\, \Ima \lp\q\, \rp \ud\q  \bar \q\\
&=&\rho \frac{4}{d}    \int_{\unitq}    2 (\Ima(\q)\cdot \vezero(\bar\q)) \Big(\lp\lp\Ima(\q)\rp \bar \q \rp \cdot \lp \lp \Ima(\q) \cdot \nabla_x \rp \bar \q \rp\Big) \,   H(\Real(\q))\, \Ima \lp\q\, \rp \ud\q  \bar \q\\
&=&\rho \frac{4}{d}    \int_{\unitq}    2 (\Ima(\q)\cdot \vezero(\bar\q)) \Big(\lp\Ima(\q) \rp \cdot \lp \lp  \lp\Ima(\q) \cdot\nabla_x  \rp \bar \q\rp \bar \q^\ast\rp\Big) \,   H(\Real(\q))\, \Ima \lp\q\, \rp \ud\q  \bar \q\\
&=&\rho \frac{4}{d}    \int_{\unitq}    2 (\Ima(\q)\cdot \vezero(\bar\q)) \Big(\lp\Ima(\q) \rp \cdot \lp A(\bar \q) \Ima(\q)\rp\Big) \,   H(\Real(\q))\, \Ima \lp\q\, \rp \ud\q  \bar \q,
\eeqar
where $A(\bar \q)=\Ima \lp(\nabla_x \bar \q) \bar \q^\ast\rp$, seen as a $3\times 3$ matrix, that is
\begin{equation} \label{def:A} A(\bar \q)_{i,j}:=\lp(\pa_{x_j} \bar \q) \bar \q^\ast\rp_i = \lp\pa_{x_j, \text{rel}} \bar \q\rp_i, \quad \text{ for } i,\,j=1,\,2,\,3.
\end{equation}
We replace $A(\bar \q)$ by its symmetrization $A^S(\bar \q)=\frac{1}{2}\lp A(\bar \q) + A^t(\bar \q)\rp$, and diagonalize it in a orthogonal basis $A^S(\bar \q) = O^t D O$ with $O$ an orthogonal $3\times 3$ matrix and $D$ a diagonal $3\times 3$ matrix. Then, using a change of variable $\Ima(\q)' = O \Ima(\q)$, we have
\beqar
X_{4,c}
&=&\rho \frac{4}{d}    \int_{\unitq}    2 (\Ima(\q)\cdot \vezero(\bar\q)) \lp\Ima(\q)  \cdot (O^t D O\Ima(\q))   \rp \,   H(\Real(\q))\, \Ima \lp\q\, \rp \ud\q  \bar \q\\
&=&\rho \frac{4}{d}    \int_{\unitq}    2 (\Ima(\q)\cdot O \vezero(\bar\q)) \lp\Ima(\q)  \cdot (D \Ima(\q))   \rp \,   H(\Real(\q))\, O^t \Ima \lp\q\, \rp \ud\q  \bar \q\\
&=&\rho \frac{4}{d}    \lp O^t  \int_{\unitq}    2 (\Ima(\q)\otimes \Ima(\q) ) \lp\Ima(\q)  \cdot (D \Ima(\q))   \rp \,   H(\Real(\q))\, \ud\q O \vezero(\bar\q)  \rp \bar \q.
\eeqar
Again, since the integrand is odd in $q_i$, $i=1,\hdots,3$, the off-diagonal terms in the integral are zero. We compute the $i^{th}$ diagonal term of the integral for $i\in\{1,2,3\}$:
\beqar
&&\lp\int_{\unitq}    2 (\Ima(\q)\otimes \Ima(\q) ) \lp\Ima(\q)  \cdot (D \Ima(\q))   \rp \,   H(\Real(\q))\, \ud\q\rp_{i,i}\\
&=&\int_{\unitq}    2 \q_i^2 \lp\Ima(\q)  \cdot (D \Ima(\q))   \rp \,   H(\Real(\q))\, \ud\q\\
&=&\int_{\unitq}    2 \lp d_i \q_i^4 + \Sigma_{j\neq i} d_j \q_i^2 \q_j^2 \rp \,   H(\Real(\q))\, \ud\q\\
&=& 2 d_i \int_{\unitq}    \q_1^4 \,   H(\Real(\q))\, \ud\q +  2 \Sigma_{j\neq i} d_j \int_{\unitq}  \lp \q_1^2 \q_2^2 \rp \,   H(\Real(\q))\, \ud\q\\
&=& 2 d_i \int_{\unitq}    \q_1^2(\q_1^2-\q_2^2) \,   H(\Real(\q))\, \ud\q +  2 \text{Tr}(D) \int_{\unitq}  \q_1^2 \q_2^2 \,   H(\Real(\q))\, \ud\q\\
&=& \lp 2 D \int_{\unitq}    \q_1^2(\q_1^2-\q_2^2) \,   H(\Real(\q))\, \ud\q +  2 \text{Tr}(D)\Id  \int_{\unitq}  \q_1^2 \q_2^2 \,   H(\Real(\q))\, \ud\q\rp_i.
\eeqar
Inserting this expression into $X_{4,c}$, we have
\beqar
X_{4,c} &=&\rho \frac{4}{d}    \Bigg( O^t  \Bigg( 2 D \int_{\unitq}    \q_1^2(\q_1^2-\q_2^2) \,   H(\Real(\q))\, \ud\q \\
&& \qquad \quad+\,  2 \text{Tr}(D)\Id  \int_{\unitq}  \q_1^2 \q_2^2 \,   H(\Real(\q))\, \ud\q\Bigg) O \vezero(\bar\q)  \Bigg) \bar \q\\
&=&\rho \frac{4}{d}    \Bigg(  \Bigg( 2 A^S(\bar\q) \int_{\unitq}    \q_1^2(\q_1^2-\q_2^2) \,   H(\Real(\q))\, \ud\q\\
&& \qquad\quad+ \, 2 \text{Tr}(A(\bar\q))\Id  \int_{\unitq}  \q_1^2 \q_2^2 \,   H(\Real(\q))\, \ud\q\Bigg)\vezero(\bar\q)  \Bigg) \bar \q.
\eeqar

Finally, we obtain
\beqar
X_{4} 
&=&\rho \frac{4}{d}  \Big[ C_3 \, \lp \vezero(\bar \q) \cdot \nabla_x \rp \bar\q+  \lp 2 C_4 A^S(\bar\q) \vezero(\bar\q)+  2 C_5 \text{Tr}(A(\bar\q)) \vezero(\bar\q)  \rp \bar \q \Big].
\eeqar
with
\beqar
C_3&=&\int_{\unitq}    \frac{\lp2\Real^2 \q-1\rp}{3} \lp1-\Real^2 \q\rp \, H(\Real(\q)) \ud\q, \\
C_4&=&\int_{\unitq}    \q_1^2(\q_1^2-\q_2^2) \,   H(\Real(\q))\, \ud\q,\\
C_5&=&\int_{\unitq}  \q_1^2 \q_2^2 \,   H(\Real(\q))\, \ud\q.
\eeqar

Recall that $2 A^S(\bar \q) = A (\bar \q)+ A(\bar \q)^t$ with $A(\bar \q)$ defined in Eq. \eqref{def:A}. We compute
\beqar
A(\bar \q) \vezero(\q) &=& \lp \sum_j \lp\pa_{x_j,\text{rel}} \bar \q \rp_i (\vezero(\bar\q))_j \rp_{i=1,\dots,3} = \lp\lp\vezero(\bar\q)\cdot\nabla_x\rp \bar\q\rp \bar\q^\ast,\\
A^t(\bar \q) \vezero(\q) &=& \lp \sum_j \lp\pa_{x_i,\text{rel}} \bar \q \rp_j (\vezero(\bar\q))_j \rp_{i=1,\dots,3} = \nabla_{x,\text{rel}} \bar \q \, \vezero (\bar \q),\\
\text{Tr}(A(\bar\q)) &=& \sum_i \lp\pa_{x_i,\text{rel}} \bar \q \rp_i = \nabla_{x,\text{rel}} \cdot \bar \q,
\eeqar
thanks to Eqs. \eqref{eq:def_gradient_divergence_rel}--\eqref{def:Drel_times_e1}. Therefore, we have that
\beqar
X_{4} 
&=&\rho \frac{4}{d}  \Big[ (C_3+C_4) \, \lp \vezero(\bar \q) \cdot \nabla_x \rp \bar\q+  \lp C_4 \, \nabla_{x,\text{rel}} \bar \q \, \vezero (\bar \q)+  2 C_5 \, \lp\nabla_{x,\text{rel}} \cdot \bar \q\rp \, \vezero(\bar\q)  \rp \bar \q \Big].
\eeqar

\paragraph{End of the proof.}
Finally, we conclude that $X=X_1+X_2+X_3+X_4=0$ is equivalent to
\begin{equation*} \begin{split}
0=X =&\frac{4}{d}C_2\rho \partial_t \bar \q + 2  C_2 (\vezero(\bar\q) \times \nabla_x \rho ) \bar \q \\
&+\rho \frac{4}{d}  \Big[ (C_3+C_4) \, \lp \vezero(\bar \q) \cdot \nabla_x \rp \bar\q+  \lp C_4  \, \nabla_{x,\text{rel}} \bar \q \, \vezero(\bar\q)+  2 C_5  \, \lp\nabla_{x,\text{rel}} \cdot \bar \q\rp \, \vezero(\bar\q)  \rp \bar \q \Big],
\end{split}\end{equation*}
so that
\begin{equation*} \begin{split}
&\rho  \left[ \partial_t \bar \q + \frac{C_3+C_4}{C_2} \, \lp \vezero(\bar \q) \cdot \nabla_x \rp \bar\q \right] + \frac{d}{2}  (\vezero(\bar\q) \times \nabla_x \rho ) \bar \q \\
& \qquad+\rho \left[  \lp \frac{C_4}{C_2}  \, \nabla_{x,\text{rel}} \bar \q \,\vezero(\bar\q)+  2 \frac{C_5}{C_2} \, \lp\nabla_{x,\text{rel}} \cdot \bar \q\rp \,\vezero(\bar\q)  \rp \bar \q \right]=0.
\end{split}\end{equation*}

It remains to compute each one of the constants $C_i$, $i=1,\hdots,5$. For this we will use repeatedly the change of variable of Prop. \ref{prop:volume_element} and the following:
$$ 4\pi \int_0^\pi f(\theta) H(\cos(\theta/2))\sin^2 (\theta/2)\ud \theta = \la f(\theta) \cos(\theta/2) h(\cos(\theta/2)) \rangle_{m\, \sin^2(\theta/2)} \nn,$$
which is a direct consequence of the definitions \eqref{def:bigH}, \eqref{eq:definition_angles}, and of Eq. \eqref{eq:M_and_m}.

We now compute
\beqar
C_2&=&\int_{\unitq} \frac{1-\Real^2(\q)}{3} \, H(\Real(\q))\, \ud\q \nonumber\\
&=& 4\pi \int_0^\pi \frac{\sin^2(\theta/2)}{3} H(\cos(\theta/2)) 
\sin^2(\theta/2)\, \ud\theta \nonumber\\
&=& \frac{1}{3} \la \sin^2(\theta/2) \cos(\theta/2) h(\cos(\theta/2)) \rangle_{m\, \sin^2(\theta/2)} \nn;\\
C_3 &=&\int_{\unitq}    \frac{\lp2\Real^2 \q-1\rp}{3} \lp1-\Real^2 \q\rp \, H(\Real(\q)) \ud\q \\
&=& 4\pi \int_0^\pi \frac{(2\cos^2(\theta/2)-1)}{3} \sin^2(\theta/2) H(\cos(\theta/2)) \sin^2(\theta/2)\, \ud \theta\\
&=& \frac{1}{3}\la (2\cos^2(\theta/2)-1) \sin^2(\theta/2) \cos(\theta/2)h(\cos(\theta/2)) \ra_{m\, \sin^2(\theta/2)};\\
C_5 &=& \int_{\unitq}  \q_1^2 \q_2^2 \,   H(\Real(\q))\, \ud\q\\ 
 &=& \int_{\mathbb{S}^2}\int^\pi_0 (\sin(\theta/2) n_1)^2 (\sin(\theta/2) n_2)^2 H(\cos(\theta/2))\sin^2(\theta/2)\ud\theta \ud\nvec\\
 &=& \lp\int^\pi_0 \sin^4(\theta/2) H(\cos(\theta/2))\sin^2(\theta/2)\, \ud\theta\rp \lp\int^\pi_0\int^{2\pi}_0 \cos^2\theta_2\sin^2\theta_2\cos^2\theta_3\sin\theta_2\ud\theta_3\ud\theta_2\rp\\
 &=&\frac{4 \pi}{15}\lp\int^\pi_0 \sin^4(\theta/2) H(\cos(\theta/2))\sin^2(\theta/2)\, \ud\theta\rp \\
 &=&\frac{1}{15}\la \sin^4(\theta/2) \cos(\theta/2)h(\cos(\theta/2))\ra_{m\, \sin^2(\theta/2)} ;\\
 C_4 &=& \int_{\unitq}    \q_1^2(\q_1^2-\q_2^2) \,   H(\Real(\q))\, \ud\q \\
 &=& \int_{\unitq} \q_1^4 H(\Real(\q))\, \ud\q - C_5\\
 &=& \lp\int_{\mathbb{S}^2} n_1^4 \ud\nvec\rp \lp \int^\pi_0 \sin^4(\theta/2) H(\cos(\theta/2)) \sin^2(\theta/2)\, \ud\theta\rp -C_5\\
 &=& \lp \frac{4\pi}{5}-\frac{4\pi}{15}\rp\int^\pi_0 \sin^4(\theta/2) H(\cos(\theta/2)) \sin^2(\theta/2)\, \ud\theta\\
 &=& \frac{8\pi}{15} \int^\pi_0 \sin^4(\theta/2) H(\cos(\theta/2)) \sin^2(\theta/2)\, \ud\theta\\
 &=& 2 C_5.
\eeqar
We finally compute
\beqar
c_2:=\frac{C_3+C_4}{C_2} &=& 3 \la \frac{2 \cos^2(\theta/2)-1}{3} + \frac{2}{15} \sin^2(\theta/2) \ra_{m \sin^4(\theta/2) h(\cos(\theta/2)) \cos(\theta/2)}\\
&=& \frac{1}{5} \la 1 + 4 \cos\theta \ra_{m \sin^4(\theta/2) h(\cos(\theta/2)) \cos(\theta/2)};\\
c_4:=\frac{C_4}{C_2}=2\frac{C_5}{C_2}&=&\frac{2}{5}\la \sin^{2}(\theta/2) \ra_{\sin^4(\theta/2)m h(\cos(\theta/2)) \cos(\theta/2)} \\
&=& \frac{1}{5}\la 1-\cos\theta \ra_{\sin^4(\theta/2)m h(\cos(\theta/2)) \cos(\theta/2)}.
\eeqar
\end{proof}

\section{Comparison with the results in Ref. \cite{bodyattitude}}
\label{sec:comparison}

In this section we compare the models presented here with the ones obtained for the body attitude coordination model in Ref. \cite{bodyattitude}. The crucial differnce between the two approaches is that, while in Ref. \cite{bodyattitude} the representation of body attitudes relies on rotation matrices in $SO(3)$, here it relies on unitary quaternions in $\unitq$ (which are more computationnally efficient).

After an introductory prensentation of the links between $SO(3)$ and $\unitq$ (Sec. \ref{sec:unitq_matr}), we present the two main results of this section: the equivalence between the individual based models (Th. \ref{cor:equivalence_process}, in Sec. \ref{sec:compare_micro}), and the equivalence between the macroscopic models (Th. \ref{th:equivalence_macro_equations}, in Sec. \ref{sec:comparison_macro}).

\subsection{Relation between unitary quaternions and rotation matrices}\label{sec:unitq_matr}

\medskip
  We first introduce some notations.
Rotations in $\mathbb{R}^3$ can be described mathematically in different ways. In this section we consider three particular descriptions, namely, the group of orthonormal matrices corresponding to the rotation group $SO(3)$; the description via unitary quaternions $\q\in \unitq$; and, finally, rotations  described by the pair $(\theta, \nvec)\in[0,\pi]\times \mathbb{S}^2$ where $\nvec$ indicates the axis of rotation and $\theta$ the angle of rotation anti-clockwise around $\nvec$. For $A\in SO(3)$, $\q \in \unitq$ and $(\theta, \nvec)\in[0,\pi]\times \mathbb{S}^2$ corresponding to the same rotation, we have the following identities
\beqarl
A &=& A(\theta,\nvec)=\Id+ \sin(\theta)[\nvec]_\times + (1-\cos\theta)[\nvec]_\times^2= \exp(\theta[\nvec]_\times) \mbox{ (Rodrigues' formula),}\nonumber\\ \label{eq:rodrigues_formula}\\
\q &=& \q(\theta,\nvec)= \cos(\theta/2)+\sin(\theta/2) \nvec, \label{eq:q_theta_n}\\
\label{eq:correspondance_rotation_quaternion}
A{\bf v} &=& \Ima(\q\bar {\bf v}\q^*), \quad \mbox{(rotation of ${\bf v}$) for any ${\bf v}\in \mathbb{R}^3$}, 
\eeqarl
where $\bar {\bf v}\in \mathbb{H}$ with $\Real(\bar {\bf v})=0$ and $\Ima(\bar {\bf v})={\bf v}$; and
where we abuse notation in Eq. \eqref{eq:q_theta_n} and understand $\nvec$  as written in the Hamiltonian basis $\nvec = n_1 \vec{\imath}+ n_2 \vec{\jmath}+n_3 \vec{k}$ rather than in the canonical basis $\nvec=(n_1, n_2, n_3)$. Notice that when $\theta=0$, the vector $\nvec$ is not defined, but this does not pose a problem in the sense that 
 there is an unambiguous correspondence with $A=\Id$ and $\q=1$.

Define the operator $\Phi: \unitq \to SO(3)$
 by
\be \label{eq:def_Phi}
\Phi: \unitq \longrightarrow SO(3), \, \q\mapsto (\Phi (\q): \uu\in\R^3\mapsto \Ima (\q \uu \q^\ast)\in\R^3).
\ee
This operator associates to each unitary quaternion $\q \in \unitq$, the corresponding rotation matrix $A= \Phi(\q) \in SO(3)$.
 In particular, the following identities hold for  any $\q, \mathbf{r} \in \unitq$:
\beqarl
 i)&& \Id = \Phi(1);\\
 ii)&&\Phi(\q)=\Phi(-\q) ;\\
 iii)&&\Phi(\q^*) = \left[\Phi(\q)\right]^t; \label{eq:transpose_matrix}\\
 iv)&& \Phi(\q)\Phi(\mathbf{r})= \Phi(\q \mathbf{r}). \label{eq:product_matrix}
\eeqarl
Identities $ii)$ and $iv)$ are consequences of Eq. \eqref{eq:correspondance_rotation_quaternion}; identity $iii)$ is a consequence of $iv)$, noticing that $\q^\ast=\q^{-1}$ and $A^t=A^{-1}$.

\bigskip
First, we show the relation between the inner products in $SO(3)$ and $\unitq$:
\begin{lemma} \label{lem:equivalence_scalar_product}
Let $A,B\in SO(3)$ and $\q,\mathbf{r}\in \unitq$ to be such that
$A=\Phi(\q)$ and $B=\Phi(\mathbf{r})$, then
\be \label{eq:equivalence_scalar_product}
\frac{1}{2}A \cdot B= (\q \cdot \mathbf{r})^2 -\frac{1}{4} = \q \cdot \lp \mathbf{r}\otimes \mathbf{r}- \frac{1}{4}\Id \rp \q.
\ee
where $A \cdot B = \tr(AB^t)/2$, with $\tr$ denoting the trace.
\end{lemma} 
\begin{proof}[Proof of Lem. \ref{lem:equivalence_scalar_product} \nameref{lem:equivalence_scalar_product}]
To check Eq. \eqref{eq:equivalence_scalar_product}, we  recast the inner products in $SO(3)$ and $\unitq$  in the variables $(\theta,\nvec) \in [0,\pi]\times \mathbb{S}^2$. 
By Eqs. \eqref{eq:transpose_matrix}-\eqref{eq:product_matrix}, it holds  that $AB^t = \Phi(\q\mathbf{r}^*)\in SO(3)$. Let $(\theta,\nvec) \in [0,\pi]\times \mathbb{S}^2$ be the angle and rotation axis representing the same rotation as $AB^t$ (and $\q \mathbf{r}^*$). We have that
\beqar
\q \cdot \mathbf{r}& = &\Real(\q \mathbf{r}^*)=1\cdot \q\mathbf{r}^*= \cos(\theta/2),\\
A \cdot B & = &\frac{1}{2}\tr(AB^t)=\Id\cdot AB^t= \Id \cdot \lp \Id + \sin\theta[\nvec]_\times + (1-\cos\theta)[\nvec]_\times^2 \rp = \frac{1}{2}+ \cos\theta.
\eeqar
so
$$(\q\cdot \mathbf{r})^2 - \frac{1}{4}= \cos^2(\theta/2)-\frac{1}{4}= \frac{1+\cos\theta}{2}-\frac{1}{4}= \frac{1}{4}+\frac{\cos\theta}{2}= \frac{1}{2}A\cdot B.$$
The second equality in Eq. \eqref{eq:equivalence_scalar_product} is obtained directly using that $|\q|=1$.
\end{proof}

Next, we establish the correspondence between integrals in $SO(3)$ and $\unitq$:
\begin{lemma}[Comparison of volume elements]
\label{lem:volume_elements}
Consider $g:SO(3)\to \R$, then
\be \label{eq:equivalence integrals}
\int_{SO(3)}g(A)\, \ud A= \frac{1}{2\pi^2} \int_{\mathbb{H}_1}g(\Phi(\q))\,\ud\q,
\ee
where $\ud\q$ is the Lebesgue measure on the hypersphere $\unitq$ and $\ud A$ is the normalized Lebesgue measure on $SO(3)$.
\end{lemma}
\begin{proof}[Proof of Lem. \ref{lem:volume_elements}  \nameref{lem:volume_elements}]
We apply Prop. \ref{prop:volume_element} to the (even) function $f(\q):=g(\Phi(\q))$, to get
$$\int_{\unitq} g(\Phi(\q))\, \ud\q = \int^{\pi}_0 \sin^2(\theta/2) \int_{\mathbb{S}^2}g(\Phi(\cos(\theta/2)+\sin(\theta/2)\nvec)) \, \ud\theta \ud\nvec.$$
Using Rodrigues' formula \eqref{eq:rodrigues_formula}, it yields
$$\int_{\unitq} g(\Phi(\q))\, \ud\q = \int^{\pi}_0 \sin^2(\theta/2) \int_{\mathbb{S}^2}\tilde g (\theta,\nvec) \, \ud\nvec \ud\theta,$$
where $\tilde g (\theta,\nvec):=g(\exp(\theta[\nvec]_\times))$.

On the other hand, from Ref. \cite{bodyattitude}, we know that it holds:
\be \label{eq:volume_A}
\int_{SO(3)} g(A)\, \ud A= \frac{1}{2\pi^2}\int^\pi_0 \sin^2({\theta}/{2})\int_{\mathbb{S}^2} \tilde g(\theta,\nvec ) \ud\nvec \ud\theta,
\ee
and this concludes the proof.
\end{proof}

\bigskip
Finally, one can check that $\Phi$ is continuously differentiable on $\unitq$ given that it is a quadratic function on $\unitq$. It holds the following:
\begin{proposition}
\label{prop:DPhi}
Noting $\textnormal{D}_\q \Phi: \q^\perp \longrightarrow T_{\Phi(\q)}$ the differential of $\Phi$ at $\q\in \unitq$, we have that for any $\q\in\unitq$ and any vector $\uu\in\R^3$,
\begin{equation}\label{eq:DPhi_is_uPhi0}
 \textnormal{D}_\q\Phi (\uu \q) = 2 \left[ \uu \right]_\times \Phi(\q).
 \end{equation}
Equivalently, since the tangent space at $\q\in\unitq$ is exactly the set $\q^\perp=\{\uu\q, \, \uu\in\R^3\}$ (see Prop. \ref{prop:tangent space}), we have that 
\begin{equation}\label{eq:DPhi_is_uPhi} 
 \textnormal{D}_\q\Phi (\mathbf{p}) = 2 \left[ \mathbf{p} \q^\ast \right]_\times \Phi(\q), \qquad \mbox{for all }\mathbf{p}\in \q^\perp.
 \end{equation}
  \end{proposition}
 \begin{remark} From this relation, we can deduce the links between the gradient, divergence and laplacian operators in $SO(3)$ and $\unitq$: see Prop. \ref{lem:equivalent_equations}--\ref{prop:laplacians} in Annex \ref{annex:Phi}.
 \end{remark}
 
 \begin{proof}
The operator $\Phi$ in Eq. \eqref{eq:def_Phi} is quadratic and associated to the symmetric bilinear operator $\Phi_{\textnormal{BL}}$ defined by, for $\p_1,\,\p_2 \in \mathbb{H}$ and for $\vv\in\R^3$,
\beqar
\Phi_{\textnormal{BL}} (\p_1,\p_2) (\vv)= \Ima (\p_1 \vv\p_2^*).
\eeqar
Note that this operator is indeed symmetric since $\Ima (\p_2 \vv\p_1^*)= -\Ima ((\p_2 \vv\p_1^*)^*)$ and $\vv^*=-\vv$. We then use Prop. \ref{prop:} to conclude that, for any $\q\in\unitq$, any $\p_1=\uu\q \in \q^\perp$ (with $\uu\in\R^3$) and for any $\vv\in\R^3$,
\beqar
[\textnormal{D}_\q \Phi \, (\uu\q)] (\vv) &=& 2 \Phi_\textnormal{BL}(\uu\q,\q) (\vv) = 2\Ima (\uu\q \vv\q^*)\\
&=& 2 \uu\times \Ima (\q \vv\q^*) = 2 [\uu]_\times \Phi(\q) \vv .
\eeqar
 \end{proof}

\subsection{Equivalence between individual based models}
\label{sec:compare_micro}

In this section we check that the flocking dynamic considered in \cite{bodyattitude} corresponds with that of Eqs. \eqref{eq:particleX}-\eqref{eq:particleQ}.

In \cite{bodyattitude} the authors describe an individual based model for body attitude coordination given by the evolution over time of $(X_k, A_k)_{k=1,\hdots,N}$ of $N$ agents,  where $X_k\in \R^3$ is the position of agent $k$ and $A_k\in SO(3)$ is a rotation matrix giving its body attitude. The evolution of the system is given by the following equations:
\begin{eqnarray} \label{eq:IBM_rotation}
	\ud\mathbf{X_k}(t) &=& v_0 A_k(t) \vezero \ud t,\\
	\ud A_k(t) &=& P_{T_{A_k}} \circ \left[ \nu \PD(M_k) \ud t + 2\sqrt{D} \ud W^k_t \right], \label{eq:IBM2_rotation}
\end{eqnarray}
where the Stochastic Differential Equation is in  Stratonovich sense (see Ref.~\cite{gardiner}); $W_t^k$ is the Brownian motion in the space of squared matrices; $M_k$ is defined as
\be \label{eq:Mk}
   M_k(t) := \frac{1}{N}\sum_{i=1}^N K(|\mathbf{X}_i(t)-\mathbf{X_k}(t)|) A_i(t),
\ee
where $K$ is a positive interaction kernel;
 $\nu$, $v_0$ and $D$ are positive constants; $\vezero$ is a vector; and $P_{T_A}$ is the projection in $SO(3)$ to the tangent space to $A$.
The term $PD(M)$ denotes the orthogonal matrix obtained from the polar decomposition of $M$ which is defined as follows:
\begin{lemma}[Polar decomposition of a square matrix.\cite{golub2012matrix}]~
\label{lem:polar_decomposition}

Given a matrix $M\in \mathcal{M}$, if~$\det(M)\neq 0$ then there exists a unique orthogonal matrix~$A=PD(M)$ (given by~$A= M(\sqrt{M^tM})^{-1}$) and a unique symmetric positive definite matrix~$S$ such that~$M=AS$.
\end{lemma}

\medskip
The vector $A_k\vezero$ in Eq. \eqref{eq:IBM_rotation} gives the direction of movement of agent $k$ and is obtained as the rotation of the vector $\vezero$ by $A_k$. Equivalently, we can express it as  $A_k\vezero~=~\vezero(\q_k)$ (in the notation of Eq. \eqref{eq:particleX}) as long as $A_k$ and $\q_k$ represent the same rotation. Therefore, Eqs. \eqref{eq:particleX} and \eqref{eq:IBM_rotation} represent the same dynamics and we are left to check that $\q_k=\q_k(t)$ and $A_k=A_k(t)$  in Eqs. \eqref{eq:particleQ} and \eqref{eq:IBM2_rotation} represent the same rotation for each time $t$ where the solutions are defined.  

\bigskip
The goal of this section will be to prove that the solution to the stochastic differential equation \eqref{eq:particleX}-\eqref{eq:particleQ} and the solution of the stochastic differential equation \eqref{eq:IBM_rotation}-\eqref{eq:IBM2_rotation} are the same in law (in a precise way that will be given later).

\bigskip
The main result of this section is the following:
\begin{theorem}[Equivalence in law] \label{cor:equivalence_process}
 The processes \eqref{eq:particleX}-\eqref{eq:particleQ} and \eqref{eq:IBM_rotation}-\eqref{eq:IBM2_rotation} are the same in law.
\end{theorem}

The proof is done at the end of this Section. First, we remark that in the absence of randomness (Brownian motion) the equations for the evolution of the body attitude are equivalent:
\begin{proposition} \label{prop:equivalence_normalised_models}
 Let $A_0\in SO(3)$ and $\q_0\in\unitq$ represent the same rotation. Consider the matrix $M_k$ given in Eq. \eqref{eq:Mk}; the matrix $Q_k$ given in Eq. \eqref{eq:definition of Q_k}; and $\bar \q_k \in \unitq$  given in Eq. \eqref{eq:definition bar q}. Then, if $det(M_k)>0$, the following two Cauchy problems are equivalent (in the sense that $A_k=A_k(t)$ and $\q_k=\q_k(t)$ represent the same rotation for all $t$ where the solution is uniquely defined):
\beqar
\frac{\ud A_k}{\ud t}&=& P_{T_{A_k}}(PD(M_k)), \qquad A_k(0)=A_0,\\
\frac{\ud\q_k}{\ud t} &=& P_{\q_k^\perp}\left[\lp \bar \q_k \otimes \bar \q_k - \frac{1}{4}\Id\rp \q_k \right], \qquad \q_k(0)=
\q_0.
\eeqar
Note that these two Cauchy problems can also be written, respectively
\beqar
\frac{\ud A_k}{\ud t}&=& \nabla_A\left[ PD(M_k)\cdot A \right]_{|A=A_k}, \qquad A_k(0)=A_0,\\
\frac{\ud\q_k}{\ud t} &=& \frac{1}{4}\nabla_{\q} \left[ 2\q \cdot \lp \bar \q_k \otimes \bar \q_k - \frac{1}{4}\Id\rp \q \right]_{|\q=\q_k} , \qquad \q_k(0)=
\q_0,
\eeqar
where $\nabla_A$ and $\nabla_\q$ are the gradients in $SO(3)$ and $\unitq$, respectively.
\end{proposition}

\bigskip
To prove this Proposition we first check that the average orientation of the neighbourgs is the same in the two models, in the sense described below:
\begin{lemma}
\label{lem:equivalence_non_normalized_models}
 Let $A_i=\Phi(\q_i)$ for $i=1, \hdots, N$, then, for every $k\in\{1, \hdots, N\}$ it holds
\be  \label{eq:relationPM_Q}
PD(M_k)\cdot A_k= 2 \q_k \cdot F_k(\q_k),
\ee
as long as $det(M_k)>0$, where $M_k$ is defined in Eq. \eqref{eq:Mk} and $F_k$ is given in Eq. \eqref{eq:definition F_k}.
\end{lemma}
\begin{proof}
Assume for simplicity that $K\equiv 1$ (the general case can be proven equally). Firstly, notice that for $A=\Phi(\q)$ ($\q\in\unitq$) it holds 
$$M_k\cdot A = \frac{1}{N} \sum_{i=1}^N A_i\cdot A=\frac{1}{N}\sum_{i=1}^N \Phi(\q_i) \cdot \Phi(\q) =2 \q\cdot\lp \frac{1}{N} \sum_{i=1}^N \lp \q_i\otimes \q_i - \frac{1}{4}\Id \rp \rp \q = 2 \q \cdot Q_k\q,$$
for $Q_k$ given in Eq. \eqref{eq:definition of Q_k} and
where we used Lem. \ref{lem:equivalence_scalar_product} to compute the inner product. 
Therefore, for any $\q\in\unitq$,
$$2 \q \cdot Q_k\q = M_k\cdot \Phi(\q).$$
Now, the definition of $\bar \q_k$ implies that it maximizes $\q\mapsto \q \cdot Q_k\q$ in $\unitq$. Since $\q \cdot Q_k\q= \frac{1}{2}M_k\cdot A$, this implies that $\Phi(\bar \q_k)$ maximizes $A \mapsto M_k\cdot A$ in $SO(3)$ which is a property that  characterises the matrix $PD(M_k)$ (see \cite[Prop 3.1]{bodyattitude}). Therefore, it holds that $\Phi(\bar \q_k)= PD(M_k)$ as long as $det (M_k)~>~0$, and in this case, using again Lem. \ref{lem:equivalence_scalar_product}, we have
\be
PD(M_k) \cdot A_k = \Phi(\bar \q_k) \cdot \Phi(\q_k) =  2 \q_k \cdot \lp\lp \bar \q_k \otimes \bar \q_k- \frac{1}{4}\mbox{Id} \rp\q_k\rp.
\ee
\end{proof}

We are now ready to prove Prop. \ref{prop:equivalence_normalised_models}.

\begin{proof}[Proof of Prop. \ref{prop:equivalence_normalised_models}] The fact that we can rewrite the first pair of Cauchy problems as the second one comes from the equalities
\be \label{eq:gradient_form}
P_{T_A}(PD(M))= \nabla_A(PD(M)\cdot A), \qquad P_{\q^\perp}F(\q) = \frac{1}{4}\nabla_\q (2\q \cdot F(\q)),
\ee
where
$$ F(\q):=\lp \bar \q_k \otimes \bar \q_k- \frac{1}{4}\mbox{Id} \rp\q.$$
We conclude the equivalence thanks to Lem. \ref{lem:equivalence_non_normalized_models} and Prop. \ref{lem:equivalent_equations}.
\end{proof}

To prove Th. \ref{cor:equivalence_process} we need the following result:
\begin{proposition}
\label{prop:stochastic_equivalence} Let $\sigma>0$, and let $H$ be a time-dependent tangent vector field on $\unitq$:
$$H:\unitq\times [0,\infty) \to \mathbb{H} \quad \text{with} \quad H(\q,t) \in \q^\perp, \; \text{for all}\, \q\in\unitq,\,t\ge0.$$
Let $\tilde \sigma>0$, and let $\tilde H$ be a time-dependent tangent vector field on $SO(3)$:
 $$\tilde H:SO(3)\times [0,\infty) \to M_3 \quad \text{with} \quad \tilde H(A,t) \in T_A, \; \text{for all}\, A\in SO(3),\,t\ge0.$$
Suppose that the following relations hold:
\be \label{eq:PTA_PTq0}
\tilde H (\Phi(\q)) = \textnormal{D}_\q \Phi \, (H(\q)), \qquad\mbox{for all } \q\in \unitq,
\ee
 and 
\be \label{eq:sigma} 
\tilde \sigma = \sqrt{8} \, \sigma.
\ee
Let $\tilde p_t$ be the law over time of a stochastic process in $SO(3)$ defined by
$$\ud A=\tilde H(A,t) \ud t +\tilde \sigma P_{T_A}\circ \ud\tilde B_t,$$
for $\tilde B_t$ a 9-dimensional Brownian motion.
Then, if $\tilde p_t$ is an absolutely continuous measure, the absolutely continous measure $p_t$ defined by
\be \label{eq:relation_laws}
\tilde p_t(\Phi(\q))=2\pi^2 p_t(\q), \quad\mbox{for all } \q\in \unitq,\, t\ge0,
\ee
is the law over time of a stochastic process in $\unitq$ defined by
\be \label{eq:dq_1particle}
\ud\q = H(\q,t) \ud t + \sigma P_{\q^\perp}\circ \ud B_t,
\ee
for $B_t$ a 4-dimensional Brownian motion.
\end{proposition}
\begin{proof}
Firstly, notice that for any borel set $B\subset \unitq$ it holds
$$
\int_{\Phi(B)} \tilde p_t(A)\,\ud A\, = \int_B  p_t(\q)\,\ud\q\,
$$
thanks to Lem. \ref{lem:volume_elements}. Note that this is the reason why we introduce the factor $2\pi^2$ in \eqref{eq:relation_laws}, which allows to have this equivalence of integrals. 
We start from the equation for $\tilde p_t$:
\be \label{eq:p_tilde}
\partial_t \tilde p_t(A) + \nabla_A\cdot \big( \tilde H(A,t) \, \tilde p_t (A)\big) = \frac{\tilde \sigma^2}{4}\Delta_A \,\tilde p_t(A).
\ee
Notice that the fact that we obtain a factor $\tilde \sigma^2/4$ is consequence of considering the inner product $A\cdot B = \mbox{trace}(A^tB)/2$ (see Ref. \cite{bodyattitude}).
By Prop. \ref{lem:equivalence_divergence} we have that
$$\nabla_A\cdot \big(\tilde H(\cdot,t) \, \tilde p_t\big)(\Phi(\q))= \nabla_\q \cdot\big(H(\q,t)\, 2\pi^2 p_t(\q)\big),$$
and by Prop. \ref{prop:laplacians} we have that
$$\Delta_A \,\tilde p_t (\Phi(\q)) = \frac{1}{4} \Delta_\q \,(2\pi^2p_t(\q)).$$
Therefore, we recast Eq. \eqref{eq:p_tilde} into
$$\partial_t p_t(\q)  + \nabla_\q \cdot \big(H(\q,t)\,  p_t(\q)\big) = \frac{\tilde \sigma^2}{16}\Delta_\q \,p_t(\q).$$
Consequently, $p_t$ is the law of the process
$$\ud\q = H(\q,t)\, \ud t + \frac{\tilde \sigma}{\sqrt{8}} P_{\q^\perp} \circ \ud B_t.$$
which is exactly Eq. \eqref{eq:dq_1particle} thanks to Eq. \eqref{eq:sigma}.
\end{proof}

Finally we are ready to prove
\begin{proof}[Proof of Th. \ref{cor:equivalence_process}]
Using Prop. \ref{lem:equivalent_equations}, we deduce from Eq. \eqref{eq:relationPM_Q}, that
\be
\nabla_A \left.\lp PD(M_k)\cdot A\rp\right|_{A=A_k} =  \textnormal{D}_\q \Phi \lp  \frac12 \left. \nabla_\q\lp \q\cdot F_k\rp\right|_{\q=\q_k}\rp,
\ee
which can be rewritten thanks to Eq. \eqref{eq:gradient_form} as
\be
P_{T_{A_k}}\lp PD(M_k)\rp =  \textnormal{D}_\q \Phi \lp P_{\q_k^\perp}F_k\rp.
\ee
The condition \eqref{eq:PTA_PTq0} in Prop. \ref{prop:stochastic_equivalence} is then satisfied with $\tilde H(A)=P_{T_{A}}\lp PD(M_k)\rp$ and $H(\q)=P_{\q^\perp}F(\q)$, so that, proceeding similarly as in Prop. \ref{prop:stochastic_equivalence} for the $N$-particles system, we conclude the result.

\end{proof}

\subsection{Comparison of the macroscopic model with Eqs. \eqref{eq:macro_rho}--\eqref{eq:macro_lambda}}
\label{sec:comparison_macro}
In this section we show the equivalence between the macroscopic system \eqref{eq:continuity_equation}--\eqref{eq:constants} (or, equivalently, Eq. \eqref{eq:q_relative_form} for the last expression), expressed in terms of unitary quaternions, and the system \eqref{eq:macro_rho}--\eqref{eq:macro_lambda}, expressed in term of rotation matrices from Ref. \cite{bodyattitude}.   

Recall $\Phi$ the natural map between unitary quaternions and rotation matrices defined in Eq. \eqref{eq:def_Phi}. We first notice that if $\q$ and $\Lambda$ represent the same rotation (that is, if $\Phi(\q)=\Lambda$), then
$$\Lambda \vezero = \vezero(\bar\q).$$
Therefore, the continuity equations \eqref{sys:macro1} and \eqref{eq:macro_rho} represent the same dynamics (it is direct from their definitions in Eq. \eqref{eq:c1} and in Ref. \cite{bodyattitude} that the constants $c_1$ and $\tilde c_1$ are identical). We are left with comparing the various differential operators in $\bar \q$ and $\Lambda$ in Eqs. \eqref{eq:q_relative_form} and \eqref{eq:macro_lambda}.

\subsubsection{Relation between the differential operators $\delta_x$, $\rvec_x$ and $\pa_\textnormal{rel}$.} 

\begin{proposition}
Let $\bar\q=\bar\q(t,x)$ be a function on $\R_+\times\R^3$ with values in $\unitq$. We define $\Lambda=\Phi(\bar\q)$ the matrix representation of the rotation represented by $\bar\q$. Let $\mathbf{v}=\mathbf{v}(t,x)$ be a vector field in $\R^3$. Then we have the following equalities (everywhere on $\R_+\times\R^3$):
\begin{eqnarray}
 \lp\pa \Lambda\rp \Lambda^t &=&2 \Big[\pa_\textnormal{rel} \bar\q\Big]_\times, \quad \text{for }\pa\in\{\pa_t, \pa_{x_1},\pa_{x_2},\pa_{x_3}\}, \label{equiv:partial}\\
\label{equiv:DLambda}\mathscr{D}_x(\Lambda) &=&2 (\nabla_{x,\textnormal{rel}} \bar\q)^t,\\
\label{equiv:delta}\delta_x(\Lambda) &=& 2\nabla_{x,\textnormal{rel}} \cdot \bar\q,\\
\label{equiv:rvec}\mathbf{v}\times\rvec_x(\Lambda) &=&2\lp\nabla_{x,\textnormal{rel}} \bar\q\rp\mathbf{v}-2\lp\mathbf{v}\cdot\nabla_{x,\textnormal{rel}}\rp\bar\q,
\end{eqnarray}
where the operators $\nabla_{x,\textnormal{rel}}$ and $\nabla_{x,\textnormal{rel}} \cdot$ are defined in Eqs. \eqref{eq:def_gradient_divergence_rel}--\eqref{eq:def_gradient_divergence_rel2}.
\end{proposition}
\begin{proof}
Eq. \eqref{equiv:partial} is obtained by first differentiating the equality $\Lambda=\Phi(\bar\q)$,
$$\partial\Lambda = \textnormal{D}_\q\Phi_{|\bar\q} (\partial\bar\q), $$
then using Eq. \eqref{eq:DPhi_is_uPhi},
$$\partial\Lambda = 2\Big[\lp\pa \bar\q\rp \bar\q^\ast\Big]_\times \Lambda= 2\Big[\pa_{\text{rel}} \bar\q\Big]_\times \Lambda. $$
 Let $\mathbf{w}\in\R^3$. We compute
\begin{eqnarray*}
[(\nabla_{x,\textnormal{rel}} \bar\q)^t \mathbf{w}]_\times \Lambda &=& [((\mathbf{w}\cdot\nabla_{x}) \bar\q) \bar \q^\ast]_\times \Lambda \\
&=&\sum_{i=1,2,3} w_i \,\left[(\pa_{x_i} \bar\q) \bar \q^\ast\right]_\times \Lambda\\
&=& \frac{1}{2}\sum_{i=1,2,3} w_i \pa_{x_i} \Lambda \\
&=& \frac{1}{2}(\mathbf{w}\cdot\nabla_{x}) \Lambda,
\end{eqnarray*}
where we have used successively the definition of $\pa_\textnormal{rel}$; the fact that in components $\mathbf{w}=(w_1, w_2,w_3)$; and Eq. \eqref{equiv:partial}.
Recall that since $\pa_{x_i}\bar\q$ is orthogonal to $\bar \q$, the product $\pa_{\textnormal{rel},x_i}=(\pa_{x_i}\bar\q) \bar\q^\ast$ is purely imaginary and can be identified with a vector in $\R^3$, so all the above terms make sense.
Recall now the definition of $\mathscr{D}_x(\Lambda)$ in Eq. \eqref{def:D_x}: we have just proved Eq. \eqref{equiv:DLambda}.

We now use Eq. \eqref{equiv:DLambda} and the definitions of $\delta_x$ and $\rvec_x$ in Eq. \eqref{def:delta_and_r} to verify
\be
\delta_x(\Lambda) = \textnormal{Tr}\lp \mathscr{D}_x(\Lambda) \rp= 2\textnormal{Tr}\lp \nabla_{x,\textnormal{rel}} \bar\q \rp = 2 \nabla_{x,\textnormal{rel}} \cdot \bar\q,
\nonumber
\ee
and
\be
\mathbf{v}\times\rvec_x(\Lambda) = - [\rvec_x(\Lambda)]_\times\mathbf{v}=-\mathscr{D}_x(\Lambda)\mathbf{v}+\mathscr{D}_x(\Lambda)^t\mathbf{v}
=-2\lp\mathbf{v}\cdot\nabla_{x,\textnormal{rel}}\rp\bar\q+2\lp\nabla_{x,\textnormal{rel}} \bar\q\rp\cdot\mathbf{v},
 \nonumber
\ee
which concludes the result.

\end{proof}

\subsubsection{Interpretation in terms of a local vector $\mathbf{b}$.} In Ref. \cite{bodyattitude}, an interpretation in terms of a locally defined vector field $\mathbf{b}$ was proposed for the operators $\delta_x(\Lambda)$ and $\rvec_x(\Lambda)$. We summarize it here: let $(t_0,x_0)\in\R_+\times\R^3$ be fixed. When $\Lambda=\Lambda(t,x)$ is smooth enough, we can write 
\be \Lambda(t,x)=\exp\lp[2\mathbf{b}(t,x)]_\times \rp\Lambda(t_0,x_0),\label{def:b_Lambda}\ee
with $\mathbf{b}(t,x)$ a uniquely defined vector in $\R^3$, smooth around $(t_0,x_0)$ and with $\mathbf{b}(t_0,x_0)=0$. Then   
\be\label{int_delta_rvec}
\delta_x(\Lambda)(t_0,x_0) = 2\left.\nabla_x \cdot \mathbf{b}(x)\right|_{|(t,x)=(t_0,x_0)}\quad\text{ and }\quad
\rvec_x(\Lambda)(t_0,x_0) = 2 \left.\nabla_x \times \mathbf{b}(x)\right|_{|(t,x)=(t_0,x_0)},
\ee
where~$\nabla_x\times$ is the curl operator.

We propose a similar interpretation for our model. Let $(t_0,x_0)\in\R_+\times\R^3$ be fixed. We define $\mathbf{r}=\mathbf{r}(t,x)$ similarly as in Eq. \eqref{eq:r}:
\be \label{eq:r_tx}
\bar\q(t,x) = \mathbf{r}(t,x) \bar\q(t_0,x_0).
\ee
Since $\mathbf{r}\in\unitq$, its logarithm is a purely imaginary quaternion $\mathbf{b}=\mathbf{b}(t,x)$ with $\mathbf{b}(t_0,x_0)=0$. With these notation, we recast Eq. \eqref{eq:r_tx} into
$$ \bar\q(t,x) = \exp(\mathbf{b}(t,x)) \bar\q(t_0,x_0),$$
and differentiating with respect to any variable ($\pa\in\{\pa_t,\pa_{x_1},\pa_{x_2},\pa_{x_3}\}$), by definition of $\pa_{\textnormal{rel}}$,
\be\label{int_drel_b}
\pa_{\text{rel}} \bar\q (t_0,x_0)=\left.(\pa \exp(\mathbf{b}))\right|_{(t,x)=(t_0,x_0)}=\left.(\pa \mathbf{b})\right|_{(t,x)=(t_0,x_0)}.
\ee

Note that if $\Lambda$ and $\bar\q$ represent the same rotation, that is, if $\Lambda=\Phi(\bar\q)$, applying the morphism $\Phi$ to Eq. \eqref{eq:r_tx}, we end up with
$$ \Lambda (t,x) = \Phi(\exp(\mathbf{b}(t,x))) \Lambda(t_0,x_0).$$
We have that $\mathbf{b}=\theta \mathbf{n}/2$ in the Euler axis-angle representation (see Eqs. \eqref{eq:exponential form unitary quaternion}--\eqref{eq:rodrigues_formula}), where $\theta\in [0,2\pi]$ and $\mathbf{n}$ is a unitary vector in $\R^3$. The unitary quaternion $\exp(\mathbf{b})=\exp(\theta \mathbf{n}/2)$ represents the rotation of angle $\theta$ anti-clockwise around the axis $\mathbf{n}$, whose matrix representation is the corresponding matrix formulation, given by Rodrigues formula (Eq. \eqref{eq:rodrigues_formula}): 
$$\Phi(\exp(\mathbf{b}))=\exp(2[\mathbf{b}]_\times),$$
which implies that 
$$ \Lambda (t,x) = \exp(2[\mathbf{b}]_\times) \Lambda(t_0,x_0),$$
and we recover Eq. \eqref{def:b_Lambda}.

\begin{remark} The combination of Eqs. \eqref{int_delta_rvec} and \eqref{int_drel_b} gives an alternative proof of Eqs. \eqref{equiv:delta} and \eqref{equiv:rvec}.
\end{remark}

\subsubsection{Summary: comparison between quaternions, matrices, and $\mathbf{b}$}
We summarize the discussion of the two previous paragraphs in the
\begin{proposition}
\label{prop:table_equivalences}
Let $\rho=\rho(t,x)$ and $\bar\q=\bar\q(t,x)$ be two functions on $\R_+\times\R^3$ with values in $\R_+$ and $\unitq$ respectively. We define $\Lambda=\Phi(\bar\q)$ the matrix representation of the rotation represented by $\bar\q$. For any fixed $t_0\in\R_+$, $x_0\in\R^3$, we also define the vector field $\mathbf{b}^{t_0,x_0}=\mathbf{b}^{t_0,x_0}(t,x)$ as
$$\mathbf{b}^{t_0,x_0}(t,x) = \log [\bar\q(t,x)\bar \q(t_0,x_0)^\ast]. $$
Finally we define the velocity vector field
$$\mathbf{v}=\vezero(\bar\q)=\Lambda \vezero.$$

Then the following equivalence table holds:
\begin{center}
\def\arraystretch{1.5}
\begin{tabular}{|c|l|l|l|}
\hline
$i$ & Quaternion & Vector $\mathbf{b}$ locally at point $(t_0,x_0)$ & Orthonormal matrix \\
\hline
$1$ & $X_{\q,1}:=2\rho \partial_{t,\textnormal{rel}} \bar \q $ & $X_{\mathbf{b},1}^{t_0,x_0}:=2 \rho \pa_t \mathbf{b}^{t_0,x_0}$ & $X_{\Lambda,1}:=\rho \lp \partial_t\Lambda\rp \Lambda^t$ \\
\hline
$2$ & $X_{\q,2}:=2\rho (\vezero(\bar\q) \cdot \nabla_{x,\textnormal{rel}}) \bar \q$ & $X_{\mathbf{b},2}^{t_0,x_0}:=2 \rho \lp \mathbf{v}\cdot \nabla_x\rp \mathbf{b}^{t_0,x_0}$ & $X_{\Lambda,2}:=\rho \lp\big((\Lambda \vezero) \cdot \nabla_x\big)\Lambda \rp\Lambda^t$ \\
\hline
$3$ & $X_{\q,3}:=\vezero(\bar\q) \times \nabla_x \rho$ & $X_{\mathbf{b},3}^{t_0,x_0}:= \mathbf{v}\times \nabla_x \rho$ & $X_{\Lambda,3}:=\left[(\Lambda \vezero)\times \nabla_x \rho\right]_\times$ \\
\hline
$4$ & $X_{\q,4}:=2\rho \lp \nabla_{x,\textnormal{rel}}\bar \q\rp \vezero(\bar \q)$ & $X_{\mathbf{b},4}^{t_0,x_0}:=2\rho \lp\nabla_x \mathbf{b}^{t_0,x_0} \rp \, \mathbf{v}$ & $\begin{matrix} X_{\Lambda,4}:=\rho\,[ (\Lambda \vezero) \times \rvec_x(\Lambda)]_\times\\ + X_{\Lambda,2}\end{matrix}$ \\
\hline
$5$ & $X_{\q,5}:=2\rho \,(\nabla_{x,\textnormal{rel}}\cdot \bar \q) \vezero(\bar\q)$ & $X_{\mathbf{b},5}^{t_0,x_0}:=2\rho \lp\nabla_x  \cdot\mathbf{b}^{t_0,x_0} \rp \mathbf{v}$ & $X_{\Lambda,5}:=\left[ \rho\,\delta_x(\Lambda)\Lambda \vezero\right]_\times$ \\
\hline
\end{tabular}
\end{center}

The equivalence is to be read in the following sense: for $i=1\dots 5$, we have everywhere on $\R_+\times\R^3$,
\be \label{equiv_rep} X_{\q,i}({t_0,x_0}) = X_{\mathbf{b},i}^{t_0,x_0}({t_0,x_0}) \qquad \text{and} \qquad X_{\Lambda,i} = [X_{\q,i}]_\times.
\ee
\end{proposition}

As a consequence, it holds that
\begin{theorem}
\label{th:equivalence_macro_equations}
Let $\rho_0=\rho_0(x)\ge0$. Let $\bar\q_0=\bar\q_0(x)\in\unitq$ and $\Lambda_0=\Lambda_0(x)\in SO(3)$ represent the same rotation, i.e., $\Lambda_0(x)=\Phi(\bar\q_0(x))$ for all $x\in\R^3$. Then the system \eqref{sys:macro1}--\eqref{sys:macro3} and the system  \eqref{eq:macro_rho}--\eqref{eq:macro_lambda} are equivalent (in the sense that any solution $(\rho,\bar\q)$ of \eqref{sys:macro1}--\eqref{sys:macro3} gives a solution  $(\rho,\Lambda=\Phi(\bar\q))$ of the system \eqref{eq:macro_rho}--\eqref{eq:macro_lambda}).
\end{theorem}
\begin{proof}
We already checked that  the continuity equations  \eqref{sys:macro1} and \eqref{eq:macro_rho} are equivalent. Using the notations of Prop. \ref{prop:table_equivalences}, we recast, after multiplying by 2, Eq. \eqref{eq:q_relative_form} for $\bar\q$ (which is equivalent to Eq. \eqref{sys:macro3})  into
\be  
X_{\q,1} + c_2 X_{\q,2}  + 2\,c_3 X_{\q,3} + c_4 X_{\q,4} + c_4 X_{\q,5} = 0, \label{eq:aux_relative_quaternion}
\ee
and Eq. \eqref{eq:macro_lambda} for $\Lambda$ into
\be \label{eq:aux_relative_matrix}
 X_{\Lambda,1} + (\tilde c_2-\tilde c_4) X_{\Lambda,2}  + \tilde c_3 X_{\Lambda,3}+ \tilde c_4 X_{\Lambda,4}+ \tilde c_4 X_{\Lambda,5} = 0,
 \ee
where (see Ref. \cite{bodyattitude})
\beqar
\tilde c_3 &=& d,\\
\tilde c_2 &=&  \tfrac{1}{5}\langle 2+3\cos\theta\rangle_{\widetilde{m}(\theta)\sin^2(\theta/2)} ,\\
\tilde c_4&=& \tfrac{1}{5}\langle 1-\cos\theta\rangle_{\widetilde{m}(\theta)\sin^2(\theta/2)},
\eeqar
and where the notation
$\langle \cdot \rangle_{\widetilde{m}(\theta)\sin^2(\theta/2)}$ is given in Eq. \eqref{eq:definition_angles}.
The function~$\widetilde{m}:(0,\pi)\to(0,+\infty)$ is given by
\be \label{eq:weight_mtilde}
\widetilde{m}(\theta):= \sin^2\theta\, m(\theta)\, k(\theta),
\ee
where $m(\theta)=\exp\left(d^{-1}(\tfrac{1}{2}+\cos\theta)\right)$ is the same as in Eq. \eqref{eq:tilde_m} and $k$ is the solution of Eq. \eqref{eq:ode_GCI_matrices}.

To check that Eqs. \eqref{eq:aux_relative_quaternion}--\eqref{eq:aux_relative_matrix} are equivalent, it suffices to show the correspondence between the constants since the equivalence of the terms is already given by Eq. \eqref{equiv_rep}. Therefore, we are left to check that $\tilde c_2-\tilde c_4 = c_2$ and $\tilde c_4= c_4$.

Recall the values of the constants
\beqar
c_3&=&\frac{d}{2},\\
c_2&=&\frac{1}{5} \la 1+4\cos \theta \ra_{m(\theta) \sin^4(\theta/2) h(\cos(\theta/2)) \cos(\theta/2)},\\
c_4&=& \frac{1}{5}\la 1-\cos\theta \ra_{m(\theta)\sin^4(\theta/2) h(\cos(\theta/2)) \cos(\theta/2)}.
\eeqar
Using Prop. \ref{propo:h_and_k}, we have
\be k(\theta)=4\,\frac{h\lp\cos(\theta/2)\rp}{\cos(\theta/2)},\ee
so that
\beqar
\widetilde{m}(\theta)\sin^2(\theta/2) &=& \sin^2\theta\, m(\theta)\, k(\theta) \sin^2(\theta/2)\\
&=& 4\,\frac{h\lp\cos(\theta/2)\rp}{\cos(\theta/2)} \sin^2\theta\, m(\theta)\, \sin^2(\theta/2)\\
&=& 16\, h\lp\cos(\theta/2)\rp\,\cos(\theta/2)  m(\theta)\, \sin^4(\theta/2).
\eeqar
Therefore, we have that
\beqar
\langle \cdot \rangle_{\widetilde{m}(\theta)\sin^2(\theta/2)}
=
\langle \cdot \rangle_{m(\theta) \,\sin^4(\theta/2)\, h\lp\cos(\theta/2)\rp\,\cos(\theta/2) },
\eeqar
(notice that the constant $16$ is simplified),
which allows to conclude the equivalence of the constants, and hence, of the equations. 
\end{proof}

\section{Conclusion}

In the present work we have introduced a flocking model for body attitude coordination where the body attitude is described through rotations represented by unitary quaternions. The deliberate choice of representing rotations by unitary quaternions is based on their numerical efficiency in terms of memory usage and operation complexity. This will be key for future applications of this model. 
At the modelling level, we introduce an individual based model where agents try to coordinate their bodies attitudes with those of their neighbours. To express this we needed to define an appropriate `averaged' quaternion based on nematic alignment. This average is related to the Gennes $Q$-tensor that appears in liquid crystal theory. From the Individual Based Model we have derived the macroscopic equations (SOHQ) via the mean-field equations. We also show the equivalence between the SOHQ and the macroscopic equations (SOHB) of Ref. \cite{bodyattitude} where the body attitude is expressed through rotation matrices. However, we observe that the SOHQ is simpler to interpret than the equivalent SOHB. In particular, all the terms in the SOHQ are explicit. 
We have also seen that the dynamics of the SOHQ system are more complex than those of the SOH system (macroscopic equations corresponding to the Vicsek model). The body attitude coordination model presented here opens many questions and perspectives. We refer the reader to Ref. \cite[Conclusions and open questions]{bodyattitude} for an exposition.

One may wonder why we did not consider to translate directly into quaternions the results in Ref. \cite{bodyattitude} for rotation matrices. The answer is that, firstly, for the individual based model, it is not possible to obtain a direct translation, in the sense that we need to consider some particular modeling choices (like the average in Eq. \eqref{eq:definition average} and the relaxation in Eq. \eqref{eq:relaxation_term}) and check a posteriori the equivalence with the model in Ref. \cite{bodyattitude}. Secondly, the relation at the macroscopic level is not easy to obtain a priori. It is the macroscopic limit that gives us the necessary information and intuition to establish the link between both results.

In a future work, we will carry out simulations of the Individual Based Model and the SOHQ model; study the patterns that arise; and compare them with the ones of the Vicsek and SOH model. 

\paragraph{Acknowledgements} 

P.D. acknowledges support from the Royal Society and the Wolfson foundation through a Royal Society Wolfson Research Merit Award ref WM130048; the British “Engineering and Physical Research Council” under grants ref: EP/M006883/1 and EP/P013651/1; the National Science Foundation under NSF Grant RNMS11-07444 (KI-Net). P.D. is on leave from CNRS, Institut de Math\'ematiques de Toulouse, France.\\
A.F. and A.T. acknowledge support for the ANR projet “KIBORD”, ref: ANR-13-BS01-0004 funded by the French Ministry of Research.\\
S.M.A. was supported by the British “Engineering and Physical Research Council” under grant ref: EP/M006883/1.\\
A.T. was supported by the European Research Council under the European Union's Seventh Framework Programme (FP/2007-2013) / ERC Grant Agreement n. 279600.

\paragraph{Data statement}

No new data was generated in the course of this research

\paragraph{Conflict of interest} The authors declare that they have no conflict of interest.

\appendix

\section{Unitary quaternions: some properties}

\begin{proposition} 
\label{prop:} Let $Q$ be a symmetric $4\times4$ matrix. For the function $\q\in\unitq \mapsto (\q\cdot Q\q)$, we have
\be 
\label{eq:gradient Q}
\nabla_\q(\q\cdot Q\q)= 2P_{\q^\perp}(Q\q), 
\ee
where $\nabla_\q$ is the gradient in $\unitq$.
\end{proposition}
\begin{proof}
Consider a path in $\unitq$ parametrized by $\eps>0$, $\q=\q(\eps)$ where $\q(0)=\q$ and $\left.\frac{d}{d\eps}\q(\eps)\right|_{\eps=0}=\delta_\q\in T_\q$, then
\beqar
\partial_\q(\q\cdot Q\q) \cdot \delta_\q &=& \lim_{\eps\to 0}\frac{\q(\eps)\cdot Q(\eps)\q(\eps)-\q\cdot Q\q}{\eps}\\
&=& \lim_{\eps \to 0}\delta_{ \q} \cdot Q\q + \q\cdot Q\delta_{ \q} + \mathcal{O}(\eps)\\
&=& 2 \delta_{\q}\cdot Q\q,
\eeqar
from which we conclude the result.
\end{proof}

\begin{proposition} 
\label{prop:tangent space}
Let $\q\in\unitq$. The tangent space $T_\q$ at $\q$ in $\unitq$ corresponds to $\q^\perp$ (the orthogonal space to $\q$). Particularly, it holds that
\be \label{eq:tangent space} 
\q^\perp = \{ \vv\q,\, \mbox{ for } \vv\in \R^3\},
\ee
considering the abuse of notation explained in Rem. \ref{rem:abuse of notation}.
\end{proposition}
\begin{proof}
The fact that $T_\q=\q^\perp$ can be seen by identifying $\unitq$ with the unit sphere $\mathbb{S}^3$.
Since $\q$ is invertible, we have
$$ \mathbb{H} = \{ \p\q,\, \textnormal{ for } \p \in \mathbb{H}\},$$
and for any $\p \in \mathbb{H}$, we have
$$ \p\q \in \q^\perp \iff (\p\q) \cdot \q =0 \iff \Real (\p)=0 \iff \p=\Ima(\p)=\vv\in\R^3.$$

 \end{proof}

\begin{proposition}[Decomposition of the volume form in $\unitq$] \label{prop:volume_element}
Let $f=f(\q)$ be a function on $\unitq$. Recall the parametrization in \eqref{eq:exponential form unitary quaternion},
$$ \q =\cos \frac{\theta}{2} + \sin\frac{\theta}{2}\lp n_1 \vec{\imath}+n_2\vec{\jmath}+n_3\vec{k} \rp,$$
where $\nvec:=(n_1,n_2,n_3)$ is a unitary vector in $\R^3$ and $\theta\in[0,2\pi]$. Let
$$\bar f(\theta,\nvec)= f(\cos(\theta/2)+\sin(\theta/2)\nvec).$$
Then we have the following change of variable
$$\int_{\unitq} f(\q)\, \ud\q = \int^{2\pi}_0 \frac{\sin^2(\theta/2)}{2}\int_{\mathbb{S}^2}\bar f(\theta,\nvec)\, \ud\theta \ud\nvec,$$
where $\ud\q$ is the Lebesgue measure on the hypersphere $\unitq$ and $\ud\nvec$ is the Lebesgue measure on the sphere $\mathbb{S}^2$.
In particular, if $f(\q)=f(-\q)$, we have
$$\int_{\unitq} f(\q)\, \ud\q = \int^{\pi}_0 \sin^2(\theta/2) \int_{\mathbb{S}^2}\bar f(\theta,\nvec)\, \ud\theta \ud\nvec,$$
and if furthermore $\bar f(\theta,\nvec)=\bar f(\theta)$ is independent of $\nvec$, we have
$$\int_{\unitq} f(\q)\, \ud\q = 4\pi \,\int^{\pi}_0 \sin^2(\theta/2) \bar f(\theta)\, \ud\theta.$$
\end{proposition}

\begin{proof}
 We consider the following change of variables for $\q=(q_1, q_2, q_3, q_4)$ corresponding to the spherical coordinates on the 4-dimensional sphere:
 \beqar
  q_1 &=& \cos (\theta/2)\\
  q_2 &=& \sin (\theta/2) \cos \theta_2\\
  q_3 &=& \sin(\theta/2) \sin \theta_2 \cos \theta_3\\
  q_4 &=& \sin (\theta/2) \sin \theta_2 \sin \theta_3
 \eeqar
 for $\theta \in [0,2\pi]$, $\theta_2 \in [0,\pi]$, $\theta_3 \in [0,2\pi)$. Then we have that
\be \label{eq:aux_volume} 
 \ud\q = \frac{1}{2}\sin^2(\theta/2) \sin \theta_2 \ud\theta \ud\theta_2 \ud\theta_3
 \ee
 by computing the Jacobian of this change of variables.
 However $\nvec \in \mathbb{S}^2$ can be parametrized as
 $$\nvec= \left[\begin{array}{l}
 0\\
 \cos\theta_2\\
 \sin\theta_2 \cos \theta_3\\
 \sin\theta_2 \sin \theta_3
 \end{array}\right],
 $$
and $\ud\nvec= \sin\theta_2 \ud\theta_2\ud\theta_3$. Substituting this in Eq. \eqref{eq:aux_volume} we conclude the proposition.
 
\end{proof}

\section{Differential operators on $SO(3)$ and on $\unitq$}\label{annex:Phi}

The next three propositions explain the relation between the gradient, divergence and laplacian operators in $SO(3)$ and $\unitq$. 

\begin{proposition} [Comparison of the gradient operator]
\label{lem:equivalent_equations}
Consider a scalar function $g: SO(3) \to \mathbb{R}$ differentiable and define the function $f: \unitq \to \mathbb{R}$ as $f(\q)= g(\Phi(\q))$. It holds that 
\be 
\label{eq:equivalence_gradients1}
(\nabla_A g) (\Phi(\q)) = \frac 14 \textnormal{D}_\q\Phi (\nabla_\q f (\q)),
\ee
or equivalently, for any $\uu\in \R^3$,
\be 
\label{eq:equivalence_gradients2}
\langle \nabla_\q f(\q), \uu\q\rangle_{\unitq}= 2\langle \nabla_A g(\Phi(\q)), [\uu]_\times \Phi(\q)\rangle_{SO(3)},
\ee
where $\langle\cdot, \cdot \rangle$ indicates the dot product and the subindex associated indicates to which space it corresponds.

Particularly, consider the following Cauchy problems for some $\q_0\in \unitq$ and $A_0=\Phi(\q_0)$:
\beqarl
\frac{\ud\q}{\ud t}&=& \frac{1}{4}\nabla_\q(f(\q)), \quad \q(0)=\q_0, \label{eq:q}\\
\frac{\ud A}{\ud t}&=& \nabla_A (g(A)), \quad A(0)=A_0. \label{eq:A}
\eeqarl
If $\q=\q(t)$ is a solution of \eqref{eq:q} on some time interval $[0,T)$, then  
$A(t):=\Phi(\q(t))$ is a solution of \eqref{eq:A} on the same time interval $[0,T)$.
\end{proposition}

\begin{proof} To make the proof clearer we will use the notation $\la \cdot, \cdot \ra$ rather than the symbol '$\cdot$' to indicate the inner product (in the sense of matrices as well as in the sense of vectors and quaternions).
We first check that Eqs. \eqref{eq:equivalence_gradients1} and \eqref{eq:equivalence_gradients2} are equivalent: indeed, since $\textnormal{D}_\q \Phi (\nabla_\q f (\q))$ belongs to $T_{\Phi(\q)}= \{[\uu]_\times \Phi(\q),\,\uu\in\R^3\}$, Eq. \eqref{eq:equivalence_gradients1} can be rewritten as, for all $\uu\in\R^3$,
\be
\langle (\nabla_A g) (\Phi(\q)),  [\uu]_\times \Phi(\q)\rangle = \frac 14\langle \textnormal{D}_\q \Phi (\nabla_\q f (\q),  [\uu]_\times \Phi(\q)\rangle.
\ee
By Prop. \ref{prop:DPhi}, the right-hand-side is equal to
\beqarl
\frac 14\langle \textnormal{D}_\q \Phi (\nabla_\q f (\q),  [\uu]_\times \Phi(\q)\rangle &=& \frac 12\langle [(\nabla_\q f (\q))\q^\ast]_\times \Phi(\q),  [\uu]_\times \Phi(\q)\rangle\\
&=&\frac 12\langle [(\nabla_\q f (\q))\q^\ast]_\times,  [\uu]_\times \rangle\\
&=&\frac 12\langle (\nabla_\q f (\q))\q^\ast,  \uu \rangle\\
&=&\frac 12\langle \nabla_\q f (\q),  \uu \q \rangle,
\eeqarl
so that we recover Eq. \eqref{eq:equivalence_gradients2}.

\bigskip

We now prove Eq. \eqref{eq:equivalence_gradients2}: fix some $\q\in\unitq$, $\uu\in\R^3$ and let $\tilde \q=\tilde \q (s) \in \unitq$ be a differentiable path in $\unitq$ with
\beqar
\tilde \q(0)&=&\q,\\
\left. \frac{\ud}{\ud s}\tilde \q\right|_{s=0} &=& \uu\q.
\eeqar
We compute
\beqar
\langle \left.\nabla_\q f(\q)\right., \uu\q \rangle
&=& \left.\frac{\ud}{\ud s} f(\tilde \q (s))\right|_{s=0}\\
&=&  \left.\frac{\ud}{\ud s} g(\Phi(\tilde{\q}(s)))\right|_{s=0}\\
&=& \left\la (\nabla_A g)(\Phi(\q)), \left. \frac{\ud}{\ud s}\Phi(\tilde{\q}(s))\right|_{s=0} \right\ra  \\
&=& \left\la (\nabla_A g)(\Phi(\q)), \textnormal{D}_\q\Phi (\uu\q)\right\ra \\
&=& 2\la (\nabla_A g)(\Phi(\q)), [\uu]_\times \Phi(\q)\ra \, ,
\eeqar
where we used Prop. \ref{prop:DPhi} to compute $ \textnormal{D}_\q\Phi$. This proves Eq. \eqref{eq:equivalence_gradients2}.

\bigskip

Let $\q=\q(t)$ be a solution of \eqref{eq:q} on some time interval $(0,T)$, let $ A(t) := \Phi(\q(t))$ on $(0,T)$. For any $\uu\in\R^3$, we compute
\beqar
\left\la \frac{\ud A}{\ud t}, [\uu]_\times A(t) \right\ra &=& \left\la \left.\textnormal{D}_\q\Phi\right|_{\q(t)}\lp \frac{\ud\q}{\ud t}\rp, [\uu]_\times A(t) \right\ra\\
&=& \la \left.\textnormal{D}_\q\Phi\right|_{\q(t)} (\vv\q(t)), [\uu]_\times A(t) \ra \\
&=& 2\la [\vv]_\times A(t), [\uu]_\times A(t)\ra \\
&=& 2\la \vv, \uu\ra,
\eeqar
where we note $\vv=\vv(t):=(\ud \q /\ud t)\q^\ast(t) \in \R^3$. On the other hand, we compute thanks to Eq. \eqref{eq:equivalence_gradients2}
\beqar
\la \nabla_A g(A(t)), [\uu]_\times A(t)\ra \, &=& \frac{1}{2}\langle\nabla_\q f(\q(t)), \uu\q(t) \rangle\\
&=&2\langle \frac{\ud\q}{\ud t}, \uu\q(t) \rangle,\\
&=& 2\la \vv, \uu\rangle,
\eeqar
so that
\be \label{eq:definition_gradient}
 \left\la \frac{\ud A}{\ud t}, [\uu]_\times A(t) \right\rangle = \left\langle \nabla_A g(A(t)), [\uu]_\times A(t) \right\rangle, \qquad \mbox{ for all } \uu\in \mathbb{R}^3.
\ee
This concludes the proof.
\end{proof}

\begin{proposition}[Comparison of the divergence operator]
\label{lem:equivalence_divergence}
Let $G$ be a vector field tangent to $SO(3)$ and $H$ a vector field tangent to $\unitq $ such that
\be \label{eq:PTA_PTq}
 G (\Phi(\q)) = \textnormal{D}_\q \Phi \, (H(\q)), \qquad\mbox{for all } \q\in \unitq.
\ee
Then,
\be \label{eq:equivalence_divergence}
(\nabla_\q \cdot H)(\q)= (\nabla_A\cdot G)(\Phi(\q)), \qquad\mbox{for all } \q\in \unitq.
\ee
\end{proposition}

\begin{proof}
Consider functions $f,g$ with $f(\q)=g(\Phi(\q))$ and define $\uu:\unitq\to \R^3$ by
$$\uu(\q)= 2 H(\q) \q^\ast,\qquad \text{for all }\q\in\unitq.$$
By Eq. \eqref{eq:PTA_PTq} and Prop. \ref{prop:DPhi}, we have
$$G(\Phi(\q)) = [\uu(\q)]_\times \Phi(\q),\qquad \text{for all }\q\in\unitq.$$
Then, we can compute
\beqar
\int_{\unitq} f(\q) \, \nabla_\q \cdot H(\q)\, \ud\q &=& -\int_{\unitq}\langle \nabla_\q f(\q) , H(\q)\rangle_{\unitq}\, \ud\q\\
&=& -\frac 12 \int_{\unitq}\langle \nabla_\q f(\q) , \uu(\q) \q \rangle_{\unitq}\, \ud\q\\
&=& - \int_{\unitq}\langle \nabla_A g(\Phi(\q) , [\uu(\q)]_\times \Phi(\q)\rangle_{SO(3)}\, \ud\q\\
&=& - \int_{\unitq}\langle \nabla_A g(\Phi(\q) , G( \Phi(\q))\rangle_{SO(3)}\, \ud\q\\
&=& -2 \pi^2\int_{\unitq}\langle \nabla_A g(A) , G(A)\rangle_{SO(3)}\, \ud A\\
&=& 2 \pi^2 \int_{\unitq}  g(A)\, (\nabla_A\cdot G)(A)\, \ud A\\
&=&\int_{\unitq}  g(\Phi(\q))\, (\nabla_A\cdot G)(\Phi(\q))\, \ud\q\\
&=&\int_{\unitq}  f(\q)\, (\nabla_A\cdot G)(\Phi(\q))\, \ud\q,
\eeqar
where we have used integration by parts and Eq. \eqref{eq:equivalence_gradients2}. We conclude that
$$\int_{\unitq}\left[ \nabla_\q \cdot H(\q)- (\nabla_A\cdot G)(\Phi(\q)) \right]\, f(\q) \, \ud\q=0,$$
for all $f$ such that $f(\q)=f(-\q)$. This implies that Eq. \eqref{eq:equivalence_divergence} holds.
\end{proof}

\begin{proposition}[Comparison of the laplacians]
\label{prop:laplacians}
Consider a scalar function $g: SO(3) \to \mathbb{R}$ twice differentiable and define the function $f: \unitq \to \mathbb{R}$ as $f(\q)= g(\Phi(\q))$. It holds that
$$(\Delta_A g) (\Phi(\q))= \frac{1}{4}\Delta_\q f (\q).$$
\end{proposition}

\begin{proof}
By Prop. \ref{lem:equivalent_equations}, we have that Eq. \eqref{eq:equivalence_gradients1} is true, so that \eqref{eq:PTA_PTq} is true for $G:= \nabla_A g$ and $H:=\nabla_\q f/4$. Applying Prop. \ref{lem:equivalence_divergence} gives the result.
\end{proof}

\section{Equivalence of the GCI equations}

\begin{proposition}\label{propo:h_and_k}
Let $h$ be a solution of Eq. \eqref{eq:ode_h}. Then the function
\be\label{def:k} k(\theta):=4\,\frac{h\lp\cos(\theta/2)\rp}{\cos(\theta/2)},\ee
is a solution of the following equation
\be\label{eq:ode_GCI_matrices}
\frac{1}{\sin^2(\theta/2)}\pa_\theta \lp \sin^2(\theta/2) m(\theta) \pa_\theta \lp \sin \theta k(\theta) \rp\rp - \frac{m(\theta) \sin\theta }{2\sin^2(\theta/2)} k(\theta) =\sin(\theta) m(\theta).
\ee
Remark: For $P$ antisymmetric matrix and $\Lambda\in SO(3)$,
$$\psi(A)= P\cdot (\Lambda^t A)\, \bar k(\Lambda\cdot A), \qquad A\in SO(3),$$
is a generalised collision invariant in the body attitude coordination model based on rotation matrices from Ref. \cite{bodyattitude}. In this case $\bar k(\Lambda \cdot A)= \bar k\lp \frac{1}{2}+\cos(\theta)\rp =:k(\theta)$.
\end{proposition}

\begin{proof}
For convenience we introduce the notation
$$c:=\cos(\theta/2), \qquad s:=\sin(\theta/2),$$
so that
$$c^2+s^2=1, \qquad \pa_\theta c = -s/2, \qquad \pa_\theta s = c/2, \qquad \cos \theta = c^2-s^2,\qquad \sin\theta = 2cs.$$
We write
$$ k(\theta)=4\frac{h\lp c\rp}{c}.$$
We use the equivalent Eq. \eqref{eq:ode_h_alt} for $h$ and rewrite it in $r=c$ as
\be\label{eq:h_of_c}
\lp \frac{-4}{d}c^2-3\rp h+ \lp\frac{4}{d}s^2-5 \rp\,c\, h' + s^2 h''=c.
\ee
Finally by definition of $m$ in Eq. \eqref{eq:tilde_m},
$$m(\theta) = \exp\lp \frac{1}{d} (\frac{1}{2}+ \cos \theta) \rp, \qquad \pa_\theta m(\theta) = - \frac{2}{d} cs m(\theta).$$

We want to check that $k$ defined by Eq. \eqref{def:k} is a solution of Eq. \eqref{eq:ode_GCI_matrices}. This is equivalent to showing that
$$D(\theta):=\pa_\theta[s^2 m(\theta) \pa_\theta (\sin\theta k(\theta))] - m(\theta) c s k(\theta) = 2 cs^3 m(\theta).$$
We first compute
\begin{eqnarray*}
\pa_\theta \lp \sin(\theta) k(\theta) \rp &=& \cos\theta k(\theta) + \sin\theta k'(\theta) \\
&=&4 \cos\theta \frac{h\lp c\rp}{c} +4 \sin\theta\frac{s}{2 c^2} h(c) + 4\frac{\sin\theta}{c}\frac{-s}{2} h'(c)\\
&=&4 (c^2-s^2) \frac{h\lp c\rp}{c} +4 cs \frac{s}{c^2} h(c) - 4 s^2 h'(c)\\
&=&4 c h\lp c\rp  - 4 s^2 h'(c).
\end{eqnarray*}
Then, inserting this expression into $D(\theta)$, we have that
\begin{eqnarray*}
D(\theta)&=& 4 \pa_\theta[s^2 m(\theta) \lp c h\lp c\rp -s^2 h'(c)\rp] - 4 m(\theta) s h(\theta) \\
&=&4 \pa_\theta[s^2 m(\theta) c] h\lp c\rp - 2 c s^3 m(\theta) h'\lp c\rp\\
&&- 4 \pa_\theta[s^4 m(\theta) ]h'(c) - 4 s^4 m(\theta) \frac{-s}{2} h''(c) - 4 m(\theta) s h(\theta)\\
&=&4 \left[sc m(\theta) c-\frac{1}{2}s^3 m(\theta)-\frac{2}{d}c^2s^3 m(\theta) - m(\theta) s\right] h\lp c\rp\\
&&+\left[- 2 c s^3 m(\theta) - 4 [2 c s^3 m(\theta) - \frac{2}{d} cs s^4 m(\theta) ]  \right] h'\lp c\rp\\
&& + 2 s^5 m(\theta) h''(c)\\
&=&2 s^3 m(\theta)  \left\{ 2 \left[ c^2/s^2-\frac{1}{2} -\frac{2}{d}c^2 - 1/s^2 \right] h\lp c\rp \right.\\
&&\left.\left[  - c - 2 [2 c - \frac{2}{d} cs^2 ]   \right] h'\lp c\rp + s^2 h''(c) \right\}\\
&=&2 s^3 m(\theta)  \left\{ (-3 -\frac{4}{d} c^2) h\lp c\rp + (- 5 + \frac{4}{d} s^2) ch'\lp c\rp + s^2 h''(c) \right\}.
\end{eqnarray*}
Using Eq. \eqref{eq:h_of_c}, this last expression is equal to
\begin{eqnarray}
D(\theta) &=&2 s^3 m(\theta) c,
\end{eqnarray}
which concludes the proof.
\end{proof}

\bibliographystyle{abbrv}
\bibliography{biblio}

\end{document}